\documentclass[lettersize,journal]{IEEEtran}

\usepackage{amsmath,amsfonts}
\usepackage{algorithmic}
\usepackage{algorithm}
\usepackage{array}
\usepackage{textcomp}
\usepackage{stfloats}
\usepackage{url}
\usepackage{verbatim}
\usepackage{graphicx}
\usepackage{cite}

\usepackage{mathrsfs}
\usepackage{mathtools}
\usepackage{bm}
\usepackage{bbm}

\usepackage{amsthm}
\usepackage{amssymb}
\usepackage{graphicx}
\usepackage{hyperref}

\usepackage{lipsum}

\usepackage{tabularx,booktabs}
\newcolumntype{C}[1]{>{\centering\arraybackslash}p{#1}}
\usepackage{array,multirow}

\makeatletter

\usepackage{todonotes}

\floatstyle{ruled}
\newfloat{algorithm}{tbp}{loa}
\providecommand{\algorithmname}{Algorithm}
\floatname{algorithm}{\protect\algorithmname}

\theoremstyle{plain}
\newtheorem{thm}{\protect\theoremname}
\newtheorem{prop}[thm]{\protect\propositionname}
\newtheorem{lem}[thm]{\protect\lemmaname}

\usepackage{thmtools}

\declaretheoremstyle[
  bodyfont=\normalfont\itshape,
  headformat=\NAME\NUMBER  
]{nospacetheorem}

\usepackage{filecontents} 
\theoremstyle{definition}
\newtheorem{definition}{Definition}

\@ifundefined{showcaptionsetup}{}{%
 \PassOptionsToPackage{caption=false}{subfig}}
\usepackage{subfig}
\makeatother

\usepackage{babel}

\providecommand{\lemmaname}{Lemma}
\providecommand{\propositionname}{Proposition}
\providecommand{\theoremname}{Theorem}

\usepackage{cellspace}
\include{IEEEtran}

\begin{document}

\title{Dictionary Learning Using Rank-One Atomic Decomposition (ROAD)}

\makeatletter
\newcommand{\linebreakand}{%
  \end{@IEEEauthorhalign}
  \hfill\mbox{}\par
  \mbox{}\hfill\begin{@IEEEauthorhalign}
}
\makeatother

\author{
  \IEEEauthorblockN{Cheng Cheng and Wei Dai}\\
  \IEEEauthorblockA{\textit{Dept. of Electrical and Electronic Engineering, Imperial College London, UK}}
}

\maketitle
\begin{abstract}
Dictionary learning aims at seeking a dictionary under which the training data can be sparsely represented. Methods in the literature typically formulate the dictionary learning problem as an optimization w.r.t. two variables, i.e., dictionary and sparse coefficients, and solve it by alternating between two stages: sparse coding and dictionary update. The key contribution of this work is a Rank-One Atomic Decomposition (ROAD) formulation where dictionary learning is cast as an optimization w.r.t. a single variable which is a set of rank one matrices. The resulting algorithm is hence single-stage. Compared with two-stage algorithms, ROAD minimizes the sparsity of the coefficients whilst keeping the data consistency constraint throughout the whole learning process. An alternating direction method of multipliers (ADMM) is derived to solve the optimization problem and the lower bound of the penalty parameter is computed to guarantees a global convergence despite non-convexity of the optimization formulation. From practical point of view, ROAD reduces the number of tuning parameters required in other benchmark algorithms. Numerical tests
demonstrate that ROAD outperforms other benchmark algorithms for both synthetic data and real data, especially when the number of training samples is small. 
\end{abstract}


\section{Introduction}

\IEEEPARstart{M}{assive} interests have been attracted in sparse signal representation and its wide range of applications, including signal denoising \cite{elad2006image,dabov2007image}, restoration \cite{mairal2008sparse,dong2013nonlocally},
source separation \cite{li2006underdetermined,abolghasemi2012blind},
classification \cite{tosic2011dictionary,huang2007sparse}, recognition \cite{wright2009robust,wright2010sparse,zhang2011sparse}, image super-resolution \cite{yang2010image,dong2011image} to name a few. The basic idea of the sparse signal representation is that a natural signal can be transformed into a sparse signal under a certain dictionary/basis. Compared with choosing a basis set from analytical dictionaries such as discrete cosine transform (DCT) \cite{ahmed1974discrete}, short time Fourier transform (STFT) \cite{allen1977unified}, wavelets \cite{stephane1999wavelet}, curvelets \cite{candes2000curvelets}, etc., a dictionary trained from the data itself can attain sparser representations \cite{olshausen1996emergence}. Therefore, the relative research topic named dictionary learning has drawn enormous efforts to find a more efficient algorithm for training the dictionary. More mathematically, dictionary learning is a bilinear inverse problem where a bunch of training signals will be decomposed as the product of a dictionary and the corresponding sparse coefficients after some learning procedures. 

As dictionary learning is an optimization problem w.r.t. two variable, i.e., dictionary and sparse coefficients, a typical algorithm is an iterative process alternating between two stages: sparse coding and dictionary update \cite{engan1999method,aharon2006k,engan2007family,skretting2010recursive,dai2012simultaneous,yu2019bilinear}. The general principle of solving this bilinear inverse problem is to fix one variable and optimize the other. Hence, in the stage of sparse coding, the purpose is to find the sparse coefficients based on a given dictionary. 
This optimization problem can be solved using three different strategies, i.e., greedy algorithms, $\ell_1$-norm optimization and $\ell_p$ quasi-norm optimization.
Greedy algorithms consists of matching pursuit (MP) \cite{mallat1993matching}, orthogonal matching
pursuit (OMP) \cite{pati1993orthogonal,tropp2007signal}, subspace
pursuit (SP) \cite{dai2009subspace}, CoSaMP \cite{needell2009cosamp}, etc., that sequentially select the support set from the sparse coefficients. The second strategy known as basis pursuit (BP) \cite{chen2001atomic} convexifies the problem by using a surrogate $\ell_1$ norm penalty to promote sparsity. A variant of this convex problem is to reformulate it to unconstrained version, or namely lasso \cite{tibshirani1996regression}. Among all the approaches for addressing the lasso problem, a class of iterative shrinkage-thresholding algorithms (ISTA) \cite{daubechies2004iterative,hale2007fixed,beck2009fast} are widely adopted. Another category of methods replace $\ell_1$-norm with non-convex $\ell_p$ quasi-norm (0<p<1) to acquire better performance \cite{chartrand2007exact,chartrand2008restricted}, and $\ell_{1/2}$ quasi-norm has been especially investigated in \cite{xu2012l_}. 

The other stage dictionary update aims to refine the dictionary using the sparse coefficients obtained from the previous stage.
In this stage, columns of dictionary, or namely dictionary atoms, are updated either simultaneously \cite{engan1999method} or sequentially  \cite{aharon2006k,dai2012simultaneous,yu2019bilinear,seghouane2015sequential,seghouane2018consistent}.
Method of optimal directions (MOD) \cite{engan1999method} is one of the earliest two-stage methods, where the whole dictionary is updated in one step. In the dictionary update stage of MOD, whole sparse coefficient matrix is fixed and then the problem is formulated as a least squares problem.
In many other methods including K-SVD \cite{aharon2006k}, SimCO \cite{dai2012simultaneous} and BLOTLESS \cite{yu2019bilinear}, only the sparsity pattern (the positions of non-zeros) of sparse coefficients is preserved, and both the dictionary and the sparse coefficients are updated. However, these methods update only one atom or a block of the dictionary and the corresponding elements in sparse coefficients at a time, and then sequentially update the whole dictionary. 
Specifically, K-SVD fixes all but one atom and the corresponding row of sparse coefficients, and obtains their difference to the input signal. Only the elements of the residual at the sparsity pattern is considered, and the dictionary atom and the corresponding sparse coefficients is updated by using singular value decomposition (SVD).
SimCO updates multiple dictionary atoms and the corresponding sparse coefficients by viewing the coefficients as a function of the dictionary and performing a gradient descent w.r.t. dictionary.
BLOTLESS recasts the dictionary update as a total least squares problem, and updates the blocks of the dictionary and the corresponding elements of sparse coefficients sequentially.
Another category of sequentially updating methods including \cite{seghouane2015sequential,seghouane2018consistent}, computes the rank-one matrix approximation to the residual defined in K-SVD algorithm before sparse coding and dictionary learning stages. Then only one atom and the corresponding row of sparse coefficients are updated alternatively in two stages. Here a penalty of $\ell_1$ norm is applied to promote the sparsity, and hence the residual can be directly used instead of only preserving the sparsity pattern. 

However, all the aforementioned two-stage algorithms have the same issue. The performance of the two stages are coupled together and the optimal tuning of one stage may not lead to the optimal performance of the overall dictionary learning. The authors of BLOTLESS \cite{yu2019bilinear} found that a well-tuned dictionary update may result in poor performance of dictionary learning because the well-tuned dictionary update based on a poorly estimated sparsity pattern may enforce the optimization procedure to converge to the local minimum defined by that particular sparsity pattern. Furthermore, the two-stage alternating process makes the analysis very challenging. Few convergence or performance guarantees have been obtained in the literature for the general dictionary learning problem. 

In this paper, a novel dictionary learning algorithm that uses rank-one atomic decomposition (ROAD) is proposed. 
The key novelty in ROAD is to formulate dictionary learning as an optimization problem involving only one unknown variable, i.e., a set of rank-one matrices. 
Specifically, dictionary learning is cast as representing training data as the sum of  rank-one matrices, each with only a few non-zero columns. 
With this formulation, the two-stage optimization procedure in the literature is replaced with a single-stage process. 
Then the popular alternating direction method of multipliers (ADMM) is adapted to solve the ROAD formulation. 
Note that ROAD involves a constrained optimization with non-smooth objective function and a non-convex constraint (the set of rank-one matrices is non-convex). Nevertheless, motivated by the recent advance in optimization theory \cite{wang2019global}, we are able to show that the ADMM solver of ROAD enjoys a global convergence guarantee.

Our main contributions are as follows.
\begin{itemize}
\item ROAD is the first in the true sense of single-stage dictionary learning algorithm. Compared with two-stage algorithms, it minimizes the sparsity level of the coefficients whilst keeping the global data consistency constraint throughout the whole learning process. The resulting process cannot be hence trapped in a singular point, which two-stage algorithms may converge to.
\item ROAD reduces the burden of parameter tuning. In the sparse coding stage of benchmark algorithms, one typically needs by trial-and-error to choose either the maximum sparsity level for greedy algorithms or a regularization constant for a Lasso type of formulation. By comparison, there is no parameter to tune in ROAD in generating all the simulations in this paper.
\item We write both the inexact ADMM and exact ADMM formulations to solve ROAD. 
Although, in the literature, ROAD formulation does not satisefies any 
Inspired by the recent work done by Wang et. al. \cite{wang2019global}, we 

we derive the lower bound of the penalty parameter of the augmented Lagrangian.
Despite the non-smoothness of the objective function and the non-convexity of the constraint, we can prove that the ADMM solver of ROAD has a global convergence guarantee.  
\item Numerical performance in both synthetic and real data tests of ROAD with other state-of-the-art dictionary learning algorithms including MOD, K-SVD and BLOTLESS are compared. All simulations demonstrate that ROAD has the capability of learning more accurate dictionaries, and less training samples are needed compared to the other benchmark algorithms.  
For noiseless cases in synthetic data tests, ROAD  is  the  only  algorithm  that  has  no visible error floor while all other algorithms suffer from non-negligible error floors. For the synthetic tests with noise, ROAD is more robust to the noise than the other benchmark algorithms. In real data test, the performance improvement of ROAD is demonstrated using examples of single image super-resolution.  
\end{itemize}

The simulation results in Figure \ref{fig:1} are presented in our conference version \cite{cheng2020dictionary}. The main differences are that we refine our algorithm designs for both noise-free and noisy cases, modify the theoretical analysis of global convergence guarantee and specify the lower bound of penalty parameter $\rho$.

This paper is organized as follows. Section \ref{sec:background} briefly reviews state-of-art dictionary learning methods. In section \ref{sec:ROAD-formulation}, we introduce how the formulations of ROAD are derived, and explain why we should insist rank-one matrices constraint instead of its convex surrogate nuclear norm. In Section \ref{sec:ADMM-solver}, we adapt ADMM technique to solve ROAD. We also compute the lower bound of the penalty parameter of augmented Lagrangian and prove the global convergence guarantee of ADMM procedures.Results of numerical tests are presented and discussed in Section \ref{sec:numerical-tests}. Finally, this paper is concluded in the Section \ref{sec:conclusion}.

\subsection{Notation} \label{subsec:notation}

In this paper, 
$\Vert\cdot\Vert_{2}$ denotes the $\ell_{2}$ norm, and $\Vert\cdot\Vert_{F}$ represents the Frobenius norm. For a matrix $\bm{X}$, $\bm{X}_{i,:}$ and $\bm{X}_{:,j}$ stand for the $i$-th row and the $j$-th column of $\bm{X}$ respectively. For a positive integer $n$, $[n]$ represents a set $\{1,2,\cdots,n\}$. The symbols $\bm{I}$, $\bm{1}$ and $\bm{0}$ refer to the identity matrix, a matrix where all the entries are $1$, and a matrix full of zeros, respectively.

\section{Background \label{sec:background}}

The goal of dictionary learning is to seek a dictionary that can sparsely represent the training data. Let $\bm{Y}\in\mathbb{R}^{M \times N}$, where $M\in\mathbb{N}$ is the dimension of training vectors and $N\in\mathbb{N}$ denotes the number of training vectors. Then dictionary learning can be written as an optimization problem
\begin{equation} \label{eq:dict-learning}
\underset{\bm{D},\bm{X}}{\min}\; \sum_n \left\Vert \bm{X}_{:,n}\right\Vert _{0}\;{\rm subject\;to}\;\bm{Y} \approx \bm{D}\bm{X},
\end{equation} 
where $\bm{D}\in\mathbb{R}^{M\times K}$ denotes the unknown dictionary, and $\bm{X}\in \mathbb{R}^{K\times N}$ are the sparse representation coefficients, $\bm{X}_{:,n}$ is the $n$-th column of the matrix $\bm{X}$, and $\Vert \cdot \Vert_0$ is the $\ell_0$ pseudo-norm counting the number of non-zeros. The constraint $\bm{Y} \approx \bm{D}\bm{X}$ can be rewritten as $\Vert \bm{Y}-\bm{D}\bm{X}\Vert_F \le \epsilon$ when the noise energy in the training data can be roughly estimated, where $\Vert \cdot \Vert_F$ denotes the Frobenius norm and $\epsilon>0$ is a constant chosen based on the noise energy. In dictionary learning problems, it is typical that $M<K$, i.e., the dictionary is over-complete. 

The optimization problem \eqref{eq:dict-learning} is non-convex due to the non-convexity of both the objective function and the constraint set. 
To make dictionary learning feasible, in the literature relaxation and/or extra constraints are imposed and suboptimal algorithms are designed \cite{engan1999method,aharon2006k,dai2012simultaneous,yu2019bilinear,seghouane2015sequential,seghouane2018consistent}. 
Note the scaling ambiguity that $\bm{D}_{:,k} \bm{X}_{k,:} = (a \bm{D}_{:,k}) (\frac{k}{a} \bm{X}_{k,:} )$. It is common to assume unit $\ell_2$-norm of columns of $\bm{D}$. 

A popular approach is to assume that the sparse representation of each training vector in $\bm{Y}$ has at most $S$ many non-zeros, where $S\in\mathbb{N}$ is a pre-defined constant typically carefully chosen by trial-and-error and typically $S \ll M$. The optimization problem then becomes  
\begin{align}
\underset{\bm{D},\bm{X}}{\min}\; & \left\Vert \bm{Y}-\bm{D}\bm{X}\right\Vert _{2}^{2}\nonumber \\
{\rm s.t.}\; & \Vert\bm{D}_{:,k}\Vert_{2}=1,\;\Vert\bm{X}_{:,j}\Vert_{0}\le S,\;\forall j\in[N],\forall k\in[K].\label{eq:DL-sparsity-k}
\end{align}

Problem \eqref{eq:DL-sparsity-k} is typically solved by iterative algorithms that alternate between two stages: sparse coding and dictionary update. In the sparse coding stage, one fixes the dictionary $\bm{D}$ and updates the coefficients $\bm{X}$: 
\begin{equation}
\underset{\bm{X}_{:,n}}{\min}\;\Vert\bm{Y}_{:,n}-\bm{D}\bm{X}_{:,n}\Vert_{2}^{2},\;{\rm s.t.}\;\Vert\bm{X}_{:,n}\Vert_{0}\le S,\;\forall n\in [N]. \label{eq:sparse-coding}
\end{equation}
Though the problem \eqref{eq:sparse-coding} is not convex, it can be solved by many pursuit algorithms \cite{chen2001atomic,mallat1993matching,pati1993orthogonal,tropp2007signal,dai2009subspace,needell2009cosamp}. 

There are multiple approaches to formulate and solve the dictionary update problem. In MOD method \cite{engan1999method}, the sparse coefficient matrix $\bm{X}$ is fixed and dictionary update is formulated as a least squares problem as
\begin{equation}
\underset{\bm{D}}{\min}\;\Vert\bm{Y}-\bm{D}\bm{X}\Vert_{F}^{2}.    
\end{equation}
In many other methods including K-SVD \cite{aharon2006k}, SimCO \cite{dai2012simultaneous}, and BLOTLESS \cite{yu2019bilinear}, only the sparsity pattern of $\bm{X}$ is fixed, and both the dictionary $\bm{D}$ and the sparse coefficients $\bm{X}$ are updated. 
In K-SVD algorithm, only one column of $\bm{D}$ and the corresponding row of $\bm{X}$ are updated at a time. To update $\bm{D}_{:,j}$ and $\bm{X}_{j,:}$, define a residual $\bm{E}_{j}=\bm{Y}-\sum_{k=1,k\neq j}^{K}\bm{D}_{:,k}\bm{X}_{k,:}$, and denote $\omega_{j}$ as the sparsity pattern of $j$-th row of $\bm{X}$. Then the dictionary update stage of K-SVD can be formulated as
\begin{equation}
\underset{\bm{D}_{:,j},\;\bm{X}_{j,:}}{\min}\;\Vert\bm{E}_{:,\Omega_{j}}-\bm{D}_{:,j}\bm{X}_{j,\Omega_{j}}\Vert_{F}^{2},
\end{equation}
the solution of which is simply taking the largest left and right singular vectors of the matrix $\bm{E}_{:,\Omega_{j}}$.

SimCO updates the whole dictionary $\bm{D}$ and the whole sparse coefficient matrix $\bm{X}$ by viewing $\bm{X}$ as a function of $\bm{D}$ and performing a gradient descent w.r.t. $\bm{D}$. By comparison, BLOTLESS is inspired by the work done by Ling \cite{ling2018self}, and updates a block of the dictionary and the corresponding sparse coefficient using a total least squares approach. To ensure the invertibility of the updated dictionary block, it requires that the number of dictionary items in the updated block is at most $m$.

Another approach to solve the dictionary learning problem \eqref{eq:dict-learning} is to replace the non-convex objective function in \eqref{eq:dict-learning} with the sparsity promoting $\ell_1$-norm. One has 
\begin{equation}\label{eq:DL-l1}
    \underset{\bm{D},\bm{X}}{\min} \sum_{k,n} |X_{k,n}| \; s.t. \; \bm{Y}\approx \bm{D}\bm{X},\;\Vert \bm{D}_{:,k} \Vert_2 = 1, \; \forall k\in [K].
\end{equation}
Methods including \cite{seghouane2015sequential,seghouane2018consistent} also adopt the residual $\bm{E}_{j}$ defined in K-SVD, and use the objective
\begin{align}
& \{\bm{D}_{:,j},\bm{X}_{j,:}\}=  \underset{\bm{D}_{:,j},\bm{X}_{j,:}}{\arg\min}\Vert\bm{E}_{j}-\bm{D}_{:,j}\bm{X}_{j,:}\Vert_{F}^{2}+\lambda\Vert\bm{X}_{j,:}\Vert_{1}, \nonumber\\
& {\rm s.t.}\;\Vert \bm{D}_{:,k} \Vert_2 = 1,       
\end{align}
Then the solution are given by
\begin{align}
& \bm{D}_{:,j}=\frac{\bm{E}_{j}\bm{X}_{j,:}^{T}}{\Vert\bm{E}_{j}\bm{X}_{j,:}^{T}\Vert_{2}},\nonumber\\
& \bm{X}_{j,:}={\rm sgn}(\bm{D}_{:,j}^{T}\bm{E}_{j})\circ\left(\vert\bm{D}_{:,j}^{T}\bm{E}_{j}\vert-\frac{\lambda}{2}\bm{1}_{1\times N}\right)_{+},
\end{align}
where the symbols ${\rm sgn}$, $\circ$ and $(x)_{+}$ refer to sign function, ${\rm max}(0,x)$ and Hadamard product respectively.

\section{Dictionary Learning Via ROAD \label{sec:ROAD-formulation}}

This section presents our Rank One Atomic Decomposition (ROAD) formulation for dictionary learning. The key difference from benchmark algorithms in the literature is to avoid alternating optimization between two variables $\bm{D}$ and $\bm{X}$.  

ROAD formulation is based on the following two facts. 
\begin{enumerate}
  \item Define \begin{align}
    & \bm{Z}_k := \bm{D}_{:,k}\bm{X}_{k,:}. \label{eq:def-Z}
  \end{align}
  It is clear that $\bm{Z}_k$ is a matrix of rank one. 

  \item The sparsity of coefficients $\bm{X}_{k,:}$ is directly translated to column sparsity of the matrix $\bm{Z}_k$. That is, 
  \begin{align*}
    & \left( \bm{Z}_k \right)_{:,n} \begin{cases}
      = \bm{0} & \text{if}~ \bm{X}_{k,n} = 0, \\
      \ne \bm{0} & \text{if}~ \bm{X}_{k,n} \ne 0,
    \end{cases}
  \end{align*}
  assuming a non-trivial dictionary item $\bm{D}_{:,k}$. 
\end{enumerate}

Based on the above facts, ROAD formulation for dictionary learning is given as follows. Define the set of rank-one matrices of the proper size 
\begin{equation}  \label{eq:rank-one-set}
  \mathcal{R}_1=\left\{ \bm{Z}\in\mathbb{R}^{M\times N}:\;{\rm rank}\left(\bm{Z}\right)\leq 1\right\}.
\end{equation}
Define $\ell_{2,1}$ norm for a matrix as 
\begin{equation} \label{eq:l-21-norm}
  \left\Vert \bm{Z} \right\Vert _{2,1} :=  \sum_{n} \Vert(\bm{Z})_{:,n}\Vert_{2},  
\end{equation}
which promotes column sparsity. Then dictionary learning can be formulated as
\begin{align}
  \min_{\bm{Z}_{k}}  ~ 
  & \sum_{k}\left\Vert \bm{Z}_{k}\right\Vert _{2,1}
  \label{eq:ROAD-formulation} \\
  \text{s.t.} ~ 
  & \Vert \bm{Y} - \sum_{k}\bm{Z}_{k} \Vert_F \le \epsilon, 
  ~ \bm{Z}_{k}\in\mathcal{R}_1,\;\forall k\in\left[K\right],
  \nonumber
\end{align}
where $\epsilon \ge 0 $ is a pre-defined approximation error bound, 
or equivalently the seminal Lasso form 
\begin{align}
  \min_{\bm{Z}_{k}}  ~ 
  & \sum_{k}\left\Vert \bm{Z}_{k}\right\Vert _{2,1}
  + \frac{\beta}{2} \Vert \bm{Y} - \sum_{k}\bm{Z}_{k} \Vert_F^2
  \label{eq:ROAD-Lasso-formulation} \\
  \text{s.t.} ~ 
  & \bm{Z}_{k}\in\mathcal{R}_1,\;\forall k\in\left[K\right],
  \nonumber
\end{align}
where $\beta > 0 $ is a regularization constant. In this paper, we focus on the formulation \eqref{eq:ROAD-formulation} for noisy case, as the constant $\epsilon$ in \eqref{eq:ROAD-formulation} is directly linked to the specific performance requirement of applications and hence its adjustment is much more straightforward than choosing the constant $\lambda$ in \eqref{eq:ROAD-Lasso-formulation}.  
After solving \eqref{eq:ROAD-formulation}, the dictionary items $\bm{D}_{:,k}$ and the corresponding coefficients $\bm{X}_{k,:}$ can be obtained using singular value decomposition (SVD) of $\bm{Z}_k$. 

It is noteworthy that ROAD formulation in either \eqref{eq:ROAD-formulation} or \eqref{eq:ROAD-Lasso-formulation} is a non-convex optimization problem. The non-convexity comes from the constraint $\bm{Z}_k \in \mathcal{R}_1$. In Section \ref{sec:ADMM-solver}, a non-convex ADMM is introduced to solve \eqref{eq:ROAD-formulation} and its global convergence is guaranteed. 

\subsection{Discussion: Why Non-Convex Optimization}

Acute readers may wonder why not convexify the formulation \eqref{eq:ROAD-formulation}. 
One way to convexity \eqref{eq:ROAD-formulation} is to replace the constraint $\bm{Z}_{k}\in\mathcal{R}_1$ with a term in the objective function given by $\mu \sum_k \Vert \bm{Z}_k \Vert_*$, where $\mu>0$ is a regularization constant and $\Vert \cdot \Vert_*$ denotes nuclear norm which promotes low rank structure. However, as the following proposition shows, such a convex relaxation admits global optima that do not fit the purpose of dictionary learning.  
\begin{prop} \label{prop:convex-not-work}
  Consider the optimization problem \begin{align*}
    & \min_{\bm{Z}_k} \sum_k f(\bm{Z}_k) ~ \text{s.t.} ~ \sum_k \bm{Z}_k = \bm{Y},
  \end{align*}
  where $f(\cdot)$ is any well-defined norm. 
  The trivial solutions $\bm{Z}_k = a_k \bm{Y}$ where $0<a<1$ and $\sum_k a_k = 1$ are globally optimal. 
\end{prop}
\begin{proof}
  A norm $f(\cdot)$ satisfies the following properties: 
  \begin{itemize}
    \item non-negativity $f(\bm{Z})\ge 0$;
    \item absolute scalability $f(a \bm{Z}) = |a| f(\bm{Z})$; and
    \item triangular inequality $f(\bm{Z}_1 + \bm{Z}_2) \le f(\bm{Z}_1) + f(\bm{Z}_2)$.
  \end{itemize}
  
  For any feasible $\bm{Z}_k$s, by triangle inequality it holds that 
  \begin{align*}
    & \sum_k f(\bm{Z}_k) \ge f\left(\sum_k \bm{Z}_k\right) = f(\bm{Y}). 
  \end{align*}
  On the other hand, 
  \begin{align*}
    & \sum_k f(a_k \bm{Y}) = \sum_k a_k f(\bm{Y}) = f(\bm{Y}),
  \end{align*}
  where the first equality follows from absolute scalability and the second equality comes from the assumption that $\sum_k a_k =1 $. This establishes the global optimality of $\bm{Z}_k = a_k \bm{Y}$. 
\end{proof}

It is clear that $\Vert \cdot \Vert_{2,1} + \mu \Vert \cdot \Vert_*$ is a well-defined norm. Hence Proposition \ref{prop:convex-not-work} applies. But the global optima $\bm{Z}_k = a_k \bm{Y}$ are not rank-one matrices and do not generate the desired dictionary. The non-convex constraint $\bm{Z}_k \in \mathcal{R}_1$ in \eqref{eq:ROAD-formulation} is necessary. 

\subsection{Discussion: From Dictionary Coefficient Decomposition to Rank One Atomic Decomposition }

Benchmark algorithms in the literature are based on dictionary coefficient decomposition $\bm{Y} \approx \bm{D}\bm{X}$. Dictionary learning problem is hence bilinear in two unknown variables $\bm{D}$ and $\bm{X}$. In sparse coding stage, for example, the involved linear operator comes from $\bm{D}$ which can be ill-conditioned. It has been observed in \cite{dai2012simultaneous} that an ill-conditioned dictionary in the training process may not only introduce numerical instability but also trap the overall training process towards a singular point. The ill-conditionedness can be mitigated \cite{dai2012simultaneous} but is unavoidable in theory as long as the bilinear form is used. The advantage of dictionary coefficient decomposition is the efficiency of memory use for storing unknown variables and numerical computation in each iteration. 

By contrast, ROAD formulation involves single unknown variable $\bm{Z}_k$s. The involved operators in \eqref{eq:ROAD-formulation} are fixed and well-conditioned. As will be shown in Section \ref{sec:ADMM-solver}, this formulation allows a straightforward calculation to establish a universal Lipschitz constant for the atomic functions involved, and hence simplifies the analysis for global convergence of the ADMM solver. 

\section{Non-Convex ADMM for ROAD \label{sec:ADMM-solver}}

In this section, a non-convex ADMM procedure is developed to solve \eqref{eq:ROAD-formulation}, of which the global convergence is guaranteed. 

\subsection{Background on ADMM}\label{subsec:ADMM-background}

Alternating direction of method of multiplier (ADMM) \cite{boyd2011distributed} solves problems of the form
\begin{equation}
    \min_{\bm{x},\bm{z}}\; f(\bm{x})+g(\bm{z}) \quad {\rm s.t.}\; \bm{A}\bm{x}+\bm{B}\bm{z} = \bm{c}, 
    \label{eq:ADMM-general}
\end{equation}
where $f$ and $g$ are usually convex. The augmented  Lagrangian of \eqref{eq:ADMM-general} is given by 
\begin{align*} 
  & \mathcal{L}_{\rho}(\bm{x}, \bm{z}, \bm{\mu}) \\
  & = f(\bm{x})+g(\bm{z})+\bm{\mu}^{T}
  (\bm{A}\bm{x}+\bm{B}\bm{z} - \bm{c})\\ 
  & \quad + \frac{\rho}{2}\|\bm{A}\bm{x}+\bm{B}\bm{z} - \bm{c} \|_{2}^{2} \\
  & = f(\bm{x})+g(\bm{z}) + \frac{\rho}{2} \Vert\bm{A}\bm{x}+\bm{B}\bm{z} - \bm{c} + \frac{\bm{\mu}}{\rho} \Vert_{2}^{2} - \frac{1}{2\rho}\Vert \bm{\mu} \Vert_2^2,
\end{align*}
where $\bm{\mu}$ is the Lagrange multiplier, and $\rho>0$ is a predefined penalty parameter. Let $\bm{\lambda} = \frac{\bm{\mu}}{\rho}$
be the scaled Lagrange multiplier. The ADMM algorithm iteratively updates 
\begin{align} 
\bm{x}^{l+1} & :=\arg \min _{\bm{x}} \mathcal{L}_{\rho}\left(\bm{x}, \bm{z}^{l}, \bm{\lambda}^{l} \right), \label{eq:ADMM-step1} \\ 
\bm{z}^{l+1} & :=\arg \min _{\bm{z}} \mathcal{L}_{\rho}\left(\bm{x}^{l+1}, \bm{z}, \bm{\lambda}^{l}\right), \label{eq:ADMM-step2} \\ 
\bm{\lambda}^{l+1} & := \bm{\lambda}^{l}+\left( \bm{A} \bm{x}^{l+1} + \bm{B} \bm{z}^{l+1}- \bm{c}\right),
\end{align}
where the superscript $l$ denotes the index of iteration. 

Stopping criteria of ADMM can be derived from its optimality conditions. A solution of ADMM \eqref{eq:ADMM-general} is optimal if and only if 
\begin{align}
&  \bm{A}\bm{x}^{*}+\bm{B}\bm{z}^{*}-\bm{c}=\bm{0},\label{eq:pri-feasible-general}\\
&  \bm{0}\in\partial f(\bm{x}^{*})+\rho\bm{A}^{T}\bm{\lambda}^{*},\label{eq:dual-feasible1-general}\\
&  \bm{0}\in\partial g(\bm{z}^{*})+\rho\bm{B}^{T}\bm{\lambda}^{*},\label{eq:dual-feasible2-general}
\end{align}
where $\partial$ denotes the subdifferential operator. As $\bm{z}^{l+1}$ minimizes $\mathcal{L}_{\rho}(\bm{x}^{l+1},\bm{z},\bm{\lambda}^{l})$, it can be verified that $\bm{z}^{l+1}$ and $\bm{\lambda}^{l+1}$ always satisfy the condition \eqref{eq:dual-feasible2-general}. But from that $\bm{x}^{l+1}$ minimizes $\mathcal{L}_{\rho}(\bm{x},\bm{z}^{l},\bm{\lambda}^{l})$, it holds that 
\begin{equation}
    \rho\bm{A}^{T}\bm{B}(\bm{z}^{l+1}-\bm{z}^{l})\in \partial f(\bm{x}^{l+1})+\rho\bm{A}^{T}\bm{\lambda}^{l+1}.
\end{equation}
Hence, define primal and dual residues as 
\begin{align*}
  & \bm{r}^{l}=\bm{A} \bm{x}^{l} + \bm{B} \bm{z}^{l}- \bm{c}, \\
  & \bm{s}^{l}=\rho\bm{A}^{T}\bm{B}(\bm{z}^{l}-\bm{z}^{l-1}),
\end{align*}
respectively. Typical stopping criteria of ADMM iterations are given by 
\begin{equation}
    \Vert\bm{r}^{l}\Vert_{2}\leq\epsilon^{\rm pri}\;\;{\rm and}\;\; \Vert\bm{s}^{l}\Vert_{2}\leq\epsilon^{\rm dual},
\end{equation}
where $\epsilon^{\rm pri}>0$ and $\epsilon^{\rm dual}>0$ are tolerances for primal and dual residues respectively.

\subsection{An ADMM Solver for ROAD  \label{subsec:ADMM-solver}}

\subsubsection{Inexact ADMM
\label{subsubsec:ADMM-solver-inexact}}

The ADMM formulation for ROAD is defined as follows. For simplicity of composing, define an indicator function 
\begin{equation}
  \mathbbm{1}_{\mathcal{R}_{1}} (\bm{X})
  :=\begin{cases}
    0, & \text{if}~ \text{rank}(\bm{X}) \le 1, \\
    +\infty, & \text{otherwise}.
  \end{cases} 
  \label{eq:indicator-function}
\end{equation}
Then the optimization problem \eqref{eq:ROAD-formulation} is equivalent to
\begin{align}
  \min ~ & \sum_k \left\| \bm{Z}_{k}  \right\|_{2,1} + \sum_k \mathbbm{1}_{\mathcal{R}_{1}}(\bm{Z}_{k}) + \frac{\beta}{2}\left\| \bm{W} \right\|_{F}^{2}
  \nonumber\label{eq:ROAD-Noisy-EasyForm} \\
  {\rm s.t.} ~ &\bm{Y}=\sum_{k}\bm{Z}_{k}+\bm{W},
\end{align}
where we rewrite $\Vert\bm{Y}-\sum_{k}\bm{Z}_{k}\Vert_{F}^{2}\leq\epsilon$ as $\bm{Y}=\sum_{k}\bm{Z}_{k}+\bm{W}$ and minimize the power of the noisy term $\bm{W}$, and $\beta>0$ is a penalty parameter. We introduce three auxiliary variables $\bm{X}_{1,k},\; \bm{X}_{2,k},\;{\rm and}\; \bm{X}_{3,k} \in \mathbb{R}^{M \times N}$ w.r.t. the variable $\bm{Z}_k$ and another three auxiliary variables $\bm{Z}_{1,k},\; \bm{Z}_{2,k},\;{\rm and}\; \bm{Z}_{3} \in \mathbb{R}^{M \times N}$ regarding to the noisy term $\bm{W}$, and we write \eqref{eq:ROAD-Noisy-EasyForm} into inexact ADMM form as
\begin{align}
  \min ~ & \sum_k \left\| \bm{X}_{1,k}  \right\|_{2,1} + \sum_k \mathbbm{1}_{\mathcal{R}_{1}}(\bm{X}_{2,k}) + \sum_k \frac{\beta_{1}}{2}\left\| \bm{Z}_{1,k} \right\|_{F}^{2} \nonumber\\
  & + \sum_k \frac{\beta_{2}}{2}\left\| \bm{Z}_{2,k} \right\|_{F}^{2} + \frac{\beta_{3}}{2}\left\| \bm{Z}_{3} \right\|_{F}^{2}
  \nonumber\label{eq:ADMM-inexact-EasyForm} \\
  {\rm s.t.} ~ & \bm{X}_{1,k}=\bm{X}_{3,k}+\bm{Z}_{1,k},\;\bm{X}_{2,k}=\bm{X}_{3,k}+\bm{Z}_{2,k},\:\forall k\in[K], \nonumber\\
  &\bm{Y}=\sum_{k}\bm{X}_{3,k}+\bm{Z}_{3},
\end{align}
where $\beta_1$, $\beta_2$ and $\beta_3$ are penalty parameters to the noisy term, and the minimization is w.r.t. all $\bm{X}$'s and $\bm{Z}$'s. We choose the form in \eqref{eq:ADMM-inexact-EasyForm} for the convenience of the convergence proof in Section \ref{subsec:Convergence}.
 
For writing simplicity, we consider the ADMM formulations in scaled form. As there are $2MNK+MN$ many equality constraints in \eqref{eq:ADMM-inexact-EasyForm}, we denote the Lagrange multipliers by $\bm{\Lambda}_{1,k} \in \mathbb{R}^{M \times N}$, $\bm{\Lambda}_{2,k} \in \mathbb{R}^{M \times N}$ and $\bm{\Lambda}_{3} \in \mathbb{R}^{M \times N}$, corresponding to the equality constraints $\bm{X}_{1,k} = \bm{X}_{3,k}+\bm{Z}_{1,k}$, $\bm{X}_{2,k} = \bm{X}_{3,k}+\bm{Z}_{2,k}$ and $\bm{Y}=\sum_k \bm{X}_{3,k}+\bm{Z}_{3}$, respectively. The augmented Lagrangian can be formulated as 
\begin{align}
& \mathcal{L}_{\rho}\left(\bm{X}'s,\bm{Z}'s,\bm{\Lambda}'s\right)  \nonumber  \\
= & \sum_{k}( \Vert\bm{X}_{1,k}\Vert_{2,1} + \mathbbm{1}_{\mathcal{R}_{1}}(\bm{X}_{2,k})+\frac{\beta_1}{2}\Vert\bm{Z}_{1,k}\Vert_{F}^{2}  \nonumber\\
+&\frac{\beta_2}{2}\Vert\bm{Z}_{1,k}\Vert_{F}^{2}) 
+\frac{\beta_3}{2}\Vert\bm{Z}_{3}\Vert_{F}^{2} + \frac{\rho}{2} \sum_k \Vert \bm{X}_{1,k}-\bm{X}_{3,k}-\bm{Z}_{1,k}\nonumber\\
+&\bm{\Lambda}_{1,k}\Vert _{F}^{2} +\frac{\rho}{2} \sum_k \Vert\bm{X}_{2,k}-\bm{X}_{3,k}-\bm{Z}_{2,k}+\bm{\Lambda}_{2,k}\Vert _{F}^{2}  \nonumber  \\
+ & \frac{\rho}{2} \Vert\bm{Y}-\sum_k\bm{X}_{3,k}-\bm{Z}_{3}+\bm{\Lambda}_{3}\Vert _{F}^{2} \nonumber\\
- & \frac{\rho}{2}\sum_k\Vert\bm{\Lambda}_{1,k}\Vert_{F}^{2}-\frac{\rho}{2}\sum_k\Vert\bm{\Lambda}_{2,k}\Vert_{F}^{2} -\frac{\rho}{2} \Vert\bm{\Lambda}_{3}\Vert_{F}^{2} \label{eq:Augmented-Lagrangian}
\end{align}

Note that optimization in \eqref{eq:ADMM-inexact-EasyForm} can be rewritten into the standard ADMM form as introduced in \eqref{eq:ADMM-general}. Let $\bm{x} = [[\cdots, {\rm vec}(\bm{X}_{1,k})^T, \cdots],[\cdots, {\rm vec}(\bm{X}_{2,k})^T, \cdots],[\cdots, {\rm vec}$ $(\bm{X}_{3,k})^T, \cdots]]^T\in \mathbb{R}^{3MNK}$ and $\bm{z} = [[\cdots, {\rm vec}(\bm{Z}_{1,k})^T, \cdots],[\cdots, {\rm vec}(\bm{Z}_{2,k})^T, \cdots],{\rm vec}(\bm{Z}_{3})]\in \mathbb{R}^{2MNK+MN}$, where ${\rm vec}(\bm{X})$ denotes vectorization by stacking the columns of matrix $\bm{X}$. We further define $\bm{A}\in\mathbb{R}^{(2MNK+MN)\times 3MNK}$, $\bm{B}\in\mathbb{R}^{(2MNK+MN)\times 2MNK+MN}$ and $\bm{c}\in\mathbb{R}^{2MNK+MN}$ as
\begin{align*}
&\bm{A}:=\left[\begin{array}{ccc}
\bm{I}_{MNK}&\bm{0}_{MNK\times MNK}&-\bm{I}_{MNK}\\
\bm{0}_{MNK\times MNK}&\bm{I}_{MNK}&-\bm{I}_{MNK}\\
\bm{0}_{MN\times MNK}&\bm{0}_{MN\times MNK}&\bm{1}_{1\times K}\otimes\bm{I}_{MN}
\end{array}\right],\\
&\bm{B}:=\left[\begin{array}{ccc}
-\bm{I}_{MNK}&\bm{0}_{MNK\times MNK}&\bm{0}_{MN\times MN}\\
\bm{0}_{MNK\times MNK}&-\bm{I}_{MNK}&\bm{0}_{MN\times MN}\\
\bm{0}_{MN\times MNK}&\bm{0}_{MN\times MNK}&\bm{0}_{MN\times MN}
\end{array}\right],\\    
&\bm{c}:=\left[\begin{array}{c}
\bm{0}_{MNK}\\
\bm{0}_{MNK}\\
{\rm vec}(\bm{Y})
\end{array}\right],
\end{align*}
respectively, where the subscript is used to emphasize the dimension of the matrices, $\bm{I}_{M}$ denotes an $M$-by-$M$ identity matrix, and the symbol $\otimes$ stands for the Kronecker product. Then problem \eqref{eq:ADMM-inexact-EasyForm} can be written in the standard ADMM form
\begin{align}
    \min_{\bm{x},\bm{z}} ~ & \underbrace{\sum_{k} \Vert x_{1,k}(\bm{x}) \Vert _{2,1}+\sum_{k} \mathbbm{1}_{\mathcal{R}_1}( x_{2,k}(\bm{x}) )}_{f(\bm{x})}  +g(\bm{z})\nonumber \\
    {\rm s.t.}~ &  \bm{A}\bm{x}+\bm{B}\bm{z}=\bm{c},\label{eq:ADMM-inexact-StandardForm}    
\end{align}
where 
\begin{align*}
    & g(\bm{z}) = \frac{1}{2}\Vert\bm{\beta}\bm{z}\Vert_{2}^{2}, \\
    & \bm{\beta}= \left[\begin{array}{ccc}
    \beta_1\bm{I}_{MNK}& &\\
    & \beta_2\bm{I}_{MNK} & \\
    &  & \beta_3\bm{I}_{MN}
    \end{array}\right]
\end{align*}
and
\begin{align*}
     x_{1,k}(\bm{x}) = & [\bm{x} _{(k-1)MN+1 : (k-1)MN+M}, \\
    & \cdots,\bm{z} _{(k-1)MN+(N-1)M+1 : kMN}],\\
     x_{2,k}(\bm{x}) = & [\bm{x} _{MNK+(k-1)MN+1 : MNK+(k-1)MN+M}, \cdots, \\
    & \bm{x} _{MNK+(k-1)MN+(N-1)M+1 : MNK+kMN}]
\end{align*}
construct matrices of size $M\times N$ from $\bm{x}$. The augmented Lagrangian can be derived as
\begin{align}
\mathcal{L}_{\rho}\left(\bm{x},\bm{z},\bm{\lambda}\right)
& =f(\bm{x})+ g(\bm{z})+\frac{\rho}{2}\Vert\bm{A}\bm{x}+\bm{B}\bm{z}-\bm{c}+\bm{\lambda}\Vert_{2}^{2} \nonumber\\
& -\frac{\rho}{2}\Vert\bm{\lambda}\Vert_{2}^{2},\label{eq:augmented-Lagrangian-StandardForm}
\end{align}
where $\bm{\lambda} = [[\cdots, {\rm vec}(\bm{\Lambda}_{1,k})^T, \cdots],[\cdots, {\rm vec}(\bm{\Lambda}_{2,k})^T, \cdots],$ $[{\rm vec}(\bm{\Lambda}_{3})^T]]^T\in \mathbb{R}^{2MNK+MN}$.

However, it is more explanatory to derive ADMM steps using \eqref{eq:ADMM-inexact-EasyForm} rather than \eqref{eq:ADMM-inexact-StandardForm}. Furthermore, note that the functions of the variables $\bm{X}$'s include both convex and non-convex functions. To ensure the convergence rate and the convergence of non-convex ADMM, we employ different penalty parameter $\rho$'s corresponding to different constraints instead of fixing $\rho$, and reformulate the augmented Lagrangian \eqref{eq:Augmented-Lagrangian} as 
\begin{align}
& \mathcal{L}_{\rho 's}\left(\bm{X}'s,\bm{Z}'s,\bm{\Lambda}'s\right)  \nonumber  \\
= & \sum_{k}( \Vert\bm{X}_{1,k}\Vert_{2,1} + \mathbbm{1}_{\mathcal{R}_{1}}(\bm{X}_{2,k})+\frac{\beta_1}{2}\Vert\bm{Z}_{1,k}\Vert_{F}^{2}+\frac{\beta_2}{2}\Vert\bm{Z}_{1,k}\Vert_{F}^{2})   \nonumber\\
+& \frac{\beta_3}{2}\Vert\bm{Z}_{3}\Vert_{F}^{2} + \frac{\rho_{1}}{2} \sum_k (\Vert \bm{X}_{1,k}-\bm{X}_{3,k}-\bm{Z}_{1,k}+\bm{\Lambda}_{1,k}\Vert _{F}^{2} \nonumber\\
-& \Vert\bm{\Lambda}_{1,k}\Vert_{F}^{2})  + \frac{\rho_{2}}{2}\sum_k (\Vert\bm{X}_{2,k}-\bm{X}_{3,k}-\bm{Z}_{2,k}+\bm{\Lambda}_{2,k}\Vert _{F}^{2}   \nonumber  \\
- & \Vert\bm{\Lambda}_{2,k}\Vert_{F}^{2}) + \frac{\rho_{3}}{2} (\Vert\bm{Y}-\sum_k\bm{X}_{3,k}-\bm{Z}_{3}+\bm{\Lambda}_{3}\Vert _{F}^{2}-\Vert\bm{\Lambda}_{3}\Vert_{F}^{2}). \label{eq:Augmented-Lagrangian-DifRho}
\end{align}
Then the inexact ADMM iterations are given by 
\begin{align}
& \bm{X}_{1,k}^{l+1}=\underset{\bm{X}_{1,k}}{\arg\min}\;\Vert\bm{X}_{1,k}\Vert_{2,1}+\frac{\rho_1}{2}\Vert \bm{X}_{1,k}-\bm{X}_{3,k}^{l}-\bm{Z}_{1,k}^{l} \label{eq:ADMM-X1-update} \nonumber\\
& \quad\quad\,\,  +\bm{\Lambda}_{1,k}^{l}\Vert _{F}^{2}, \\
& \bm{X}_{2,k}^{l+1} = \underset{\bm{X}_{2,k}}{\arg\min} \;\mathbbm{1}_{\mathcal{R}_{1}}(\bm{X}_{2,k}) +\frac{\rho_2}{2} \Vert  \bm{X}_{2,k}-\bm{X}_{3,k}^{l} -\bm{Z}_{2,k}^{l} \nonumber\\
& \quad\quad\,\, + \bm{\Lambda}_{2,k}^{l} \Vert _{F}^{2},\label{eq:ADMM-X2-update}  \\
&(\cdots,\bm{X}_{3,k}^{l+1},\cdots) = \underset{\cdots,\bm{X}_{3,k},\cdots}{\arg\min} \; \rho_1 \sum_k\Vert \bm{X}_{1,k}^{l+1}-\bm{X}_{3,k}-\bm{Z}_{1,k}^{l} +\nonumber \\
& \quad\quad\,\, \bm{\Lambda}_{1,k}^{l}\Vert _{F}^{2} +\rho_2\sum_k\Vert\bm{X}_{2,k}^{l+1} -\bm{X}_{3,k}-\bm{Z}_{2,k}^{l} + \bm{\Lambda}_{2,k}^{l} \Vert _{F}^{2} \nonumber \\
& \quad\quad\,\, +\rho_3\Vert\bm{Y}- \sum_{k}\bm{X}_{3,k}-\bm{Z}_{3}^{l}+\bm{\Lambda}_{3}^{l}\Vert_{F}^{2}, \label{eq:ADMM-X3-update}   \\
& \bm{Z}_{1,k}^{l+1}=\underset{\bm{Z}_{1,k}}{\arg\min}\;\beta_1\Vert\bm{Z}_{1,k}\Vert_{F}^{2}+\rho_1\Vert\bm{X}_{1,k}^{l+1}-\bm{X}_{3,k}^{l+1}-\bm{Z}_{1,k} \nonumber\\
& \quad\quad\,\, +\bm{\Lambda}_{1,k}^{l}\Vert _{F}^{2},\label{eq:ADMM-Z1-update}   \\
& \bm{Z}_{2,k}^{l+1}=\underset{\bm{Z}_{2,k}}{\arg\min}\;\beta_2\Vert\bm{Z}_{2,k}\Vert_{F}^{2}+\rho_2\Vert\bm{X}_{2,k}^{l+1}-\bm{X}_{3,k}^{l+1}-\bm{Z}_{2,k} \nonumber\\
& \quad\quad\,\, +\bm{\Lambda}_{2,k}^{l}\Vert _{F}^{2},\label{eq:ADMM-Z2-update}   \\
& \bm{Z}_{3}^{l+1}=\underset{\bm{Z}_{3}}{\arg\min}\;\beta_3\Vert\bm{Z}_{3}\Vert_{F}^{2}+\rho_3\Vert \bm{Y}-\sum_k\bm{X}_{3,k}^{l+1}-\bm{Z}_{3} \label{eq:ADMM-Z3-update}  \nonumber \\
& \quad\quad \,\, +\bm{\Lambda}_{3}^{l}\Vert _{F}^{2}, \\
&\bm{\Lambda}_{1,k}^{l+1}  =\bm{\Lambda}_{1,k}^{l}+(\bm{X}_{1,k}^{l+1}-\bm{X}_{3,k}^{l+1}-\bm{Z}_{1,k}^{l+1}),\label{eq:ADMM-Lambda1k-Update}\\
&\bm{\Lambda}_{2,k}^{l+1}  =\bm{\Lambda}_{2,k}^{l}+(\bm{X}_{2,k}^{l+1}-\bm{X}_{3,k}^{l+1}-\bm{Z}_{2,k}^{l+1}),\label{eq:ADMM-Lambda2k-Update}\\
&\bm{\Lambda}_{3}^{l+1}  =\bm{\Lambda}_{3}^{l}+(\bm{Y}-\sum_k\bm{X}_{3,k}^{l+1}-\bm{Z}_{3}^{l+1}),\label{eq:ADMM-Lambda3-Update}
\end{align}
where $l$ represents the step number.

The six optimization problems (\ref{eq:ADMM-X1-update}-\ref{eq:ADMM-Z3-update}) involved in ADMM iterations are conceptually easy to solve. 
The optimization problem \eqref{eq:ADMM-X1-update} is convex but involves a non-differential term $\Vert \cdot \Vert_{2,1}$ in its objective function. The closed form of the optimal solution of \eqref{eq:ADMM-X1-update} can be obtained by setting the sub-gradient of the objective function to zero. Define $\hat{\bm{X}}_{1,k} := \bm{X}_{3,k}^{l}+\bm{Z}_{1,k}^{l}-\bm{\Lambda}_{1,k}^{l}$. Then 
\begin{equation}
(\bm{X}_{1,k}^{l+1})_{:,n}=\left(1-\frac{1}{\rho_1\Vert(\hat{\bm{X}}_{1,k})_{:,n}\Vert_{2}}\right)_{+}(\hat{\bm{X}}_{1,k})_{:,n},\label{eq:X1-solution}
\end{equation}
where $(x)_+ := \max(0,x)$. 

The optimization problem \eqref{eq:ADMM-X2-update} is non-convex. Fortunately, by Eckart-Young-Mirsky theorem, it can be solved by using singular value decomposition (SVD). Define $\hat{\bm{X}}_{2,k} = \bm{X}_{3,k}^{l}+\bm{Z}_{2,k}^{l}-\bm{\Lambda}_{2,k}^{l}$. Considering the SVD of the matrix $\hat{\bm{X}}_{2,k}=\bm{U}diag(\bm{\sigma})\bm{V}^T$, we obtain $\bm{X}_{2,k}^{l+1}$ via rank-one projection of $\hat{\bm{X}}_{2,k}$.
\begin{definition}(Rank-$r$ projection)
For a matrix $\bm{X}\in\mathbb{R}^{M\times N}$ with rank $k\leq \min\{M,N\}$, we define rank-$r$ ($1\leq r<k$) projection of $\bm{X}$ as
\begin{equation}
    \bm{H}_{\sigma_r}(\bm{X})=\bm{U}diag(\bm{H}_{r}(\bm{\sigma}))\bm{V}^T
\end{equation}
where $\bm{X}=\bm{U}diag(\bm{\sigma})\bm{V}^T$ is the SVD of $\bm{X}$, $\bm{H}_{r}$ denotes the hard thresholding function defined as
\begin{equation}
    (\bm{H}_{r}(\bm{\sigma}))_i=\left\{\begin{array}{ll}{\bm{\sigma}_{i}} & {\text { if } i \leq r,} \\ {0} & {\rm otherwise.}\end{array}\right.
\end{equation}
\end{definition}
Then 
\begin{equation}
\bm{X}_{2,k}^{l+1}=\bm{H}_{\sigma_1}(\hat{\bm{X}}_{2,k})=\bm{\sigma}_{1}\bm{U}_{:,1}\bm{V}_{:,1}^{T}.\label{eq:X2-solution}
\end{equation}

The optimization problem \eqref{eq:ADMM-X3-update} is a quadratic programming, the solution of which can be obtained by many approaches such as pseudo-inverse or conjugate gradient method \cite{nocedal2006conjugate}, in principle. However, as it involves $2MNK+MN$ linear mappings, the computational complexity of all typical methods is huge. Fortunately, according to the specific structure of this quadratic problem, we can efficiently derive a closed form optimal solution using Woodbury matrix identity to save the time complexity significantly. 
Define $\bm{\mathcal{B}}_1=[\bm{B}_{1,1}^{T},\cdots,\bm{B}_{1,K}^{T}]^{T}$, $\bm{\mathcal{B}}_2=[\bm{B}_{2,1}^{T},\cdots,\bm{B}_{2,K}^{T}]^{T}$ and $\bm{\mathcal{B}}_3=[\bm{B}_{3}^{T},\cdots,\bm{B}_{3}^{T}]^{T}$ such that 
$\bm{B}_{1,k}=\bm{X}_{1,k}^{l+1}-\bm{Z}_{1,k}^{l}+\bm{\Lambda}_{1,k}^{l}$, $\bm{B}_{2,k}=\bm{X}_{2,k}^{l+1}-\bm{Z}_{2,k}^{l}+\bm{\Lambda}_{2,k}^{l}$ and $\bm{B}_{3}=\bm{Y}-\bm{Z}_{3}^{l}+\bm{\Lambda}_{3}^{l}$.
Then, following Woodbury matrix identity, we can efficiently compute the closed form solution to $\bm{X}_{3,k}^{l+1}$ as
\begin{align}
    \bm{X}_{3,k}^{l+1}=\bm{\mathcal{A}}_{(k+1)N+1:kN,:} \left(\bm{\mathcal{B}}_{1}+\bm{\mathcal{B}}_{2}+\bm{\mathcal{B}}_{3}\right)
\end{align}
where the matrix $\bm{\mathcal{A}}$ is formulated by
\begin{equation}
    \bm{\mathcal{A}}=\frac{1}{2}\left( \bm{I}_{MK}-\frac{1}{K+2} \bm{1}_{K\times K}\otimes \bm{I}_{M} \right).
\end{equation}

The optimization problems (\ref{eq:ADMM-Z1-update}-\ref{eq:ADMM-Z3-update}) can be simply solved using first order optimality. The solutions are given by
\begin{align}
\bm{Z}_{1,k}^{l+1}=\frac{\rho_1}{\beta_1 +\rho_1}(\bm{X}_{1,k}^{l+1}-\bm{X}_{3,k}^{l+1}+\bm{\Lambda}_{1,k}^{l}),\label{eq:Z1-solver}
\end{align}
\begin{align}
\bm{Z}_{2,k}^{l+1}=\frac{\rho_2}{\beta_2 +\rho_2}(\bm{X}_{2,k}^{l+1}-\bm{X}_{3,k}^{l+1}+\bm{\Lambda}_{2,k}^{l}),\label{eq:Z2-solver}
\end{align}
\begin{align}
\bm{Z}_{3}^{l+1}=\frac{\rho_3}{\beta_1 +\rho_3}(\bm{Y}-\sum_{k}\bm{X}_{3,k}^{l+1}+\bm{\Lambda}_{3}^{l}).\label{eq:Z3-solver}
\end{align}

\subsubsection{Exact ADMM}

Considering the noise-free case, we can exactly write $\bm{Y}=\sum_{k}\bm{Z}_{k}$ and rewrite the general formulation \eqref{eq:ADMM-inexact-EasyForm} into
\begin{align}
  \min ~ & \sum_k \left\| \bm{X}_{1,k}  \right\|_{2,1} + \sum_k \mathbbm{1}_{\mathcal{R}_{1}}(\bm{X}_{2,k}) 
  \nonumber\label{eq:ADMM-exact-EasyForm} \\
  {\rm s.t.} ~ & \bm{X}_{1,k}=\bm{X}_{3,k},\;\bm{X}_{2,k}=\bm{X}_{3,k},\;\bm{Y}=\sum_{k}\bm{X}_{3,k}\:\forall k\in[K],
\end{align}
where the minimization is w.r.t. all $\bm{X}$'s. Note that there are different ways to formulate an ADMM optimization problem equivalent to \eqref{eq:ROAD-Lasso-formulation}.

Again,
we denote the Lagrange multipliers by $\bm{\Lambda}_{1,k} \in \mathbb{R}^{M \times N}$, $\bm{\Lambda}_{2,k} \in \mathbb{R}^{M \times N}$ and $\bm{\Lambda}_{3} \in \mathbb{R}^{M \times N}$, corresponding to the equality constraints $\bm{X}_{1,k} = \bm{X}_{3,k}$, $\bm{X}_{2,k} = \bm{X}_{3,k}$ and $\bm{Y}=\sum_k \bm{X}_{3,k}$, respectively. The augmented Lagrangian can be formulated as 
\begin{align}
& \mathcal{L}_{\rho}\left(\bm{X}_{1,k},\bm{X}_{2,k},\bm{X}_{3,k},\bm{\Lambda}_{1,k},\bm{\Lambda}_{2,k},\bm{\Lambda}_3\right)  \nonumber  \\
= & \sum_{k}( \Vert\bm{X}_{1,k}\Vert_{2,1} + \frac{\rho}{2} \Vert \bm{X}_{3,k}-\bm{X}_{1,k} + \bm{\Lambda}_{1,k}\Vert _{F}^{2} - \frac{\rho}{2}\Vert\bm{\Lambda}_{1,k}\Vert_{F}^{2} )  \nonumber  \\
+ & \sum_{k}( \mathbbm{1}_{\mathcal{R}_{1}}(\bm{X}_{2,k}) + \frac{\rho}{2} \Vert  \bm{X}_{3,k} -\bm{X}_{2,k} + \bm{\Lambda}_{2,k} \Vert _{F}^{2}   \nonumber \\
- & \frac{\rho}{2}\Vert\bm{\Lambda}_{2,k}\Vert_{F}^{2})  +  \frac{\rho}{2}\Vert\sum_{k}\bm{X}_{3,k}-\bm{Y}+\bm{\Lambda}_3\Vert_{F}^{2}-\frac{\rho}{2}\Vert\bm{\Lambda}_3\Vert_{F}^{2}  \label{eq:Augmented-Lagrangian-exact}
\end{align}
Then the ADMM iterations are given by 
\begin{align}
&\bm{X}_{1,k}^{l+1}=\underset{\bm{X}_{1,k}}{\arg\min}\;\Vert\bm{X}_{1,k}\Vert_{2,1}+\frac{\rho}{2}\Vert \bm{X}_{3,k}^{l}-\bm{X}_{1,k}+\bm{\Lambda}_{1,k}^{l}\Vert _{F}^{2},\label{eq:ADMM-P-update}   \\
&\bm{X}_{2,k}^{l+1} = \underset{\bm{X}_{2,k}}{\arg\min} \;\mathbbm{1}_{\mathcal{R}_{1}}(\bm{X}_{2,k}) +\frac{\rho}{2} \Vert  \bm{X}_{3,k}^{l} -\bm{X}_{2,k} ,\label{eq:ADMM-Q-update} \nonumber  \\
& \quad \quad\,\,+ \bm{\Lambda}_{2,k}^{l} \Vert _{F}^{2} \\
&(\cdots,\bm{X}_{3,k}^{l+1},\cdots) = \underset{\cdots,\bm{X}_{3,k},\cdots}{\arg\min} \;  \sum_k\Vert \bm{X}_{3,k}-\bm{X}_{1,k}^{l+1} + \bm{\Lambda}_{1,k}^{l}\Vert _{F}^{2} \nonumber \\
&+\sum_k\Vert\bm{X}_{3,k} -\bm{X}_{2,k}^{l+1} + \bm{\Lambda}_{2,k}^{l} \Vert _{F}^{2}+\Vert\sum_{k}\bm{X}_{3,k}-\bm{Y}+\bm{\Lambda}_3^{l}\Vert_{F}^{2}, \label{eq:ADMM-Z-update}   \\
&\bm{\Lambda}_{1,k}^{l+1} =\bm{\Lambda}_{1,k}^{l}+(\bm{X}_{3,k}^{l+1}-\bm{X}_{1,k}^{l+1}),\\
&\bm{\Lambda}_{2,k}^{l+1}  =\bm{\Lambda}_{2,k}^{l}+(\bm{X}_{3,k}^{l+1}-\bm{X}_{2,k}^{l+1}),\\
&\bm{\Lambda}_{3}^{l+1}  =\bm{\Lambda}_{3}^{l}+(\sum_k\bm{X}_{3,k}^{l+1}-Y),
\end{align}
where $l$ represents the step number.

The three sub-problems (\ref{eq:ADMM-Z-update}-\ref{eq:ADMM-Q-update}) involved in ADMM iterations are similar to the optimization problems (\ref{eq:ADMM-X1-update}-\ref{eq:ADMM-X3-update}) and can be solved using the same mechanisms.

Considering the noisy case, we denote the noise power $\epsilon$ and define the indicator function.
\begin{equation}
  \mathbbm{1}_{\Vert \cdot \Vert_F \le \epsilon}(\bm{X})
  :=\begin{cases}
    0, & \text{if}~ \Vert \cdot \Vert_F \le \epsilon, \\
    +\infty & \text{otherwise}.
  \end{cases} \label{eq:indicator-function-noise}
\end{equation}
Then the noisy form of the optimization problem \eqref{eq:ADMM-exact-EasyForm} can be written as
\begin{align}
 \min ~  &  \sum_{k} \Vert \bm{X}_{1,k} \Vert _{2,1} + \sum_{k} \mathbbm{1}_{\mathcal{R}_{1}}(\bm{X}_{2,k}) \nonumber\\
&   + \mathbbm{1}_{\left\| \cdot \right\|_F \le \epsilon} ( \bm{W}-\bm{Y} ) \nonumber \\
{\rm s.t.} ~ & \bm{X}_{1,k}=\bm{X}_{3,k},\;\bm{X}_{2,k}=\bm{X}_{3,k},\;\bm{W}=\sum_{k}\bm{X}_{3,k},\:\forall k\in[K].\label{eq:ADMM-noisy-form}
\end{align}
Denote the Lagrange multipliers by $\bm{\Lambda}_{1,k} \in \mathbb{R}^{M \times N}$, $\bm{\Lambda}_{2,k} \in \mathbb{R}^{M \times N}$ and $\bm{\Lambda}_3 \in \mathbb{R}^{M \times N}$ corresponding to the equality constraints $\bm{X}_{1,k} = \bm{X}_{3,k}$, $\bm{X}_{2,k} = \bm{X}_{3,k}$ and $\bm{W}=\sum_k \bm{X}_{3,k}$, respectively. Then the complete ADMM iterations are thus given by
\begin{equation}
\bm{X}_{1,k}^{l+1}=\underset{\bm{X}_{1,k}}{\arg\min}\;\Vert\bm{X}_{1,k}\Vert_{2,1}+\frac{\rho}{2}\Vert \bm{X}_{3,k}^{l}-\bm{X}_{1,k}+\bm{\Lambda}_{1,k}^{l}\Vert _{F}^{2},
\end{equation}
\begin{equation}
\bm{X}_{2,k}^{l+1} = \underset{\bm{X}_{2,k}}{\arg\min} \;\mathbbm{1}_{\mathcal{R}_{1}}(\bm{X}_{2,k}) +\frac{\rho}{2} \Vert  \bm{X}_{3,k}^{l} -\bm{X}_{2,k} + \bm{\Lambda}_{2,k}^{l} \Vert _{F}^{2},
\end{equation}
\begin{align}
& (\cdots,\bm{X}_{3,k}^{l+1},\cdots) = \underset{\cdots,\bm{X}_{3,k},\cdots}{\arg\min} \;
\sum_{k} \Vert \bm{X}_{3,k}-\bm{X}_{1,k}^{l+1}+\bm{\Lambda}_{1,k}^{l} \Vert_{F}^{2} \nonumber\\
& \sum_{k} \Vert \bm{X}_{3,k}-\bm{X}_{1,k}^{l+1}+\bm{\Lambda}_{2,k}^{l} \Vert_{F}^{2} + \Vert \sum_{k} \bm{X}_{3,k}-\bm{W} +\bm{\Lambda}_{3}^{l} \Vert_{F}^{2} ,
\end{align}
\begin{equation}
\bm{W}^{l+1} = \underset{\bm{W}}{\arg\min} \;\mathbbm{1}_{\left\| \cdot \right\|_F \le \epsilon}(\bm{W}-\bm{Y})
+ \frac{\rho}{2}\Vert \sum_k\bm{X}_{3,k}^{l+1} - \bm{W} +\bm{\Lambda}_{3}^{l} \Vert_{F}^{2}  ,\label{eq:ADMM-W-update}
\end{equation}
\begin{align}
\bm{\Lambda}_{1,k}^{l+1} & =\bm{\Lambda}_{1,k}^{l}+(\bm{X}_{3,k}^{l+1}-\bm{X}_{1,k}^{l+1}),\\
\bm{\Lambda}_{2,k}^{l+1} & =\bm{\Lambda}_{2,k}^{l}+(\bm{X}_{3,k}^{l+1}-\bm{X}_{1,k}^{l+1}),\\
\bm{\Lambda}_{3}^{l+1} & =\bm{\Lambda}_{3}^{l}+(\sum_{k}\bm{X}_{3,k}^{l+1}-\bm{W}).
\end{align}
Here we add the $\bm{W}^{l+1}$ updating step in \eqref{eq:ADMM-W-update}, and the solution to this optimization problem is straightforward. Define $\hat{\bm{W}}:=\sum_k\bm{Z}_k^{l+1}+\bm{\Lambda}_{3}^{l}$, and then
\begin{equation}
\bm{W}^{l+1}=\epsilon\frac{\hat{\bm{W}}-\bm{Y}}{\Vert\hat{\bm{W}}-\bm{Y}\Vert_{F}}+\bm{Y}.\label{eq:W-solution}
\end{equation}

\subsection{Convergence of ROAD \label{subsec:Convergence}}

ROAD involves a non-convex ADMM with a non-smooth objective function. It is important to ensure its convergence before using it in practice. From the results in \cite[Theorem 1]{wang2019global}, our ROAD algorithm indeed enjoys the global convergence guarantee.

\begin{thm}\label{thm:convergence-ROAD}
Consider the ADMM formulation of our ROAD problem \eqref{eq:ADMM-inexact-EasyForm} and its corresponding ADMM iterations (\ref{eq:ADMM-X1-update}-\ref{eq:ADMM-Lambda3-Update}) defined in Section \ref{subsubsec:ADMM-solver-inexact}. Suppose the lower bounds of the penalty parameters $\rho$'s of the Lagrangian \eqref{eq:Augmented-Lagrangian-DifRho} are $\rho_1>\beta_1+2$, $\rho_2>\beta_2+2$ and $\rho_3>\beta_3+2$, respectively.
Then the ADMM process (\ref{eq:ADMM-X1-update}-\ref{eq:ADMM-Lambda3-Update}) converges to a stationary point of $\left(\bm{X}^{*}{'}s,\bm{Z}^{*}{'}s,\bm{\Lambda}^{*}{'}s\right)$, 
where $\bm{0} \in \partial \mathcal{L}_{\rho 's}\left(\bm{X}^{*}{'}s,\bm{Z}^{*}{'}s,\bm{\Lambda}^{*}{'}s\right)$, or equivalently,
\begin{subequations}\label{eq:stop-conditions}
\begin{align}
& \bm{X}_{1,k}^{*} =\bm{X}_{3,k}^{*}+\bm{Z}_{1,k}^{*},\;\bm{X}_{2,k}^{*}=\bm{X}_{3,k}^{*}+\bm{Z}_{2,k}^{*}, \nonumber\\
& \bm{Y}=\sum_{k}\bm{X}_{3,k}^{*}+\bm{Z}_{3}^{*}, \label{eq:subdif-Lambdas}\\
& \bm{0} \in\partial\Vert\bm{X}_{1,k}^{*}\Vert_{2,1}+\rho_1\bm{\Lambda}_{1,k}^{*},  \label{eq:subdif-Pk}  \\
& \bm{0} \in\partial_{p} \mathbbm{1}_{\mathcal{R}_{1}}(\bm{X}_{2,k}^{*})+\rho_2\bm{\Lambda}_{2,k}^{*}, \label{eq:subdif-Qk}  \\
& \bm{0} \in \rho_1\bm{\Lambda}_{1,k}^{*}+\rho_2\bm{\Lambda}_{2,k}^{*}+\rho_3\bm{\Lambda}_{3}^{*},\:\forall k\in[K]. \label{eq:subdif-Zk}
\end{align}
\end{subequations}
\end{thm}

Before proving Theorem \ref{thm:convergence-ROAD}, we first consider the properties of the objective function of ADMM formulation \eqref{eq:ADMM-inexact-EasyForm} and the ADMM process (\ref{eq:ADMM-X1-update}-\ref{eq:ADMM-Lambda3-Update}).

\begin{definition}
(Lower semi-continuity) We say that $f$ is lower semi-continuous
at $x_{0}$ if
\[
\lim_{x\rightarrow x_{0}}\inf f(x)\geq f(x_{0}).
\]
\end{definition}

\begin{prop}
The indicator function \eqref{eq:indicator-function} is lower semi-continuous.
\end{prop}
\begin{proof}
Indicator function of a closed set is lower semi-continuous. It is clear that the set of rank-one matrices is closed. Hence the indicator function \eqref{eq:indicator-function} is lower semi-continuous.
\end{proof}

\begin{prop} \label{prop:subprob-Lipschitz}
The subproblems (\ref{eq:ADMM-X1-update}-\ref{eq:ADMM-X3-update}) are Lipschitz continuous.
\end{prop}
\begin{proof}
The subproblem \eqref{eq:ADMM-X1-update} includes $\ell_{2,1}$ norm, which can be solved by \eqref{eq:X1-solution}. It is hence Lipschitz continuous with Lipschitz constant $L_{p}=1$. According to the results in \cite{andersson2016operator}, the rank-one projection function \eqref{eq:X2-solution}, which is the solution of the subproblem \eqref{eq:ADMM-X2-update}, is Lipschitz continuous with Lipschitz constant $L_{q}=1$
\end{proof}

\subsection{Proof}

\subsubsection{Sufficient descent of $\mathcal{L}_{\rho 's}$}

We first prove the following lemma
\begin{lem}\label{lem:Lagranian-descends}
If the penalty parameters $\rho_1>\beta_1+2$, $\rho_2>\beta_2+2$ and $\rho_3>\beta_3+2$, respectively, the augmented Lagrangian $\mathcal{L}_{\rho 's}$ descends for all sufficient large $l$, that is 
\begin{align}\label{eq:lower-bound-difLagrangian}
&\mathcal{L}_{\rho 's}\left(\bm{X}^{l}{\rm 's},\bm{Z}^{l}{'s},\bm{\Lambda}^{l}{'s}\right) -\mathcal{L}_{\rho 's}\left(\bm{X}^{l+1}{\rm 's},\bm{Z}^{l+1}{'s},\bm{\Lambda}^{l+1}{'s}\right) \nonumber\\
\geq & \frac{\rho_1}{2}\sum_{k}\Vert\bm{X}_{1,k}^{l}-\bm{X}_{1,k}^{l+1}\Vert_{F}^{2} + \frac{\rho_1+\rho_2}{2}\sum_k\Vert\bm{X}_{3,k}^{l}-\bm{X}_{3,k}^{l+1}\Vert_{F}^{2} \nonumber\\
+&\frac{\rho_3}{2}\Vert\sum_k(\bm{X}_{3,k}^{l}-\bm{X}_{3,k}^{l+1})\Vert_{F}^{2} +\frac{\rho_1-\beta_1}{2}\sum_k\Vert\bm{Z}_{1,k}^{l}-\bm{Z}_{1,k}^{l+1}\Vert_{F}^{2} \nonumber\\
+& \frac{\rho_2-\beta_2 }{2}\sum_k\Vert\bm{Z}_{2,k}^{l}-\bm{Z}_{2,k}^{l+1}\Vert_{F}^{2} + \frac{\rho_3-\beta_3 }{2}\Vert\bm{Z}_{3}^{l}-\bm{Z}_{3}^{l+1}\Vert_{F}^{2} \nonumber\\
>& 0.
\end{align}
\end{lem}

Before we go through the proof of Lemma \ref{lem:Lagranian-descends}, we first consider the subdifferential of $\ell_{2,1}$ norm and define the proximal subdifferential of indicator function $\mathbbm{1}_{\mathcal{R}_{1}}$.
The subdifferential of $\Vert\bm{Z}\Vert_{2,1}$ to a matrix $\bm{Z}$ is simply given by
\begin{equation}    \label{eq:subdifferential-L21}
(\partial\Vert\bm{Z}\Vert_{2,1})_{i,j}=\begin{cases}
\begin{array}{ll}
\frac{(\bm{Z})_{i,j}}{\Vert(\bm{Z})_{:,j}\Vert_{2}},\\
\{\bm{W}_{i,j}:\Vert\bm{W}_{:,j}\Vert_{2}\leq 1\},
\end{array} & \begin{array}{c}
(\bm{Z})_{:,j}\neq\bm{0},\\
\mathrm{otherwise}.
\end{array}\end{cases}
\end{equation}
To derive the proximal subdifferential of $\mathbbm{1}_{\mathcal{R}_{1}}$, we first define the concept of prox-regularity.
\begin{definition}\label{def:prox-regularity}\cite[Definition 1.1.]{poliquin2000local}
(Prox-regularity) A closed set $\mathcal{C}$ is prox-regular at $\bar{x}$ for $\bar{v}$, where $\bar{x}\in\mathcal{C}$, $\bar{v}\in\mathcal{N}_{\mathcal{C}}(\bar{x})$ and $\mathcal{N}_{\mathcal{C}}(\bar{x})$ denotes the general cone of normals to $\mathcal{C}$ at $\bar{x}$, if there exist $\epsilon>0$ and $\rho>0$ such that whenever $x\in\mathcal{C}$ and $v\in\mathcal{N}_{\mathcal{C}}(x)$ with $|x-\bar{x}|<\epsilon$ and $|v-\bar{v}|<\epsilon$, then x is the unique nearest point of $\{x^{'}\in\mathcal{C}\big\vert\vert x^{'}-\bar{x}\vert<\epsilon\}$ to $x+\gamma^{-1}v$.
\end{definition}
\begin{prop}\label{prop:prox-regular-inequation}\cite[Proposition 1.2.]{poliquin2000local}
A closed set $\mathcal{C}$ is prox-regular at $\bar{x}$ for $\bar{v}$, or equivalently, there exists an $\epsilon>0$ and $\gamma>0$ such that whenever $x\in\mathcal{C}$ and $v\in\mathcal{N}_{\mathcal{C}}(x)$ with $|x-\bar{x}|<\epsilon$ and $|v|<\epsilon$, one has
\begin{equation}
    0\geq\langle v,x^{'}-x\rangle-\frac{\gamma}{2}|x^{'}-x|^{2}\quad \forall x^{'}\in\mathcal{C}\;with\;|x^{'}-\bar{x}|<\epsilon.
\end{equation}
\end{prop}
Once we have the definition of prox-regularity, we have the following proposition.
\begin{prop} \label{prop:prox-regularity-rank-one-set}
The rank-one set $\mathcal{R}_{1}$ in \eqref{eq:rank-one-set} is prox-regular at all points with rank exactly equal to one.
\end{prop}
\begin{proof}
Consider a full-rank matrix $\bm{X}=\bm{X}_{1}+\bm{X}_{\perp}\in\mathbb{R}^{M\times N}$ $(M\leq N)$ and its SVD $\bm{X}=\sum_{j=1}^{M}\sigma_{j}\bm{U}_{:,j}\bm{V}_{:,j}^{T}$, where $\bm{X}_{1}=\sigma_{1}\bm{U}_{:,1}\bm{V}_{:,1}^{T}$ and $\bm{X}_{\perp}=\sum_{j=2}^{M}\sigma_{j}\bm{U}_{:,j}\bm{V}_{:,j}^{T}$. It is straightforward to have $\bm{X}_{\perp}\in\mathcal{N}_{\mathcal{R}_1}(\bm{X}_1)$.
According to Eckart-Young-Mirsky theorem, we have
\begin{equation*}
    \bm{X}_1=\underset{\bm{X}^{'}\in\mathcal{R}_1}{\arg\min}\;\Vert\bm{X}^{'}-(\bm{X}_{1}+\bm{X}_{\perp})\Vert_{F}^{2}.
\end{equation*}
As $\sigma_{1}>\sigma_{2}>\cdots>\sigma_{M}$, for any $\gamma\geq 1$, we always have 
\begin{equation*}
    \bm{X}_1=\underset{\bm{X}^{'}\in\mathcal{R}_1}{\arg\min}\;\Vert\bm{X}^{'}-(\bm{X}_{1}+\gamma^{-1}\bm{X}_{\perp})\Vert_{F}^{2},
\end{equation*}
that is $\bm{X}_1$ is the unique nearest point in $\mathcal{R}_{1}$ to $\bm{X}_1+\gamma^{-1}\bm{X}_{\perp}$. Hence, $\mathcal{R}_1$ is prox-regular at all points with rank exactly equal to one.
\end{proof}
To further obtain the prox-regularity of indicator function $\mathbbm{1}_{\mathcal{R}_{1}}$ and its proximal subdifferential $\partial_{p}\mathbbm{1}_{\mathcal{R}_{1}}$, we invoke the proposition in \cite[Proposition 2.11.]{poliquin1996prox}.
\begin{prop} \cite[Proposition 2.11.]{poliquin1996prox} \label{prop:prox-regularity-indicator-function}
A set $\mathcal{C}$ is prox-regular at a point $\bar{x}\in\mathcal{C}$ for $\bar{v}$ if and only if its indicator function $\mathbbm{1}_{\mathcal{C}}$ is prox-regular at $\bar{x}$ for $\bar{v}$, and $\bar{v}\in\mathcal{N}_{\mathcal{C}}(\bar{x})=\partial_{p}\mathbbm{1}_{\mathcal{C}}$.
\end{prop}
Hence, the indicator function $\mathbbm{1}_{\mathcal{R}_{1}}(\bm{X})$ is prox-regular, and its prox-subdifferential is formulated as
\begin{equation}
    \partial_{p}\mathbbm{1}_{\mathcal{R}_{1}}(\bm{X})=\mathcal{N}_{\mathcal{R}_{1}}(\bm{X}).
\end{equation}

Consider the ADMM iterations (\ref{eq:ADMM-X1-update}-\ref{eq:ADMM-Z3-update}) of updating the variables $\bm{X}_{1,k}$, $\bm{X}_{2,k}$, $\bm{X}_{3,k}$, $\bm{Z}_{1,k}$, $\bm{Z}_{2,k}$ and $\bm{Z}_{3}$, we have the first order optimalities as
\begin{subequations}
\begin{align}
\bm{0} & \in\partial\Vert\bm{X}_{1,k}^{l+1}\Vert_{2,1}+\rho_1(\bm{X}_{1,k}^{l+1}-\bm{X}_{3,k}^{l}-\bm{Z}_{1,k}^{l}+\bm{\Lambda}_{1,k}^{l}),\label{eq:1stOdrOpt-X1-update} \\
\bm{0} & \in\partial_{p}\mathbbm{1}_{\mathcal{R}_{1}}(\bm{X}_{2,k}^{l+1})+\rho_2(\bm{X}_{2,k}^{l+1}-\bm{X}_{3,k}^{l}-\bm{Z}_{2,k}^{l}+\bm{\Lambda}_{2,k}^{l})\label{eq:1stOdrOpt-X2-update} \\
\bm{0} & \in\rho_1(\bm{X}_{k}^{l+1}-\bm{X}_{3,k}^{l+1}-\bm{Z}_{1,k}^{l}+\bm{\Lambda}_{1,k}^{l})+\rho_2(\bm{X}_{2,k}^{l+1}-\bm{X}_{3,k}^{l+1} \nonumber\\
& -\bm{Z}_{2,k}^{l}+\bm{\Lambda}_{2,k}^{l})+\rho_3(\bm{Y}-\sum_{k}\bm{X}_{3,k}^{l+1}-\bm{Z}_{3}^{l}+\bm{\Lambda}_{3}^{l}), \label{eq:1stOdrOpt-X3-update} \\
\bm{0} & \in \beta_1\bm{Z}_{1,k}^{l+1}-\rho_1(\bm{X}_{1,k}^{l+1}-\bm{X}_{3,k}^{l+1}-\bm{Z}_{1,k}^{l+1}+\bm{\Lambda}_{1,k}^{l})  \nonumber\\
& = \beta_1\bm{Z}_{1,k}^{l+1}-\rho_1\bm{\Lambda}_{1,k}^{l+1},\label{eq:1stOdrOpt-Z1-update}  \\
\bm{0} & \in \beta_2\bm{Z}_{2,k}^{l+1}-\rho_2(\bm{X}_{2,k}^{l+1}-\bm{X}_{3,k}^{l+1}-\bm{Z}_{2,k}^{l+1}+\bm{\Lambda}_{2,k}^{l})  \nonumber\\
& = \beta_2\bm{Z}_{2,k}^{l+1}-\rho_2\bm{\Lambda}_{2,k}^{l+1},\label{eq:1stOdrOpt-Z2-update}  \\
\bm{0} & \in \beta_3\bm{Z}_{3}^{l+1}-\rho_3(\bm{Y}-\sum_k\bm{X}_{3,k}^{l+1}-\bm{Z}_{3}^{l+1}+\bm{\Lambda}_{3}^{l})  \nonumber\\
& = \beta_3\bm{Z}_{3}^{l+1}-\rho_3\bm{\Lambda}_{3}^{l+1},\label{eq:1stOdrOpt-Z3-update}
\end{align}
\end{subequations}
for all $k\in[K]$. 

Recall the augmented Lagrangian $\mathcal{L}_{\rho 's}$ of inexact ADMM form of ROAD with different penalty parameters $\rho 's$ in \eqref{eq:Augmented-Lagrangian}. To prove sufficient descent of $\mathcal{L}_{\rho 's}$, we look into the difference in $\mathcal{L}_{\rho 's}$ after updating each variable.
We firstly focus on the steps of updating variables $\bm{X}^l$'s. As $\bm{X}_{1,k}^{l+1}$ is the optimal solution to the subproblem \eqref{eq:ADMM-X1-update}, we have
\begin{align}   \label{eq:LagDif-X1-Update}
&\mathcal{L}_{\rho 's}\left(\bm{X}_{1,k}^{l},\bm{X}_{2,k}^{l},\bm{X}_{3,k}^{l},\bm{Z}^{l}{'}s,\bm{\Lambda}^{l}{'}s\right) \nonumber\\
&-\mathcal{L}_{\rho 's}\left(\bm{X}_{1,k}^{l+1},\bm{X}_{2,k}^{l},\bm{X}_{3,k}^{l},\bm{Z}^{l}{'}s,\bm{\Lambda}^{l}{'}s\right)\nonumber\\
=&\sum_{k}(\Vert\bm{X}_{1,k}^{l}\Vert_{2,1}-\Vert\bm{X}_{1,k}^{l+1}\Vert_{2,1}+\frac{\rho_1}{2}\Vert\bm{X}_{1,k}^{l}-\bm{X}_{1,k}^{l+1}\Vert_{F}^{2}  \nonumber\\
+&\rho_1{\rm tr}\langle\bm{X}_{1,k}^{l+1}-\bm{X}_{3,k}^{l}-\bm{Z}_{1,k}^{l}+\bm{\Lambda}_{1,k}^{l},\bm{X}_{1,k}^{l}-\bm{X}_{1,k}^{l+1}\rangle \nonumber\\
=& \sum_{k}(\Vert\bm{X}_{1,k}^{l}\Vert_{2,1}-\Vert\bm{X}_{1,k}^{l+1}\Vert_{2,1}+\frac{\rho_1}{2}\Vert\bm{X}_{1,k}^{l}-\bm{X}_{1,k}^{l+1}\Vert_{F}^{2}  \nonumber\\
-&{\rm tr}\langle\partial\Vert\bm{X}_{1,k}^{l+1}\Vert_{2,1},\bm{X}_{1,k}^{l}-\bm{X}_{1,k}^{l+1}\rangle \nonumber\\
\geq&\frac{\rho_1}{2}\sum_{k}\Vert\bm{X}_{1,k}^{l}-\bm{X}_{1,k}^{l+1}\Vert_{F}^{2}, 
\end{align}
where the first equality follows the cosine rule: $\Vert \bm{B}+\bm{C}\Vert_{F}^{2}-\Vert \bm{A}+\bm{C}\Vert_{F}^{2}=\Vert \bm{B}-\bm{A}\Vert_{F}^{2}+2{\rm tr}\langle\bm{A}+\bm{C},\bm{B}-\bm{A}\rangle$, the second equality holds due to the first order optimalities \eqref{eq:1stOdrOpt-X1-update}, and the last inequality holds for the convexity  of $\Vert\cdot\Vert_{2,1}$.

As the indicator function of rank-one set $\mathbbm{1}_{\mathcal{R}_{1}}(\cdot)$ is lower-semicontinuous, and $\bm{X}_{2,k}^{l+1}$ minimizes the subproblem \eqref{eq:ADMM-X2-update}, it is straightforward to obtain that
\begin{align}\label{eq:LagDif-X2-Update}
&\mathcal{L}_{\rho 's}\left(\bm{X}_{1,k}^{l},\bm{X}_{2,k}^{l},\bm{X}_{3,k}^{l},\bm{Z}^{l}{'}s,\bm{\Lambda}^{l}{'}s\right) \nonumber\\
&-\mathcal{L}_{\rho 's}\left(\bm{X}_{1,k}^{l+1},\bm{X}_{2,k}^{l+1},\bm{X}_{3,k}^{l},\bm{Z}^{l}{'}s,\bm{\Lambda}^{l}{'}s\right)\geq 0 
\end{align}

Similarly, for $\bm{X}_{3}^{l+1}$ is the minimum of the subproblem \eqref{eq:ADMM-X3-update}, we have
\begin{align}  \label{eq:LagDif-X3-Update}
&\mathcal{L}_{\rho 's}\left(\bm{X}_{1,k}^{l+1},\bm{X}_{2,k}^{l+1},\bm{X}_{3,k}^{l},\bm{Z}^{l}{'}s,\bm{\Lambda}^{l}{'}s\right) \nonumber\\
&-\mathcal{L}_{\rho 's}\left(\bm{X}_{1,k}^{l+1},\bm{X}_{2,k}^{l+1},\bm{X}_{3,k}^{l+1},\bm{Z}^{l}{'}s,\bm{\Lambda}^{l}{'}s\right)\nonumber\\
=& \frac{\rho_1+\rho_2}{2}\sum_k\Vert\bm{X}_{3,k}^{l}-\bm{X}_{3,k}^{l+1}\Vert_{F}^{2}+2\sum_k{\rm tr}\langle\rho_1(\bm{X}_{1,k}^{l+1}-\bm{X}_{3,k}^{l+1}\nonumber\\
-& \bm{Z}_{1,k}^{l}+ \bm{\Lambda}_{1,k}^{l})+\rho_2(\bm{X}_{2,k}^{l+1}-\bm{X}_{3,k}^{l+1}-\bm{Z}_{2,k}^{l}+\bm{\Lambda}_{2,k}^{l}), \nonumber\\
& \bm{X}_{3,k}^{l}-\bm{X}_{3,k}^{l+1}\rangle) + \frac{\rho_3}{2}\Vert\sum_k\bm{X}_{3,k}^{l}-\sum_k\bm{X}_{3,k}^{l+1}\Vert_{F}^{2} \nonumber\\
+&\rho_3{\rm tr}\langle\bm{Y}-\sum_k\bm{X}_{3,k}^{l+1}-\bm{Z}_{3}^{l}+\bm{\Lambda}_{3}^{l},\sum_k\bm{X}_{3,k}^{l}-\sum_k\bm{X}_{3,k}^{l+1}\rangle \nonumber\\
=& \frac{\rho_1+\rho_2}{2}\sum_k\Vert\bm{X}_{3,k}^{l}-\bm{X}_{3,k}^{l+1}\Vert_{F}^{2}+\frac{\rho_3}{2}\Vert\sum_k(\bm{X}_{3,k}^{l}-\bm{X}_{3,k}^{l+1})\Vert_{F}^{2},
\end{align}
where the last equality follows the first order optimality \eqref{eq:1stOdrOpt-X3-update}. To ensure the sufficient decent of Augmented Lagrangian after each step of updating the variables $\bm{X}^{l}$'s, it requires that $\frac{\rho_1}{2}>1$, $\frac{\rho_2}{2}>1$ and $\frac{\rho_3}{2}>1$, that is $\rho_1>2$, $\rho_2>2$ and $\rho_3>2$.

Considering the ADMM steps of updating variables $\bm{Z}$'s in (\ref{eq:ADMM-Z1-update}-\ref{eq:ADMM-Z3-update}) and computing the difference of Lagrangian after these steps, we have
\begin{align}  \label{eq:LagDif-Z-Update}
&\mathcal{L}_{\rho 's}\left(\bm{X}^{l+1}{'}s,\bm{Z}_{1,k}^{l},\bm{Z}_{2,k}^{l},\bm{Z}_{3}^{l},\bm{\Lambda}^{l}{'}s\right) \nonumber\\
&-\mathcal{L}_{\rho 's}\left(\bm{X}^{l+1}{'}s,\bm{Z}_{1,k}^{l+1},\bm{Z}_{2,k}^{l+1},\bm{Z}_{3}^{l+1},\bm{\Lambda}^{l}{'}s\right)\nonumber\\
=&\sum_k(\frac{\beta_1}{2}\Vert\bm{Z}_{1,k}^{l}\Vert_{F}^{2}-\frac{\beta_1}{2}\Vert\bm{Z}_{1,k}^{l+1}\Vert_{F}^{2}+\frac{\rho_1}{2}\Vert\bm{Z}_{1,k}^{l}-\bm{Z}_{1,k}^{l+1}\Vert_{F}^{2} \nonumber\\
-&\rho_1{\rm tr}\langle \bm{X}_{1,k}^{l+1}-\bm{X}_{3,k}^{l+1}-\bm{Z}_{1,k}^{l+1}+\bm{\Lambda}_{1,k}^{l},\bm{Z}_{1,k}^{l}-\bm{Z}_{1,k}^{l+1} \rangle  \nonumber\\
+& \frac{\beta_2}{2}\Vert\bm{Z}_{2,k}^{l}\Vert_{F}^{2}-\frac{\beta_2}{2}\Vert\bm{Z}_{2,k}^{l+1}\Vert_{F}^{2}+\frac{\rho_2}{2}\Vert\bm{Z}_{2,k}^{l}-\bm{Z}_{2,k}^{l+1}\Vert_{F}^{2} \nonumber\\
-&\rho_2{\rm tr}\langle \bm{X}_{2,k}^{l+1}-\bm{X}_{3,k}^{l+1}-\bm{Z}_{2,k}^{l+1}+\bm{\Lambda}_{2,k}^{l},\bm{Z}_{2,k}^{l}-\bm{Z}_{2,k}^{l+1} \rangle ) \nonumber\\
+& \frac{\beta_3}{2}\Vert\bm{Z}_{3}^{l}\Vert_{F}^{2}-\frac{\beta_3}{2}\Vert\bm{Z}_{3}^{l+1}\Vert_{F}^{2}+\frac{\rho_3}{2}\Vert\bm{Z}_{3}^{l}-\bm{Z}_{3}^{l+1}\Vert_{F}^{2} \nonumber\\
-&\rho_3{\rm tr}\langle \bm{Y}-\sum_k\bm{X}_{3,k}^{l+1}-\bm{Z}_{3}^{l+1}+\bm{\Lambda}_{1,k}^{l},\bm{Z}_{3}^{l}-\bm{Z}_{3}^{l+1} \rangle  \nonumber\\
=& \sum_k(\frac{\beta_1}{2}\Vert\bm{Z}_{1,k}^{l}\Vert_{F}^{2}-\frac{\beta_1}{2}\Vert\bm{Z}_{1,k}^{l+1}\Vert_{F}^{2}+\frac{\rho_1}{2}\Vert\bm{Z}_{1,k}^{l}-\bm{Z}_{1,k}^{l+1}\Vert_{F}^{2} \nonumber\\
-&\rho_1{\rm tr}\langle\bm{\Lambda}_{1,k}^{l+1},\bm{Z}_{1,k}^{l}-\bm{Z}_{1,k}^{l+1} \rangle +\frac{\beta_2}{2}\Vert\bm{Z}_{2,k}^{l}\Vert_{F}^{2}-\frac{\beta_2}{2}\Vert\bm{Z}_{2,k}^{l+1}\Vert_{F}^{2} \nonumber\\
+&\frac{\rho_2}{2}\Vert\bm{Z}_{2,k}^{l}-\bm{Z}_{2,k}^{l+1}\Vert_{F}^{2}-\rho_2{\rm tr}\langle\bm{\Lambda}_{2,k}^{l+1},\bm{Z}_{2,k}^{l}-\bm{Z}_{2,k}^{l+1} \rangle) +\frac{\beta_3}{2}\Vert\bm{Z}_{3}^{l}\Vert_{F}^{2} \nonumber\\
-&\frac{\beta_3}{2}\Vert\bm{Z}_{3}^{l+1}\Vert_{F}^{2}+\frac{\rho_3}{2}\Vert\bm{Z}_{3}^{l}-\bm{Z}_{3}^{l+1}\Vert_{F}^{2} -\rho_3{\rm tr}\langle\bm{\Lambda}_{3}^{l+1},\bm{Z}_{3}^{l}-\bm{Z}_{3}^{l+1} \rangle) \nonumber\\
=&\frac{\beta_1 +\rho_1}{2}\sum_k\Vert\bm{Z}_{1,k}^{l}-\bm{Z}_{1,k}^{l+1}\Vert_{F}^{2}+\frac{\beta_2 +\rho_2}{2}\sum_k\Vert\bm{Z}_{2,k}^{l}-\bm{Z}_{2,k}^{l+1}\Vert_{F}^{2} \nonumber\\
+&\frac{\beta_3 +\rho_3}{2}\Vert\bm{Z}_{3}^{l}-\bm{Z}_{3}^{l+1}\Vert_{F}^{2},
\end{align}
where the second last equality holds for the updating rules of $\bm{\Lambda}'s$ and the last equality results from the first order optimalities (\ref{eq:1stOdrOpt-Z1-update}-\ref{eq:1stOdrOpt-Z3-update}).

Now we look at the difference in augmented Lagrangian after $\bm{\Lambda}' s$ updating steps. We have
\begin{align}  \label{eq:LagDif-Lambda-Update}
&\mathcal{L}_{\rho 's}\left(\bm{X}^{l+1}{'}s,\bm{Z}^{l+1}{'}s,\bm{\Lambda}_{1,k}^{l},\bm{\Lambda}_{2,k}^{l},\bm{\Lambda}_{3}^{l}\right) \nonumber\\
&-\mathcal{L}_{\rho 's}\left(\bm{X}^{l+1}{'}s,\bm{Z}^{l+1}{'}s,\bm{\Lambda}_{1,k}^{l+1},\bm{\Lambda}_{2,k}^{l+1},\bm{\Lambda}_{3}^{l+1}\right)\nonumber\\
=& \frac{\rho_1}{2}\sum_k\Vert\bm{\Lambda}_{1,k}^{l}-\bm{\Lambda}_{1,k}^{l+1}\Vert_{F}^{2} +\rho_1\sum_k{\rm tr}\langle \bm{X}_{1,k}^{l+1}-\bm{X}_{3,k}^{l+1}-\bm{Z}_{1,k}^{l+1} \nonumber\\
+& \bm{\Lambda}_{1,k}^{l+1},\bm{\Lambda}_{1,k}^{l}-\bm{\Lambda}_{1,k}^{l+1} \rangle + \frac{\rho_2}{2}\sum_k\Vert\bm{\Lambda}_{2,k}^{l}-\bm{\Lambda}_{2,k}^{l+1}\Vert_{F}^{2} \nonumber\\
+& \rho_2\sum_k{\rm tr}\langle \bm{X}_{2,k}^{l+1}-\bm{X}_{3,k}^{l+1}-\bm{Z}_{2,k}^{l+1}
+ \bm{\Lambda}_{2,k}^{l+1},\bm{\Lambda}_{2,k}^{l}-\bm{\Lambda}_{2,k}^{l+1} \rangle \nonumber\\
+& \frac{\rho_3}{2}\Vert\bm{\Lambda}_{3}^{l}-\bm{\Lambda}_{3}^{l+1}\Vert_{F}^{2} +\rho_3\langle \bm{Y}-\sum_k\bm{X}_{3,k}^{l+1}-\bm{Z}_{3}^{l+1} \nonumber\\
+& \bm{\Lambda}_{3}^{l+1},\bm{\Lambda}_{3}^{l}-\bm{\Lambda}_{3}^{l+1} \rangle  \nonumber\\
=& \frac{\rho_1}{2}\sum_k(\Vert\bm{\Lambda}_{1,k}^{l}-\bm{\Lambda}_{1,k}^{l+1}\Vert_{F}^{2} +\rho_1{\rm tr}\langle 2\bm{\Lambda}_{1,k}^{l+1}-\bm{\Lambda}_{1,k}^{l},\bm{\Lambda}_{1,k}^{l}-\bm{\Lambda}_{1,k}^{l+1} \rangle \nonumber\\
+& \frac{\rho_2}{2}\Vert\bm{\Lambda}_{2,k}^{l}-\bm{\Lambda}_{2,k}^{l+1}\Vert_{F}^{2}+ \rho_2{\rm tr}\langle 2\bm{\Lambda}_{2,k}^{l+1}-\bm{\Lambda}_{2,k}^{l},\bm{\Lambda}_{2,k}^{l}-\bm{\Lambda}_{2,k}^{l+1} \rangle ) \nonumber\\
+& \frac{\rho_3}{2}\Vert\bm{\Lambda}_{3}^{l}-\bm{\Lambda}_{3}^{l+1}\Vert_{F}^{2} +\rho_3{\rm tr}\langle 2\bm{\Lambda}_{3}^{l+1}-\bm{\Lambda}_{3}^{l},\bm{\Lambda}_{3}^{l}-\bm{\Lambda}_{3}^{l+1} \rangle  \nonumber\\
=& -\rho_1\sum_k\Vert\bm{\Lambda}_{1,k}^{l}-\bm{\Lambda}_{1,k}^{l+1}\Vert_{F}^{2}-\rho_2\sum_k\Vert\bm{\Lambda}_{2,k}^{l}-\bm{\Lambda}_{2,k}^{l+1}\Vert_{F}^{2} \nonumber\\
-& \rho_3\Vert\bm{\Lambda}_{3}^{l}-\bm{\Lambda}_{3}^{l+1}\Vert_{F}^{2}
\end{align}
According to the first order optimalities (\ref{eq:1stOdrOpt-Z1-update}-\ref{eq:1stOdrOpt-Z3-update}) w.r.t. $\bm{Z}^{l+1}{'}s$, it can be directly derived that $\bm{\Lambda}_{1,k}^{l+1}=\frac{\beta_1}{\rho}\bm{Z}_{1,k}^{l+1}$, $\bm{\Lambda}_{2,k}^{l+1}=\frac{\beta_2}{\rho}\bm{Z}_{2,k}^{l+1}$ and $\bm{\Lambda}_{3}^{l+1}=\frac{\beta_3}{\rho}\bm{Z}_{3}^{l+1}$ for all $k\in [K]$ and $l$. Adding \eqref{eq:LagDif-Z-Update}-\eqref{eq:LagDif-Lambda-Update} together, we have
\begin{align}
&\mathcal{L}_{\rho 's}\left(\bm{X}^{l+1}{'}s,\bm{Z}_{1,k}^{l},\bm{Z}_{2,k}^{l},\bm{Z}_{3}^{l},\bm{\Lambda}_{1,k}^{l},\bm{\Lambda}_{2,k}^{l},\bm{\Lambda}_{3}^{l}\right) \nonumber\\
&-\mathcal{L}_{\rho 's}\left(\bm{X}^{l+1}{'}s,\bm{Z}_{1,k}^{l+1},\bm{Z}_{2,k}^{l+1},\bm{Z}_{3}^{l+1},\bm{\Lambda}_{1,k}^{l+1},\bm{\Lambda}_{2,k}^{l+1},\bm{\Lambda}_{3}^{l+1}\right)\nonumber\\  
=&\frac{\rho_1-\beta_1}{2}\sum_k\Vert\bm{Z}_{1,k}^{l}-\bm{Z}_{1,k}^{l+1}\Vert_{F}^{2}+\frac{\rho_2-\beta_2 }{2}\sum_k\Vert\bm{Z}_{2,k}^{l}-\bm{Z}_{2,k}^{l+1}\Vert_{F}^{2} \nonumber\\
+&\frac{\rho_3-\beta_3 }{2}\Vert\bm{Z}_{3}^{l}-\bm{Z}_{3}^{l+1}\Vert_{F}^{2}.
\end{align}
Hence, sufficient descent of augmented Lagrangian requires $\frac{\rho_1-\beta_1}{2}>1$, $\frac{\rho_2-\beta_2}{2}>1$ and $\frac{\rho_3-\beta_3}{2}>1$, or equivalently, $\rho_1>\beta_1+2$, $\rho_2>\beta_2+2$ and $\rho_3>\beta_3+2$, respectively.

In conclusion, when the lower bounds of the penalty parameters $\rho' s$ satisfy  $\rho_1>\beta_1+2$, $\rho_2>\beta_2+2$ and $\rho_3>\beta_3+2$ respectively, the augmented Lagrangian $\mathcal{L}_{\rho' s}$ descents in each iteration of $l$.

\subsubsection{Convergence of $\mathcal{L}_{\rho 's}$}  

After we prove that $\mathcal{L}_{\rho 's}$ descends every iteration, we can now show the convergence of $\mathcal{L}_{\rho 's}$. 

\begin{lem}\label{lem:Lagranian-converges}
$\mathcal{L}_{\rho}\left(\bm{X}^{l}{\rm 's},\bm{Z}^{l}{\rm 's},\bm{\Lambda}^{l}{\rm 's}\right)$ is lower bounded for all $k\in\mathbb{N}$ and converges as $l\rightarrow\infty$.
\end{lem}

\begin{proof}
There exist $\bm{Z}_{1,k}^{'}$, $\bm{Z}_{2,k}^{'}$ and $\bm{Z}_{3}^{'}$ such that
 $\bm{Z}_{1,k}^{'} = \bm{X}_{1,k}^{l}-\bm{X}_{3,k}^{l}$, $\bm{Z}_{2,k}^{'} = \bm{X}_{2,k}^{l}-\bm{X}_{3,k}^{l},\;\forall k\in [K]$ and $\sum_k \bm{X}_{3,k}^{l}+\bm{Z}_{3}^{'}=\bm{Y}$. Then we have
\begin{align}
 & \mathcal{L}_{\rho 's}\left(\bm{X}^{l}{\rm 's},\bm{Z}^{l}{\rm 's},\bm{\Lambda}^{l}{\rm 's}\right) \nonumber\\
= & \sum_{k}\left( \Vert\bm{X}_{1,k}^{l}\Vert_{2,1} + \frac{\rho_1}{2} \Vert \bm{Z}_{1,k}^{'}-\bm{Z}_{1,k}^{l}\Vert _{F}^{2} \right) \nonumber\\
+ & \sum_{k}\left( \mathbbm{1}_{\mathcal{R}_{1}}(\bm{X}_{2,k}^{l})+ \frac{\rho_2}{2} \Vert  \bm{Z}_{2,k}^{'}-\bm{Z}_{2,k}^{l}\Vert _{F}^{2} \right) 
+ \frac{\rho_3}{2}\Vert\bm{Z}_{3}^{'}-\bm{Z}_{3}^{l} \Vert_{F}^{2}\nonumber\\
+& \sum_k\left(\frac{\beta_1}{2}\Vert\bm{Z}_{1,k}^{l}\Vert_{F}^{2}+{\rm tr}\langle \rho_1\bm{\Lambda}_{1,k}^{l},\bm{Z}_{1,k}^{'}-\bm{Z}_{1,k}^{l}\rangle\right) \nonumber\\
+& \sum_k\left(\frac{\beta_2}{2}\Vert\bm{Z}_{2,k}^{l}\Vert_{F}^{2}+{\rm tr}\langle \rho_2\bm{\Lambda}_{2,k}^{l},\bm{Z}_{2,k}^{'}-\bm{Z}_{2,k}^{l}\rangle\right) \nonumber\\
+ &\frac{\beta_3}{2}\Vert\bm{Z}_{3}^{l}\Vert_{F}^{2} + {\rm tr}\langle\rho_3\bm{\Lambda}_{3}^{l},\bm{Z}_{3}^{'}-\bm{Z}_{3}^{l}\rangle \nonumber\\
= & \sum_{k}\left( \Vert\bm{X}_{1,k}^{l}\Vert_{2,1} + \frac{\rho_1}{2} \Vert \bm{Z}_{1,k}^{'}-\bm{Z}_{1,k}^{l}\Vert _{F}^{2} \right) \nonumber\\
+ & \sum_{k}\left( \mathbbm{1}_{\mathcal{R}_{1}}(\bm{X}_{2,k}^{l})+ \frac{\rho_2}{2} \Vert  \bm{Z}_{2,k}^{'}-\bm{Z}_{2,k}^{l}\Vert _{F}^{2} \right) 
+ \frac{\rho_3}{2}\Vert\bm{Z}_{3}^{'}-\bm{Z}_{3}^{l} \Vert_{F}^{2}\nonumber\\
+& \sum_k\left(\frac{\beta_1}{2}\Vert\bm{Z}_{1,k}^{l}\Vert_{F}^{2}+{\rm tr}\langle \partial(\frac{\beta_1}{2}\Vert\bm{Z}_{1,k}^{l}\Vert_{F}^{2}),\bm{Z}_{1,k}^{'}-\bm{Z}_{1,k}^{l}\rangle\right) \nonumber\\
+& \sum_k\left(\frac{\beta_2}{2}\Vert\bm{Z}_{2,k}^{l}\Vert_{F}^{2}+{\rm tr}\langle \partial(\frac{\beta_2}{2}\Vert\bm{Z}_{2,k}^{l}\Vert_{F}^{2}),\bm{Z}_{2,k}^{'}-\bm{Z}_{2,k}^{l}\rangle\right) \nonumber\\
+ &\frac{\beta_3}{2}\Vert\bm{Z}_{3}^{l}\Vert_{F}^{2} + {\rm tr}\langle\partial(\frac{\beta_3}{2}\Vert\bm{Z}_{3}^{l}\Vert_{F}^{2}),\bm{Z}_{3}^{'}-\bm{Z}_{3}^{l}\rangle \nonumber\\
\geq & \sum_{k}\left( \Vert\bm{X}_{1,k}^{l}\Vert_{2,1} + \frac{\rho_1}{2} \Vert \bm{Z}_{1,k}^{'}-\bm{Z}_{1,k}^{l}\Vert _{F}^{2} \right) \nonumber\\
+ & \sum_{k}\left( \mathbbm{1}_{\mathcal{R}_{1}}(\bm{X}_{2,k}^{l})+ \frac{\rho_2}{2} \Vert  \bm{Z}_{2,k}^{'}-\bm{Z}_{2,k}^{l}\Vert _{F}^{2} \right) 
+ \frac{\rho_3}{2}\Vert\bm{Z}_{3}^{'}-\bm{Z}_{3}^{l} \Vert_{F}^{2}\nonumber\\
+ & \sum_k\left(\frac{\beta_1}{2}\Vert\bm{Z}_{1,k}^{'}\Vert_{F}^{2}+\frac{\beta_2}{2}\Vert\bm{Z}_{2,k}^{'}\Vert_{F}^{2}\right)+\frac{\beta_3}{2}\Vert\bm{Z}_{3}^{'}\Vert_{F}^{2} \nonumber\\
\geq & 0,
\end{align}
where the second equality holds for the first order optimalities of updating the variables $\bm{Z}_{1,k}^{l}$, $\bm{Z}_{2,k}^{l}$ and $\bm{Z}_{3}^{l}$, and the last second inequality follows the convexity of Frobenius norm. For $\mathcal{L}_{\rho 's}\left(\bm{X}^{l}{\rm 's},\bm{Z}^{l}{\rm 's},\bm{\Lambda}^{l}{\rm 's}\right)$ is lower bounded and descends every step, it converges as $l\rightarrow\infty$.
\end{proof}

\subsubsection{$\partial\mathcal{L}_{\rho 's}$ converges to zero}

Consider $\partial\mathcal{L}_{\rho 's}$ as
\begin{align}
\partial\mathcal{L}_{\rho 's}\left(\bm{X}^{l}{\rm 's},\bm{Z}^{l}{\rm 's},\bm{\Lambda}^{l}{\rm 's}\right)=\left(\frac{\partial\mathcal{L}_{\rho 's}}{\partial \bm{X}^{l}{\rm 's}},\frac{\partial\mathcal{L}_{\rho 's}}{\partial \bm{Z}^{l}{\rm 's}},\frac{\partial\mathcal{L}_{\rho 's}}{\partial \bm{\Lambda}^{l}{\rm 's}}\right), \label{eq:explicit-expression-pratial-rho} 
\end{align}
where
\begin{subequations}
\begin{align}
&\frac{\partial\mathcal{L}_{\rho 's}}{\partial \bm{X}_{1,k}^{l}}=\partial\Vert\bm{X}_{1,k}^{l}\Vert_{2,1}+\rho_1(\bm{X}_{1,k}^{l}-\bm{X}_{3,k}^{l}-\bm{Z}_{1,k}^{l}+\bm{\Lambda}_{1,k}^{l}), \\
&\frac{\partial\mathcal{L}_{\rho 's}}{\partial \bm{X}_{2,k}^{l}}=\partial_{p}\mathbbm{1}_{\mathcal{R}_{1}}(\bm{X}_{2,k}^{l})+\rho_2(\bm{X}_{2,k}^{l}-\bm{X}_{3,k}^{l}-\bm{Z}_{2,k}^{l}+\bm{\Lambda}_{2,k}^{l}), \\
&\frac{\partial\mathcal{L}_{\rho 's}}{\partial \bm{X}_{3,k}^{l}}=-\rho_1(\bm{X}_{1,k}^{l}-\bm{X}_{3,k}^{l}-\bm{Z}_{1,k}^{l}+\bm{\Lambda}_{1,k}^{l})-\rho_2(\bm{X}_{2,k}^{l} \nonumber\\
&-\bm{X}_{3,k}^{l}-\bm{Z}_{2,k}^{l}+\bm{\Lambda}_{2,k}^{l})-\rho_3(\bm{Y}-\sum_k\bm{X}_{3,k}^{l}-\bm{Z}_{3}^{l}+\bm{\Lambda}_{3}^{l}), \\
&\frac{\partial\mathcal{L}_{\rho 's}}{\partial \bm{Z}_{1,k}^{l}}=\beta_1\bm{Z}_{1,k}^{l}-\rho_1(\bm{X}_{1,k}^{l}-\bm{X}_{3,k}^{l}-\bm{Z}_{1,k}^{l}+\bm{\Lambda}_{1,k}^{l}), \\
&\frac{\partial\mathcal{L}_{\rho 's}}{\partial \bm{Z}_{2,k}^{l}}=\beta_2\bm{Z}_{2,k}^{l}-\rho_2(\bm{X}_{2,k}^{l}-\bm{X}_{3,k}^{l}-\bm{Z}_{2,k}^{l}+\bm{\Lambda}_{2,k}^{l}), \\
&\frac{\partial\mathcal{L}_{\rho 's}}{\partial \bm{Z}_{3}^{l}}=\beta_3\bm{Z}_{3}^{l}-\rho_3(\bm{Y}-\sum_k\bm{X}_{3,k}^{l}-\bm{Z}_{3}^{l}+\bm{\Lambda}_{3}^{l}), \\
&\frac{\partial\mathcal{L}_{\rho 's}}{\partial \bm{\Lambda}_{1,k}^{l}}=\rho_1(\bm{X}_{1,k}^{l}-\bm{X}_{3,k}^{l}-\bm{Z}_{1,k}^{l}), \\
&\frac{\partial\mathcal{L}_{\rho 's}}{\partial \bm{\Lambda}_{2,k}^{l}}=\rho_2(\bm{X}_{2,k}^{l}-\bm{X}_{3,k}^{l}-\bm{Z}_{2,k}^{l}),\\
&\frac{\partial\mathcal{L}_{\rho 's}}{\partial
\bm{\Lambda}_{3}^{l}}=\rho_3(\bm{Y}-\sum_k\bm{X}_{3,k}^{l}-\bm{Z}_{3}^{l}).
\end{align}
\end{subequations}

Recall Lemma \ref{lem:Lagranian-descends} regarding to the sufficient descent of $\mathcal{L}_{\rho 's}$. As $\mathcal{L}_{\rho 's}$ converges as $l\rightarrow\infty$, we have
\begin{align}
&\lim_{l\rightarrow\infty}\Vert\bm{X}^{l}{\rm 's}-\bm{X}^{l-1}{\rm 's}\Vert_{F}^{2}=0,   \nonumber\\ 
&\lim_{l\rightarrow\infty}\Vert\bm{Z}^{l}{\rm 's}-\bm{Z}^{l-1}{\rm 's}\Vert_{F}^{2}=0,  \nonumber\\ \label{eq:limit-dual-residual}
&\lim_{l\rightarrow\infty}\Vert\bm{\Lambda}^{l}{\rm 's}-\bm{\Lambda}^{l-1}{\rm 's}\Vert_{F}^{2}=0.
\end{align}
By the steps of updating the Lagrange multipliers $\bm{\Lambda}'s$ in (\ref{eq:ADMM-Lambda1k-Update}-\ref{eq:ADMM-Lambda3-Update}), we can further obtain that the primal residual converges to $\bm{0}$, that is
\begin{align} \label{eq:limit-primal-residual}
&\lim_{l\rightarrow\infty}\Vert\bm{X}_{1,k}^{l}-\bm{X}_{3,k}^{l}-\bm{Z}_{1,k}^{l}\Vert_{F}^{2}=0,  \nonumber\\  
&\lim_{l\rightarrow\infty}\Vert\bm{X}_{2,k}^{l}-\bm{X}_{3,k}^{l}-\bm{Z}_{2,k}^{l}\Vert_{F}^{2}=0,  \nonumber\\ 
&\lim_{l\rightarrow\infty}\Vert\bm{Y}-\sum_k\bm{X}_{3,k}^{l}-\bm{Z}_{3}^{l}\Vert_{F}^{2}=0.
\end{align}
Consider the first order optimal condition of ADMM updating steps (\ref{eq:1stOdrOpt-X1-update}-\ref{eq:1stOdrOpt-Z3-update}), as $l\rightarrow\infty$, it can be derived that 
\begin{align*}
\bm{0} & \in\partial\Vert\bm{X}_{1,k}^{l+1}\Vert_{2,1}+\rho_1(\bm{X}_{1,k}^{l+1}-\bm{X}_{3,k}^{l}-\bm{Z}_{1,k}^{l}+\bm{\Lambda}_{1,k}^{l}),\\
& = \partial\Vert\bm{X}_{1,k}^{l+1}\Vert_{2,1}+\rho_1(\bm{X}_{1,k}^{l+1}-\bm{X}_{3,k}^{l+1}-\bm{Z}_{1,k}^{l+1}+\bm{\Lambda}_{1,k}^{l+1}) \\
&=\frac{\partial\mathcal{L}_{\rho 's}}{\partial \bm{X}_{1,k}^{l+1}}, 
\\
\bm{0} & \in\partial_{p}\mathbbm{1}_{\mathcal{R}_{1}}(\bm{X}_{2,k}^{l+1})+\rho_2(\bm{X}_{2,k}^{l+1}-\bm{X}_{3,k}^{l}-\bm{Z}_{2,k}^{l}+\bm{\Lambda}_{2,k}^{l}) \\
& = \partial_{p}\mathbbm{1}_{\mathcal{R}_{1}}(\bm{X}_{2,k}^{l+1})+\rho_2(\bm{X}_{2,k}^{l+1}-\bm{X}_{3,k}^{l+1}-\bm{Z}_{2,k}^{l+1}+\bm{\Lambda}_{2,k}^{l+1})   \\
& = \frac{\partial\mathcal{L}_{\rho 's}}{\partial \bm{X}_{2,k}^{l+1}}, 
\\
\bm{0} & \in \rho_1(\bm{X}_{1,k}^{l+1}-\bm{X}_{3,k}^{l+1}-\bm{Z}_{1,k}^{l}+\bm{\Lambda}_{1,k}^{l})+\rho_2(\bm{X}_{2,k}^{l+1}-\bm{X}_{3,k}^{l+1} \\
-& \bm{Z}_{2,k}^{l}+\bm{\Lambda}_{2,k}^{l})+\rho_3(\bm{Y}-\sum_{k}\bm{X}_{3,k}^{l+1}-\bm{Z}_{3}^{l}+\bm{\Lambda}_{3}^{l})  \\
& = \rho_1(\bm{X}_{1,k}^{l+1}-\bm{X}_{3,k}^{l+1}-\bm{Z}_{1,k}^{l+1}+\bm{\Lambda}_{1,k}^{l+1})+\rho_2(\bm{X}_{2,k}^{l+1}-\bm{X}_{3,k}^{l+1} \\
- &\bm{Z}_{2,k}^{l+1}+\bm{\Lambda}_{2,k}^{l+1})+\rho_3(\bm{Y}-\sum_{k}\bm{X}_{3,k}^{l+1}-\bm{Z}_{3}^{l+1}+\bm{\Lambda}_{3}^{l+1})  \\
& = -\frac{\partial\mathcal{L}_{\rho 's}}{\rho\partial \bm{X}_{3,k}^{l+1}}, 
\\
\bm{0} & \in \beta_1\bm{Z}_{1,k}^{l+1}-\rho_1(\bm{X}_{1,k}^{l+1}-\bm{X}_{3,k}^{l+1}-\bm{Z}_{1,k}^{l+1}+\bm{\Lambda}_{1,k}^{l})  \\
& = \beta_1\bm{Z}_{1,k}^{l+1}-\rho_1(\bm{X}_{1,k}^{l+1}-\bm{X}_{3,k}^{l+1}-\bm{Z}_{1,k}^{l+1}+\bm{\Lambda}_{1,k}^{l+1})  \\
& = \frac{\partial\mathcal{L}_{\rho 's}}{\partial \bm{Z}_{1,k}^{l+1}},
\\
\bm{0} & \in \beta_2\bm{Z}_{2,k}^{l+1}-\rho_2(\bm{X}_{2,k}^{l+1}-\bm{X}_{3,k}^{l+1}-\bm{Z}_{2,k}^{l+1}+\bm{\Lambda}_{2,k}^{l})  \\
& = \beta_2\bm{Z}_{2,k}^{l+1}-\rho_2(\bm{X}_{2,k}^{l+1}-\bm{X}_{3,k}^{l+1}-\bm{Z}_{2,k}^{l+1}+\bm{\Lambda}_{2,k}^{l+1}) \\
& = \frac{\partial\mathcal{L}_{\rho 's}}{\partial \bm{Z}_{2,k}^{l+1}},
\\
\bm{0} & \in \beta_3\bm{Z}_{3}^{l+1}-\rho_3(\bm{Y}-\sum_k\bm{X}_{3,k}^{l+1}-\bm{Z}_{3}^{l+1}+\bm{\Lambda}_{3}^{l})  \\
& = \beta_3\bm{Z}_{3}^{l+1}-\rho_3(\bm{Y}-\sum_k\bm{X}_{3,k}^{l+1}-\bm{Z}_{3}^{l+1}+\bm{\Lambda}_{3}^{l+1}) \\
& = \frac{\partial\mathcal{L}_{\rho 's}}{\partial \bm{Z}_{3}^{l+1}},
\end{align*}
Hence, it can be directly obtained that as $l\rightarrow\infty$,
$\bm{0}\in\partial\mathcal{L}_{\rho 's}\left(\{\bm{X}^{l+1}{\rm 's},\bm{Z}^{l+1}{\rm 's},\bm{\Lambda}^{l+1}{\rm 's}\right)$.

\subsection{Stopping Criteria}\label{subsec:stopping}

This section gives the stopping criteria of our ROAD algorithm. According to Theorem \ref{thm:convergence-ROAD}, one way to terminate the algorithm is that $\bm{0}\in\partial\mathcal{L}_{\rho 's}\left(\{\bm{X}^{l+1}{\rm 's},\bm{Z}^{l+1}{\rm 's},\bm{\Lambda}^{l+1}{\rm 's}\right)$.

The conditions \eqref{eq:subdif-Lambdas} is straightforward. However,
the conditions \eqref{eq:subdif-Pk} and \eqref{eq:subdif-Qk} are related to the subdifferential of $\ell_{2,1}$ norm and the proximal subdifferential of $\mathbbm{1}_{\mathcal{R}_{1}}$, which are not straightforward to determine.

Inspired by \cite[\textsection 3.3]{boyd2011distributed} which is also introduced in subsection \ref{subsec:ADMM-background}, we define primal residual $\bm{r}^{l+1}=\bm{A}\bm{x}^{l+1}+\bm{B}\bm{z}^{l+1}-\bm{c}^{l+1}$ in standard form of ADMM \eqref{eq:ADMM-inexact-StandardForm}, or the matrix form $\bm{R}^{l+1}=[[\cdots,(\bm{X}_{1,k}^{l+1}-\bm{X}_{3,k}^{l+1}-\bm{Z}_{1,k}^{l+1}),\cdots],[\cdots,(\bm{X}_{2,k}^{l+1}-\bm{X}_{3,k}^{l+1}-\bm{Z}_{2,k}^{l+1}),\cdots],[\bm{Y}-\sum_k\bm{X}_{3,k}^{l+1}-\bm{Z}_{3}^{l+1}]]\in\mathbb{R}^{M\times (2NK+N)}
$ at iteration $l+1$. We also define dual residual $\bm{s}^{l+1}=\bm{B}\bm{z}^{l+1}-\bm{B}\bm{z}^{l}$, or the matrix form $\bm{S}^{l+1}=[[\cdots,(\bm{Z}_{1,k}^{l+1}-\bm{Z}_{2,k}^{l}),\cdots],[\cdots,(\bm{Z}_{2,k}^{l+1}-\bm{Z}_{2,k}^{l}),\cdots],[\bm{Z}_{3}^{l+1}-\bm{Z}_{3}^{l}]]\in\mathbb{R}^{M\times 2NK+N}$ at iteration $l+1$. By \eqref{eq:limit-dual-residual} and \eqref{eq:limit-primal-residual}, we have $\lim_{l\rightarrow\infty}\Vert\bm{R}^{l}\Vert_{F}^{2}=0$ and $\lim_{l\rightarrow\infty}\Vert\bm{S}^{l}\Vert_{F}^{2}=0$. Accordingly, the stopping criteria are that
\begin{equation}
    \frac{\Vert\bm{R}^{l}\Vert_{F}^{2}}{\Vert[\cdots,\bm{X}_{3,k}^{l},\cdots]\Vert_{F}^{2}}\leq\epsilon^{\rm primal}\;{\rm and}\;\frac{\Vert\bm{S}^{l}\Vert_{F}^{2}}{\Vert[\cdots,\bm{Z}_{3,k}^{l},\cdots]\Vert_{F}^{2}}\leq\epsilon^{\rm dual},
\end{equation}
where $\epsilon^{\rm primal}>0$ and $\epsilon^{\rm dual}>0$ are the tolerances, and a reasonable value is suggested to be $10^{-6}$.

\section{Numerical Tests}\label{sec:numerical-tests}

This section compares the numerical performance of ROAD with other benchmark dictionary learning algorithms including MOD, K-SVD, and BLOTLESS in different tests. The comparison in Section \ref{subsec:synthetic-data} is based on synthetic data while the results of real data tests are illustrated in Section \ref{subsec:real-data}. 

\subsection{Dictionary learning for synthetic data \label{subsec:synthetic-data}}

For synthetic data tests, we adopt the typical setting for data generation. 
We assume that the training data $\bm{Y}$ are generated from a ground-truth dictionary $\bm{D}^0$ and a ground-truth sparse coefficient matrix $\bm{X}^0$ via $\bm{Y}=\bm{D}^0 \bm{X}^0$.
The dictionary $\bm{D}^0$ is generated by first filling it with independent realizations of the standard Gaussian variable and then normalizing its columns to have unit $\ell_2$-norm. The sparse coefficients in $\bm{X}^0$ is generated as follows. Assume that the number of nonzero coefficients in the $n$-th column of $\bm{X}^0$ is $S_n$. The index set of the nonzero coefficients are randomly generated from the uniform distribution on $[K] \choose S_n$ and the values of the nonzero coefficients are independently generated from the standard Gaussian distribution. In our simulations, we set $S_{n} = S \in \mathbb{N}$, $\forall n\in [N]$.

Given the synthetic data, different dictionary learning algorithms are tested and compared using the criterion of dictionary recovery error. Consider the permutation ambiguity of the trained dictionary. The dictionary recovery error is defined as 
\begin{equation}
\mathrm{Error}\coloneqq\frac{1}{K}\sum_{k=1}^{K}(1-\mid\hat{\bm{D}}_{:,k}^{T}\bm{D}_{:,i_{k}}^{0}\mid),\label{eq:DRE}
\end{equation}
where $i_{k} \coloneqq{\arg\max}_{i\in\mathcal{I}_{k}}(\hat{\bm{D}}_{:,k}^{T}\bm{D}_{:,i}^{0})$, $\mathcal{I}_{k} \coloneqq[K]\backslash\{i_{1},\cdots,i_{k-1}\}$, $\hat{\bm{D}}_{:,k}$ denotes the $k$-th column of estimated
dictionary, and $\bm{D}_{:,i_{k}}^{0}$ represents the $i_{k}$-th
column of ground truth dictionary which has largest correlation with
$\hat{\bm{D}}_{k}$. The use of $\mathcal{I}_{k}$ is to avoid repetitions in $i_{k}$, $\forall k\in[K]$.

In the first test, we compare three different ROAD mechanisms including inexact ADMM (iADMM) using different penalty parameter $\rho$'s, inexact ADMM (iADMM) using fixed $\rho$ and exact ADMM (eADMM). Set $M=16$, $K=32$ and $S=3$, we test the performance of these three mechanisms in two ways. One is to fix the number of samples $N$ and to compare the number of iterations that the ADMM converges, while the other is to vary $N$ from 100 to 400 and to compare different levels of dictionary recovery error corresponding to different $N$. To ensure every method stops at a relatively stable point even with few samples, we set the stopping iteration number at 10000 for iADMM and at 300 for eADMM. We  repeat  each  test  for  100  trials  to  get  an  average  to acquire  a  preciser  result. In terms of choosing the values of parameters, we set $\beta_1=200$, $\beta_2=300$ and $\beta_3=250$ for both iADMM methods. For penalty parameter $\rho$, we set $\rho_1=210$, $\rho_2=310$ and $\rho_3=260$ for the method of using different $\rho$'s, and set $\rho=310$ for the approach using the same $\rho$. Note that there is no noise term for eADMM, which means we only need to choose one parameter $\rho$. Here we choose $\rho=10$ for eADMM, and we  keep this value for the remaining tests.
Figure \ref{fig:compare-iADMM-eADMM} depicts the significant superiority of eADMM both in convergence rate and error rate compared with two iADMM methods. It can be clearly noticed that when the number of measurements $N\geq 300$, the dictionary recovery error can converge to almost 0 in less than 300 iterations. While for both iADMM approaches, the error can only converge to a value that is slightly lower than $10^{-2}$ even at 10000 iterations.

\begin{figure}
\begin{centering}
\subfloat[$M=16$, $K=32$, $S=3$, $N=300$.]{\begin{centering}
\includegraphics[scale=0.27]{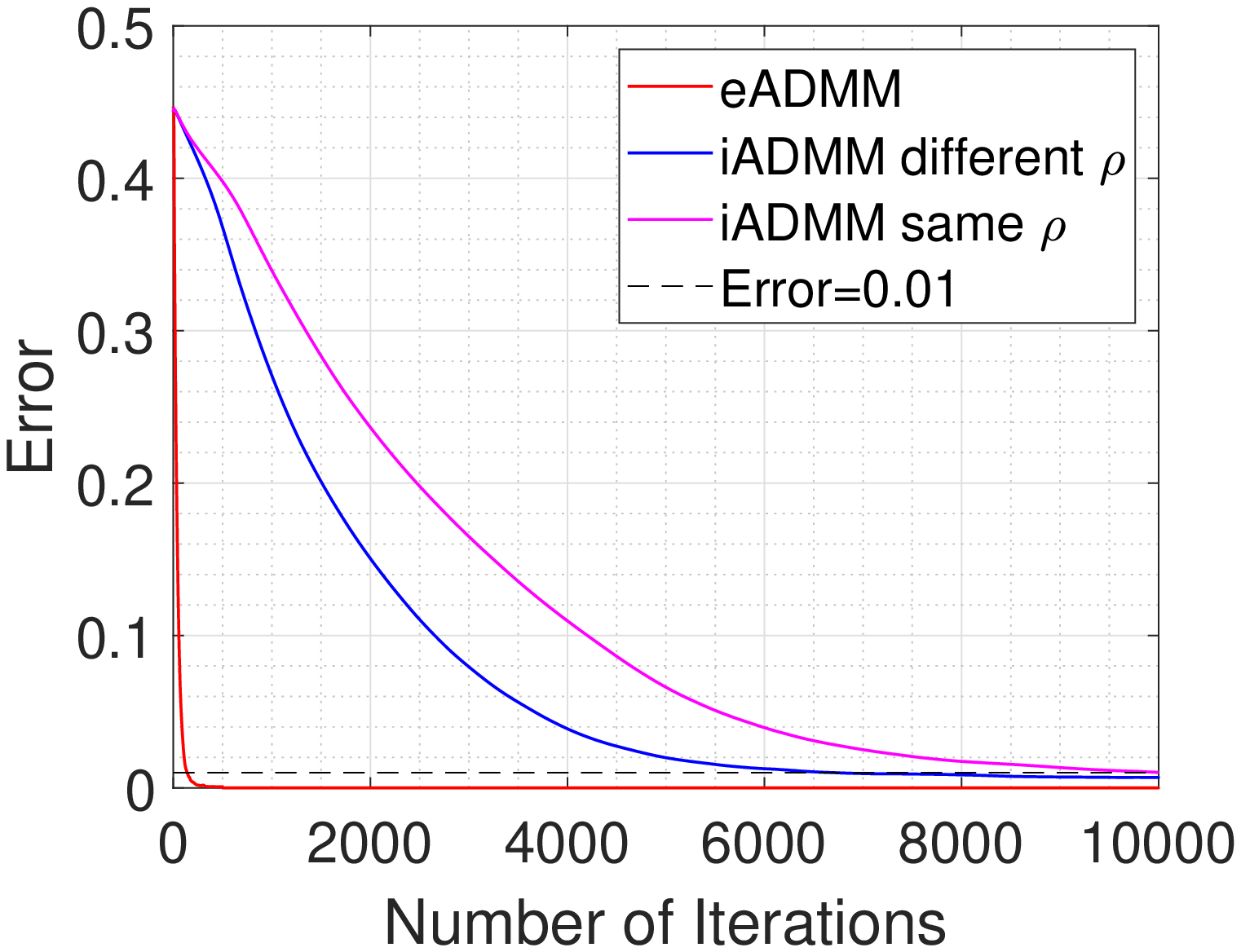}
\par\end{centering}
}\subfloat[$M=16$, $K=32$, $S=3$.]{\begin{centering}
\includegraphics[scale=0.25]{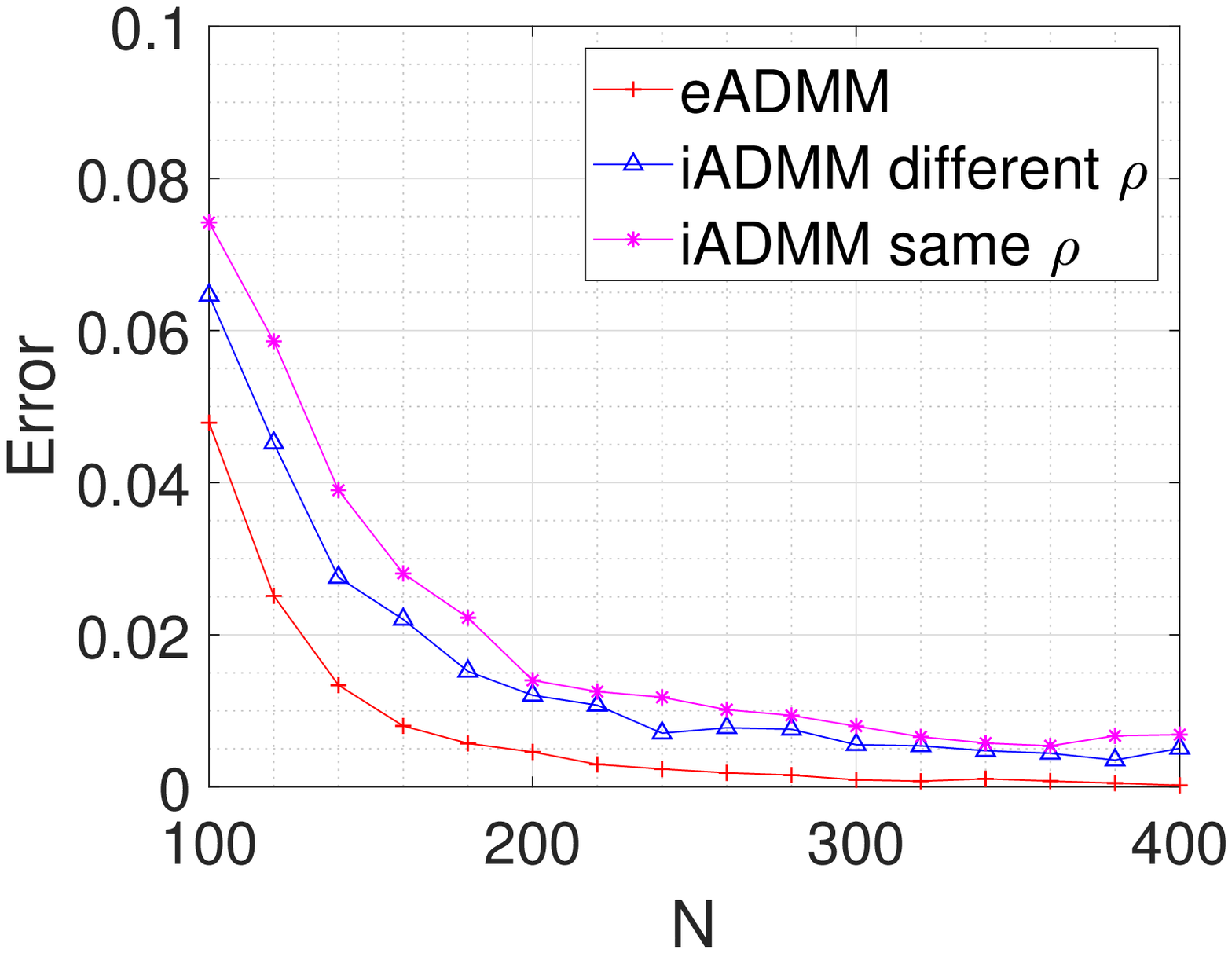}
\par\end{centering}
}
\par\end{centering}
\caption{\label{fig:compare-iADMM-eADMM}Comparison of dictionary learning methods including iADMM with different penalty parameter $\rho$'s, iADMM with the same $\rho$ and exact ADMM.}
\end{figure}

According to the advantages of eADMM, in the following performance tests compared with other benchmark algorithms, we utilize ROAD with eADMM. As all the other benchmark algorithms, including MOD, K-SVD and BLOTLESS, requires another sparse coding stage, OMP \cite{pati1993orthogonal} is used for the sparse coding stage with the prior knowledge of $S$. Note that different from other benchmark algorithms, ROAD does not need such prior information. 
Fig. \ref{fig:1} compares different dictionary learning methods for different sized dictionaries and sparsity levels. We use OMP as sparse coding stage for those two-stage dictionary learning approaches, and also in the other tests. We vary the number
of measurement to find its relation to the dictionary recovery error.
To make every algorithm stop at a stable point even with few samples,
we set the stopping iteration number at 300. We repeat each test for
100 trials to get an average to acquire a preciser result. Fig. \ref{fig:1}
illustrates that for different sized dictionaries and sparsity levels,
ROAD can always converge to almost zero dictionary recovery error,
and it requires fewer samples than other benchmarks. However for the
other dictionary learning methods, there is always a small error even
with sufficient number of samples. 

\begin{figure}
\begin{centering}
\subfloat[$M=16$, $K=32$, $S=3$.]{\begin{centering}
\includegraphics[scale=0.27]{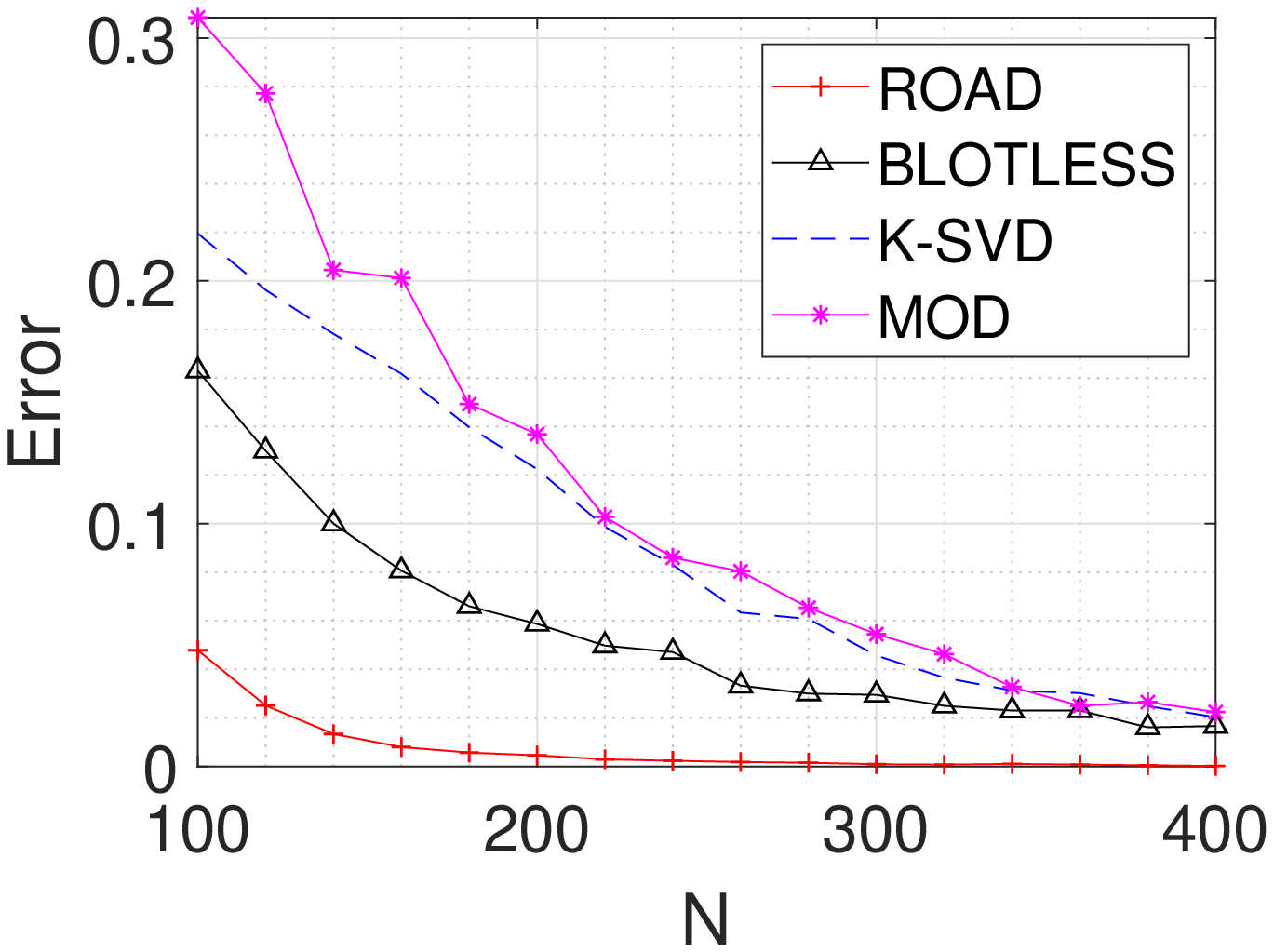}
\par\end{centering}
}\subfloat[$M=24$, $K=48$, $S=3$. ]{\begin{centering}
\includegraphics[scale=0.27]{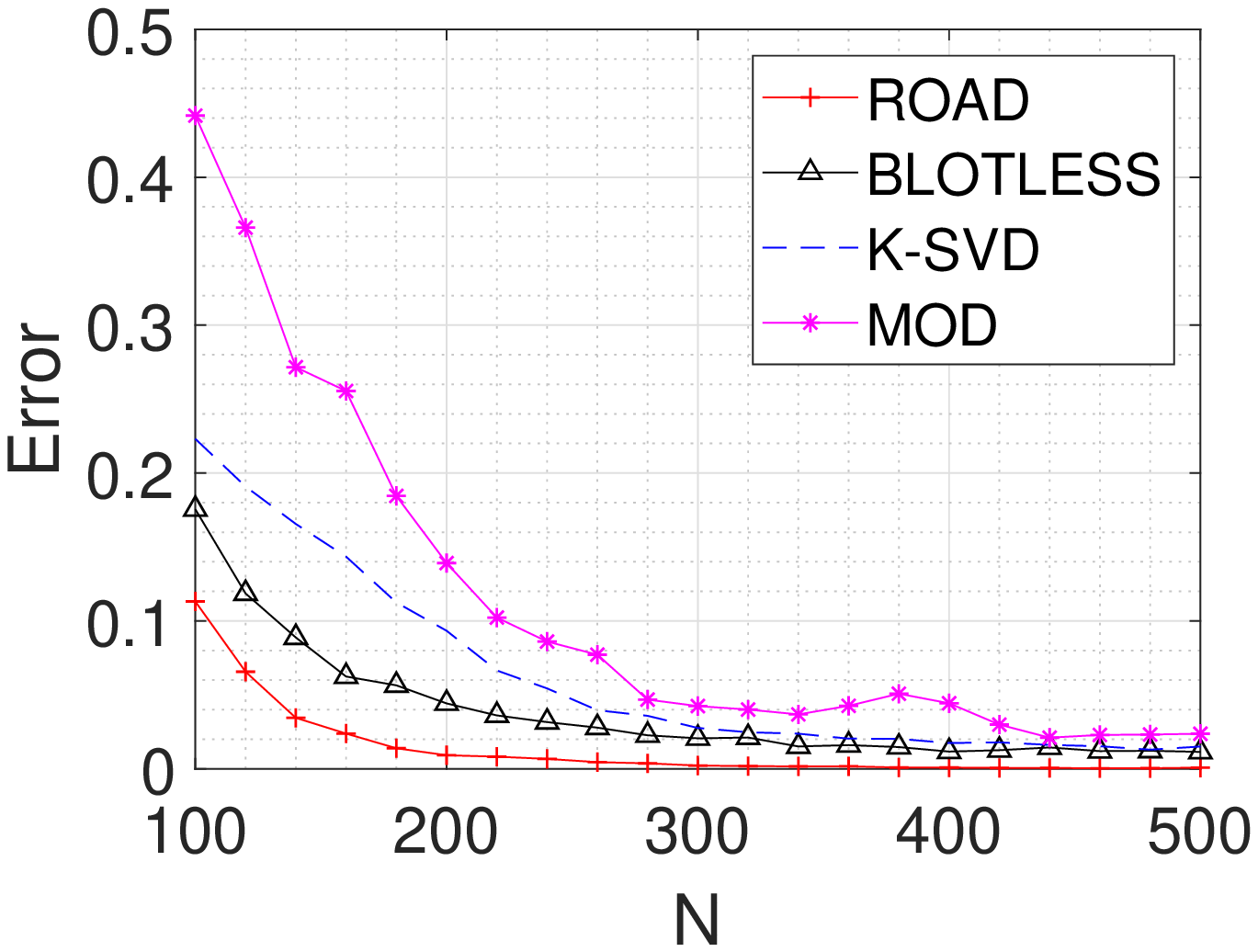}
\par\end{centering}
}
\par\end{centering}
\begin{centering}
\subfloat[$M=24$, $K=48$, $S=6$.]{\begin{centering}
\includegraphics[scale=0.27]{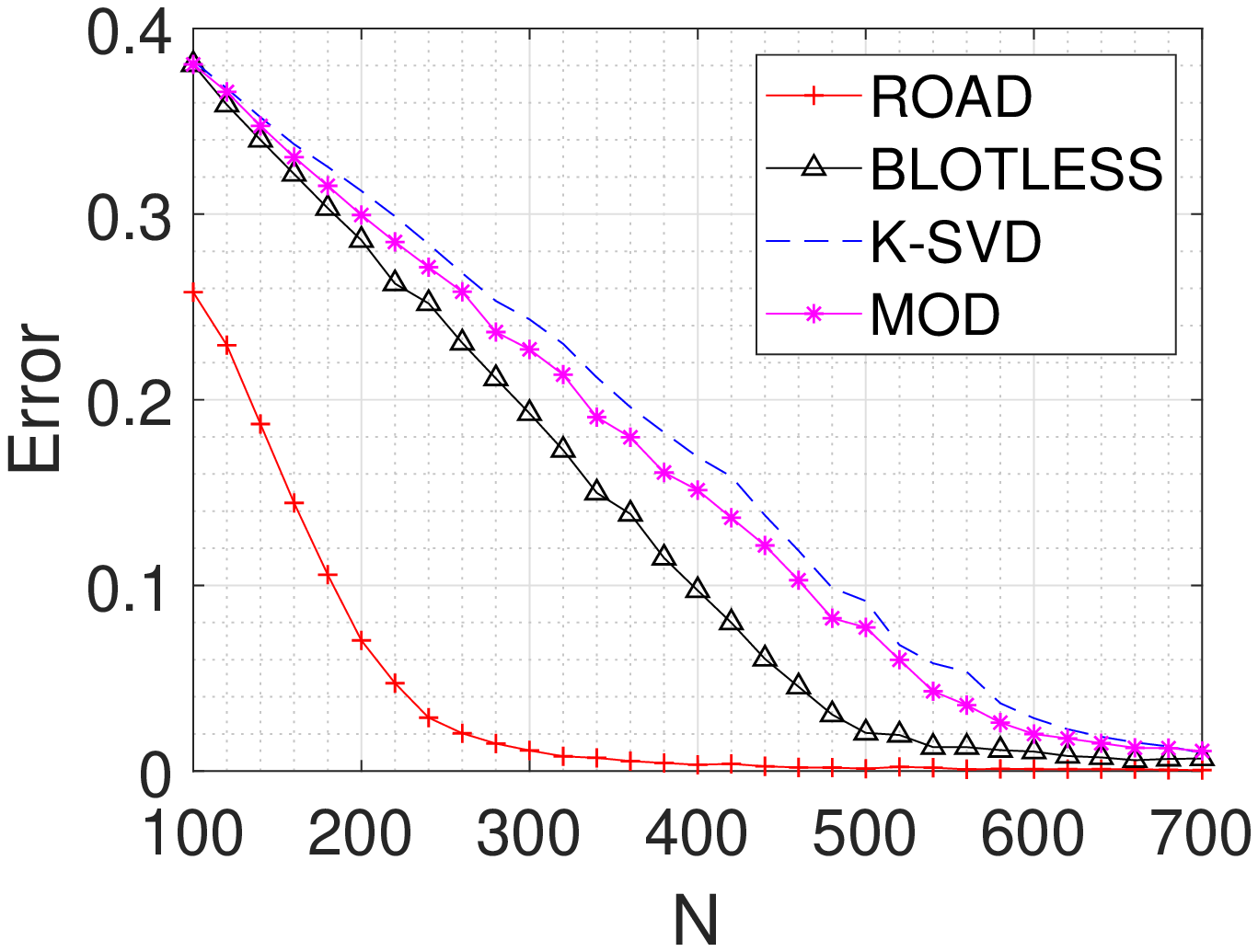}
\par\end{centering}
}\subfloat[$M=32$, $K=64$, $S=6$.]{\begin{centering}
\includegraphics[scale=0.27]{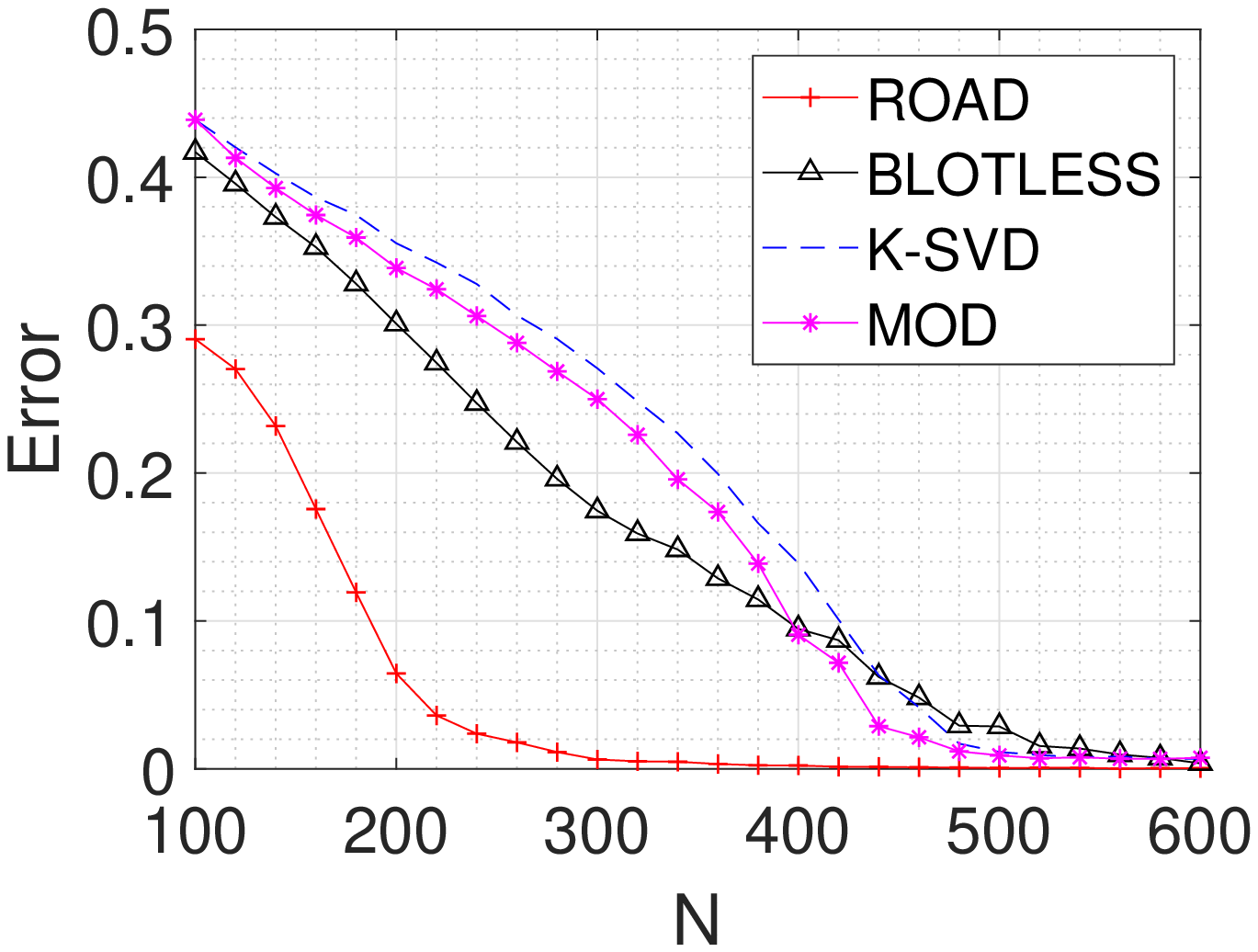}
\par\end{centering}
}
\par\end{centering}
\caption{\label{fig:1}Comparison of dictionary learning methods for the noise-free cases. Results are averages of 100 trials.}
\end{figure}

Fig. \ref{fig:2} compares different dictionary learning methods using the same settings as the first simulation in Fig. \ref{fig:1} but with noise, and SNR=30dB and 20dB respectively. The results in Fig. \ref{fig:2} demonstrate that for different noise level, ROAD also outperforms the other dictionary learning methods.

\begin{figure}
\begin{centering}
\subfloat[SNR=30dB.]{\begin{centering}
\includegraphics[scale=0.29]{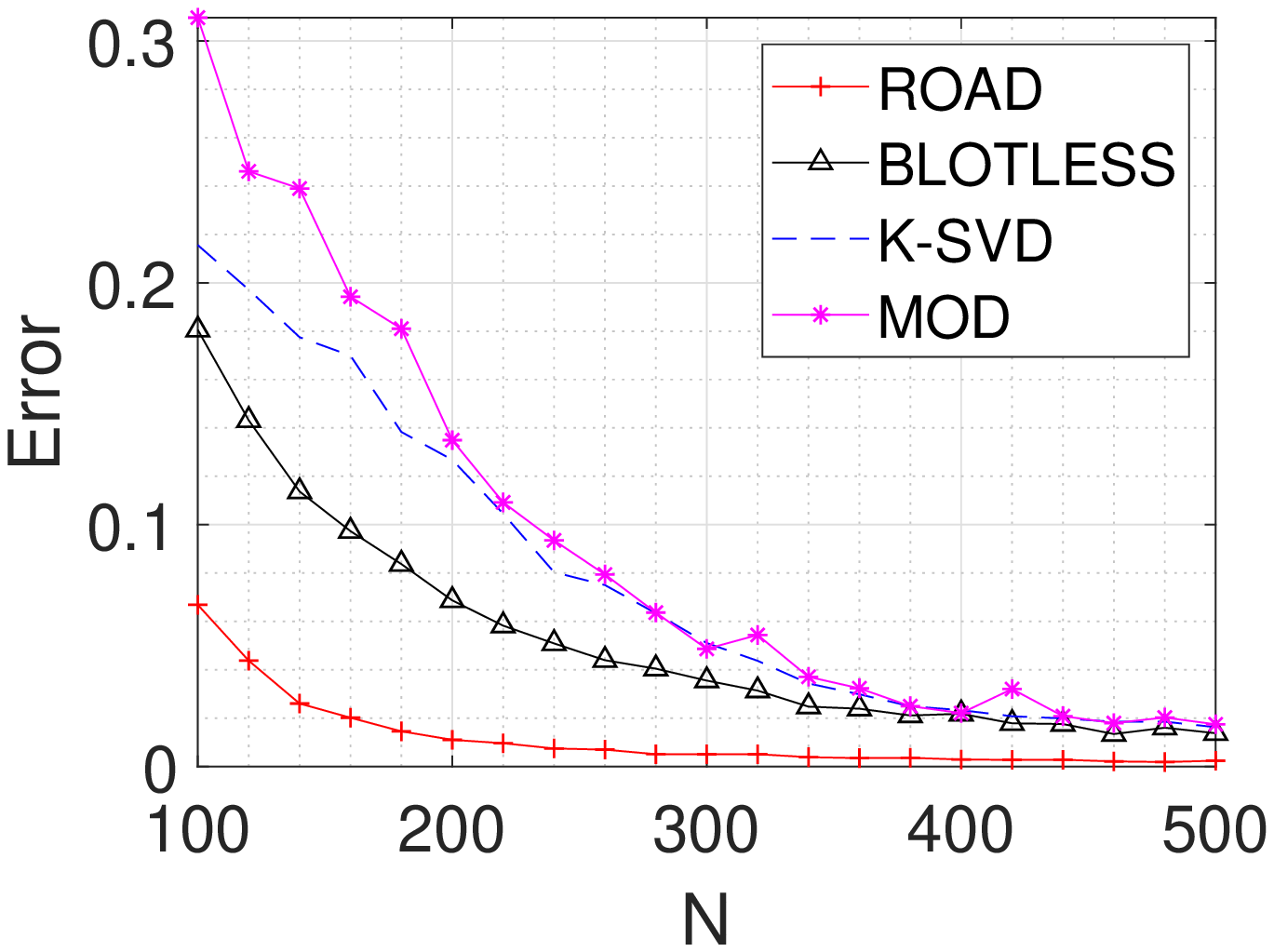}
\par\end{centering}
}\subfloat[SNR=20dB.]{\begin{centering}
\includegraphics[scale=0.29]{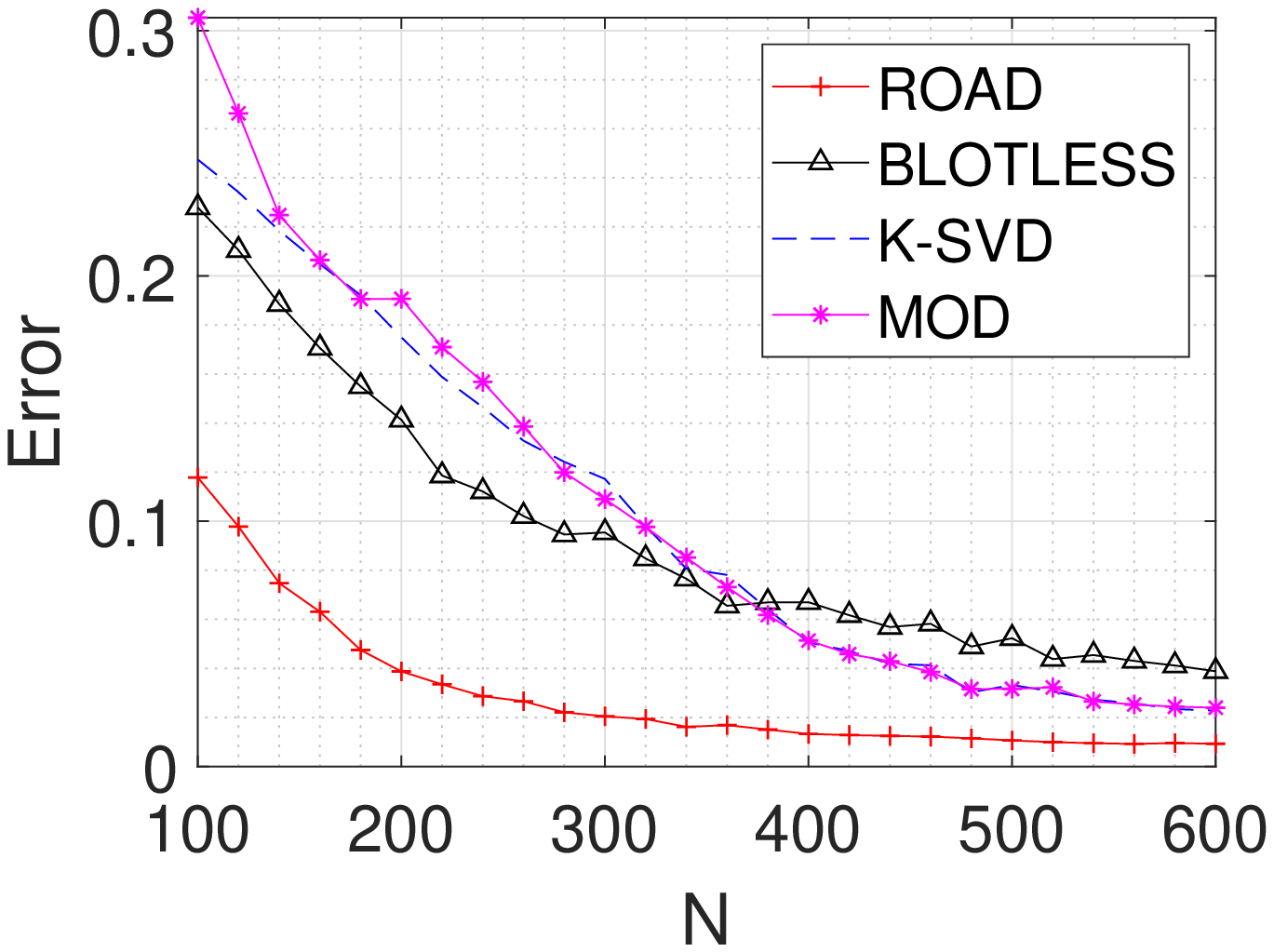}
\par\end{centering}
}
\par\end{centering}
\caption{\label{fig:2}Comparison of dictionary learning methods for the noisy
cases. Dictionary size is fixed to $M=$16, $K=$32, $S=$3. Results are averages of 100 trials.}
\end{figure}

In the previous simulations, we assume the sparsity level of each column of sparse coefficients is fixed. However, in the real data, the sparsity level is not fixed. To make the test more practical, we assume the entries of sparse coefficients are under Bernoulli distribution with the probability of $\theta$, and the values of non-zero entries are generated from independent Gaussian distribution $\mathscr{N}(0,1)$. Fig.\ref{fig:3} compares different dictionary learning methods for different sparsity levels but with the same dictionary size, where the sparse ratio $\theta$ is $3/48$ and $6/48$ respectively. For benchmark approaches, as we use OMP as sparse coding stage, it has high probability to obtain a poor sparsity pattern with not fixed the sparsity level, which will result in learning a wrong dictionary. Accordingly, for the sparse coding stage in benchmark methods, we set fixed sparsity level $S=3$ and $S=6$ respectively in these two tests. From Fig.\ref{fig:3}, it can be concluded that even with different sparsity levels, ROAD can still exactly recover the dictionary with enough samples, but the dictionaries learnt by other techniques have a significant difference with the ground truth dictionary especially the sparse ratio is high. 

\begin{figure}
\begin{centering}
\subfloat[$M=24$, $K=48$, $\theta=\frac{3}{48}$.]{\begin{centering}
\includegraphics[scale=0.23]{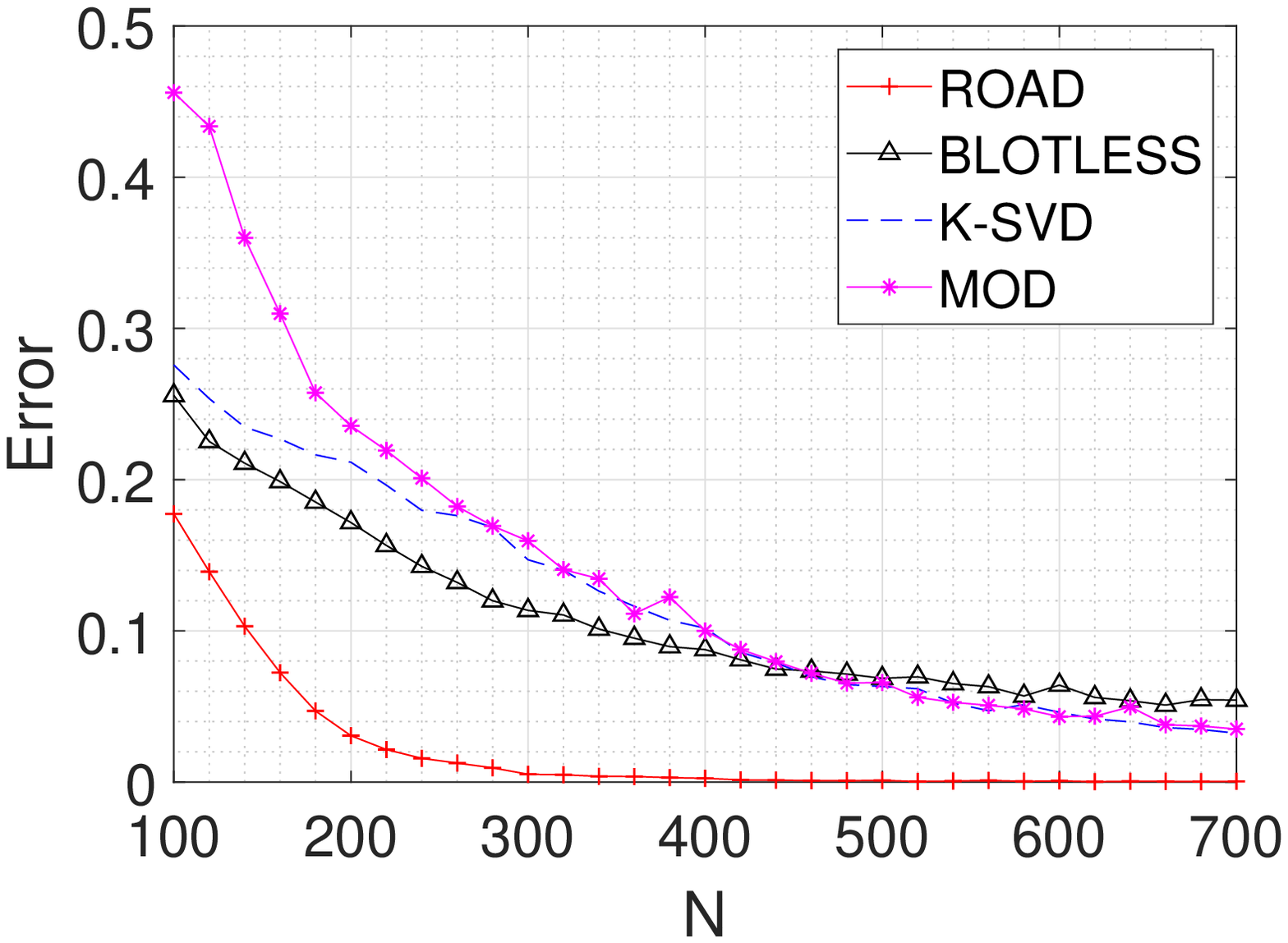}
\par\end{centering}
}\hspace{0.15cm}\subfloat[$M=24$, $K=48$, $\theta=\frac{6}{48}$.]{\begin{centering}
\includegraphics[scale=0.23]{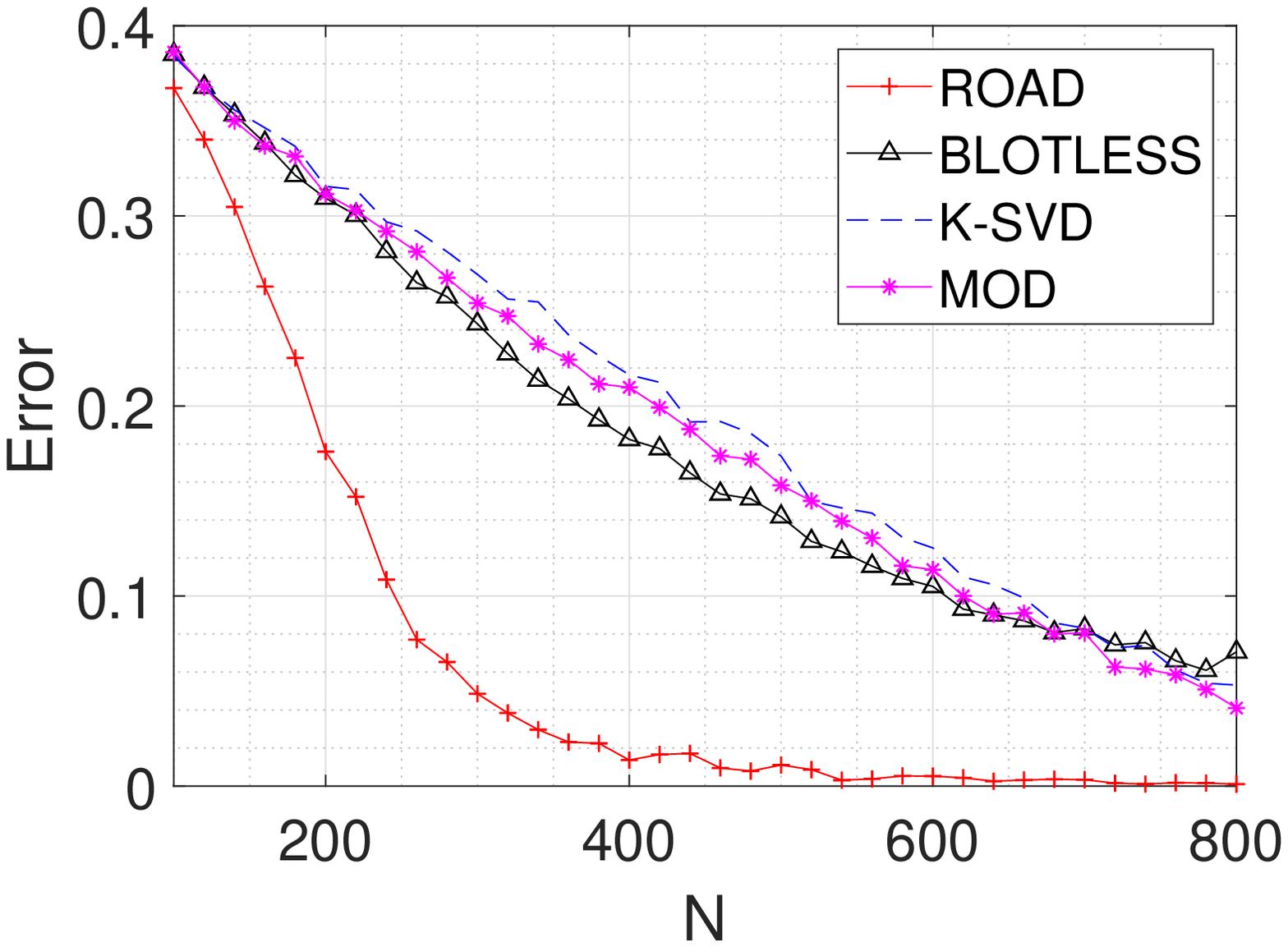}
\par\end{centering}
}
\par\end{centering}
\caption{\label{fig:3}Comparison of dictionary learning methods for different sparsity levels in the columns of sparse coefficient matrix.  Results are averages of 100 trials.}
\end{figure}

\subsection{Single image super-resolution using dictionary learning \label{subsec:real-data}}

\begin{table*}
\caption{\label{tab:1}Comparison of single image super-resolution using different
dictionary learning methods, where both the figures of super-resolution
results and the PSNR between the estimated high-resolution digits
and the ground truth digits are shown in the table.}
\renewcommand{\arraystretch}{2}
\newcommand\cincludegraphics[2][]{\raisebox{-0.3\height}{\includegraphics[#1]{#2}}}
\begin{centering}
\begin{tabular}{|C{3cm}| C{2cm}| C{1cm} C{1cm} C{1cm} C{1cm}| C{1cm} C{1cm} C{1cm} C{1cm}|}
\hline

Training images & Testing image  & \multicolumn{4}{c|}{Stacked dictionary learning \cite{yang2010image}} & \multicolumn{4}{c|}{Single dictionary learning \cite{zeyde2010single}} \\

\cline{3-10}

 & Bicubic & ROAD & Blotless & K-SVD & MOD & ROAD & Blotless & K-SVD & MOD \\

\cincludegraphics[scale=0.18]{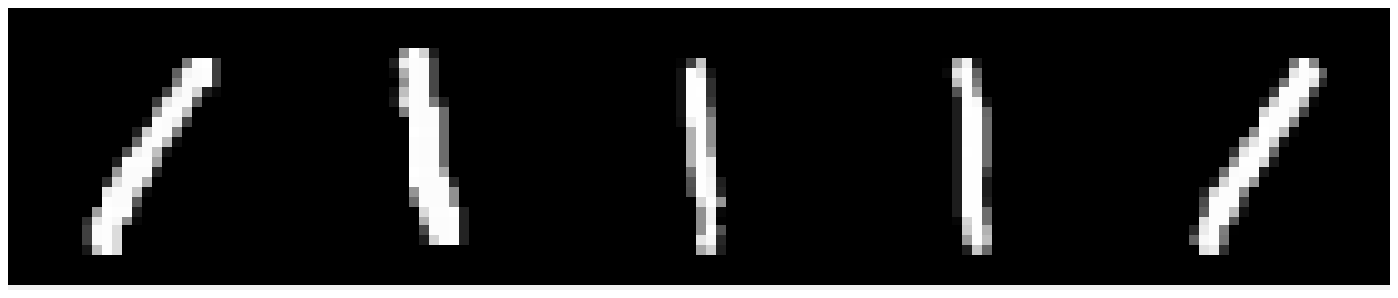}  & 
\cincludegraphics[scale=0.05]{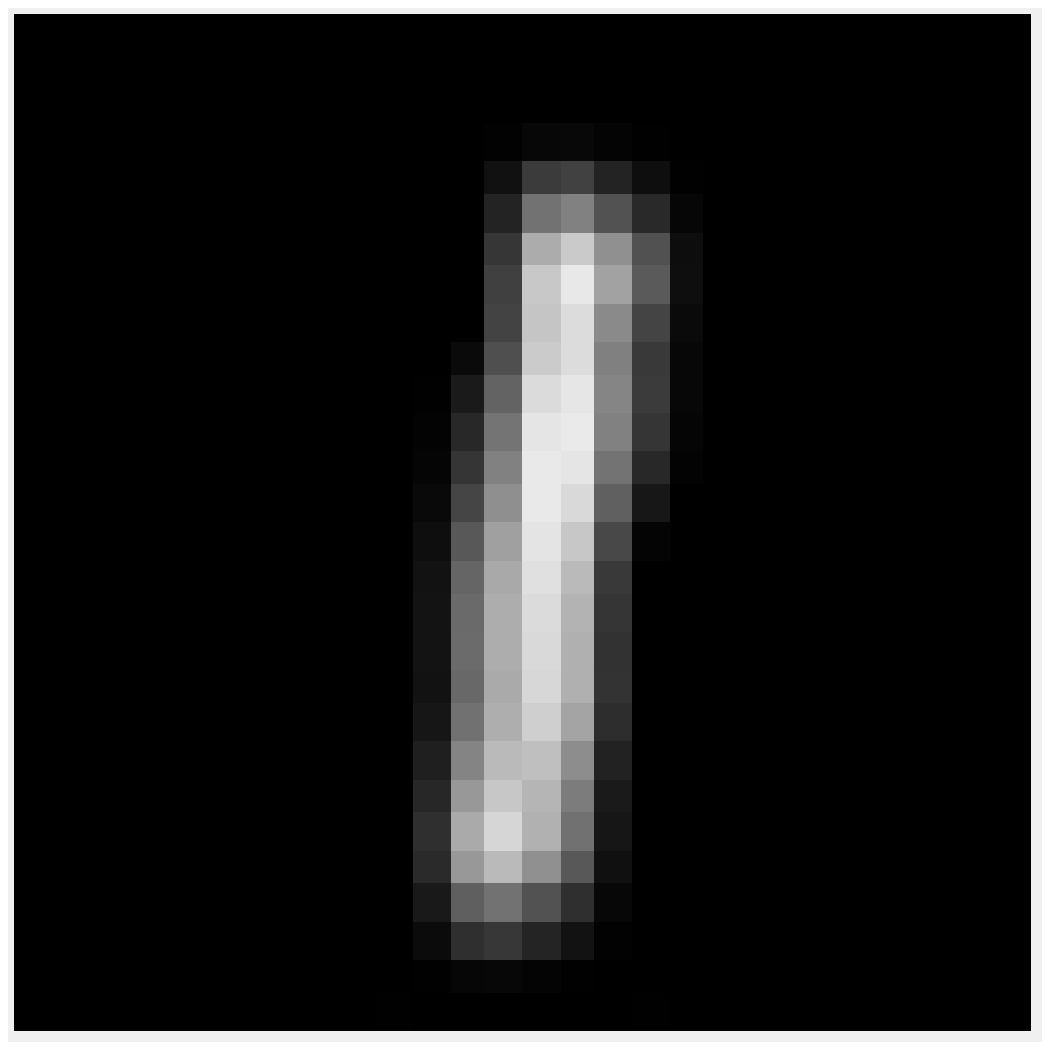}  &
\cincludegraphics[scale=0.05]{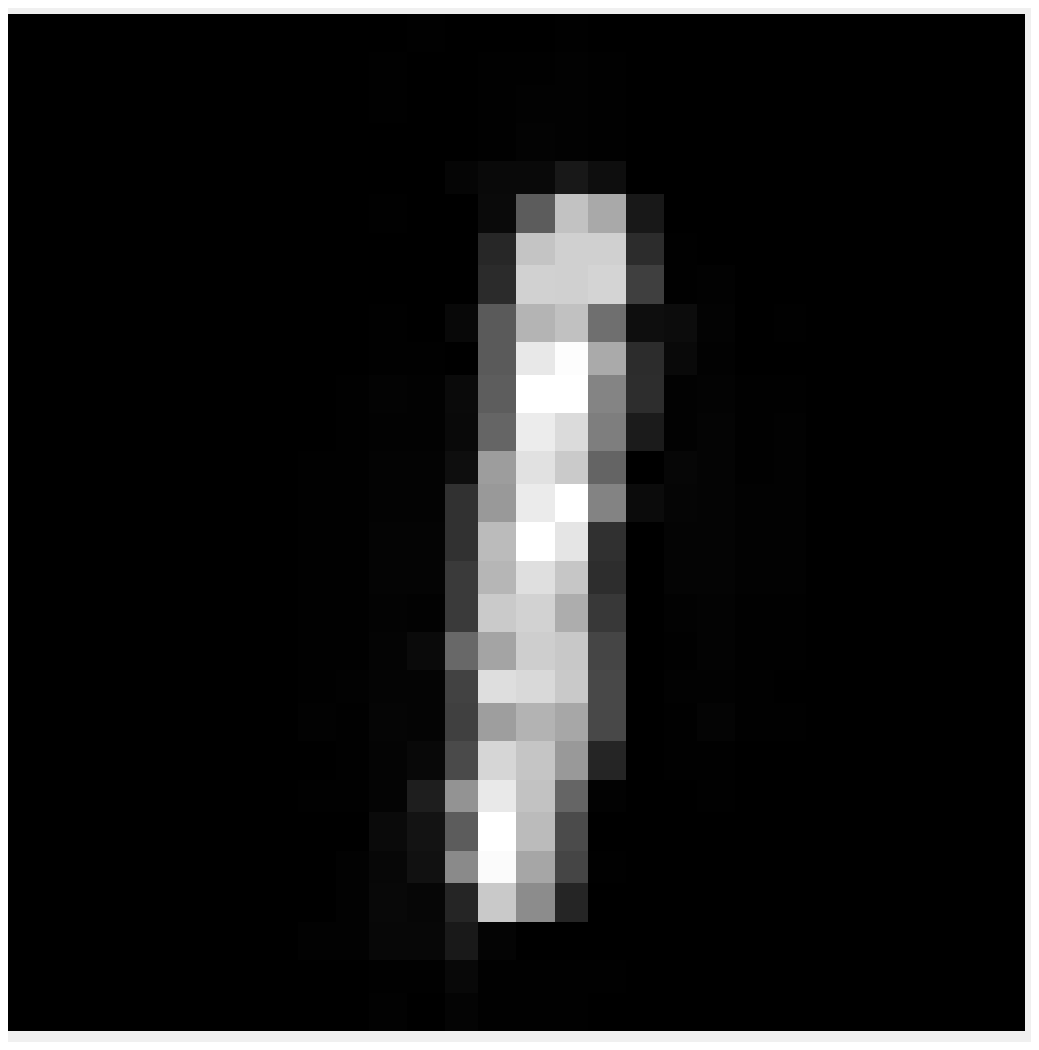} & 
\cincludegraphics[scale=0.05]{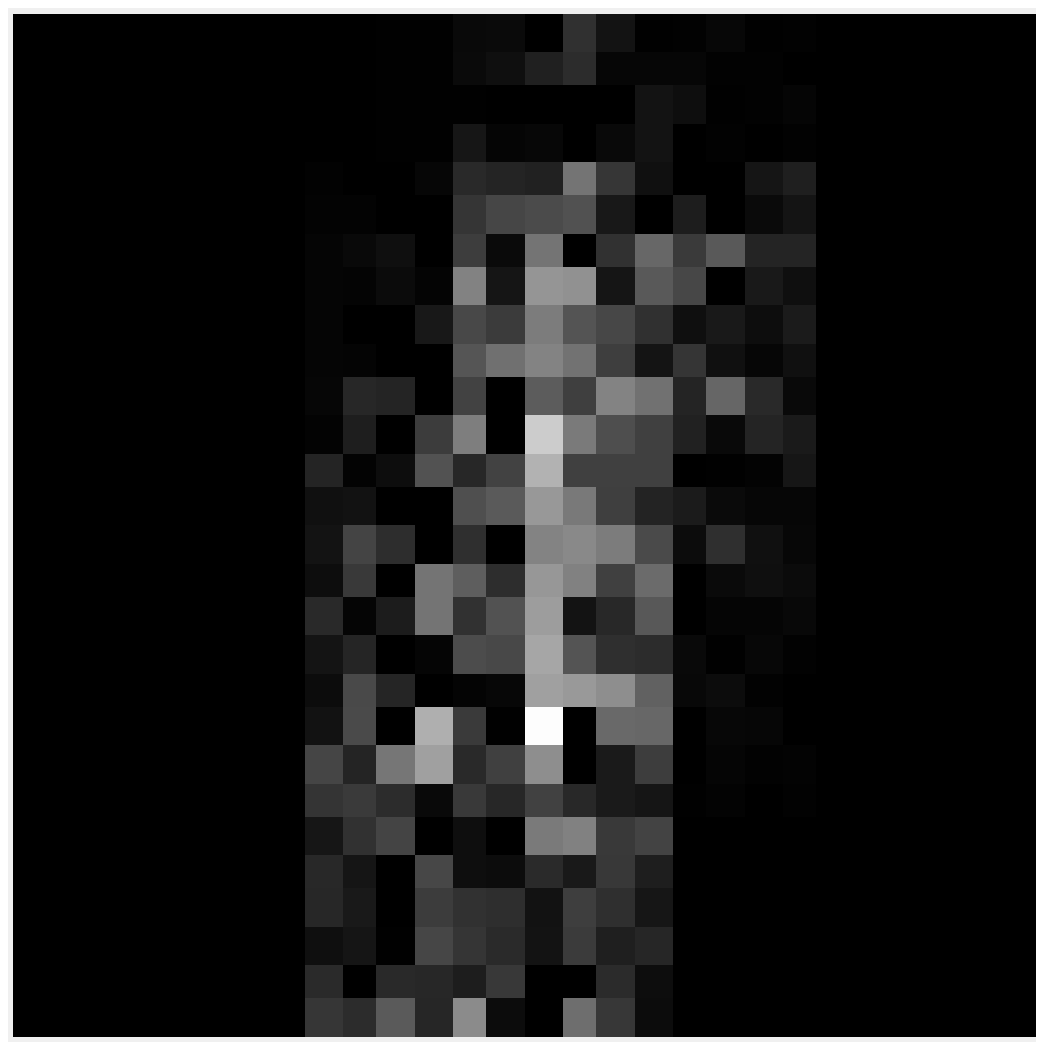} & 
\cincludegraphics[scale=0.05]{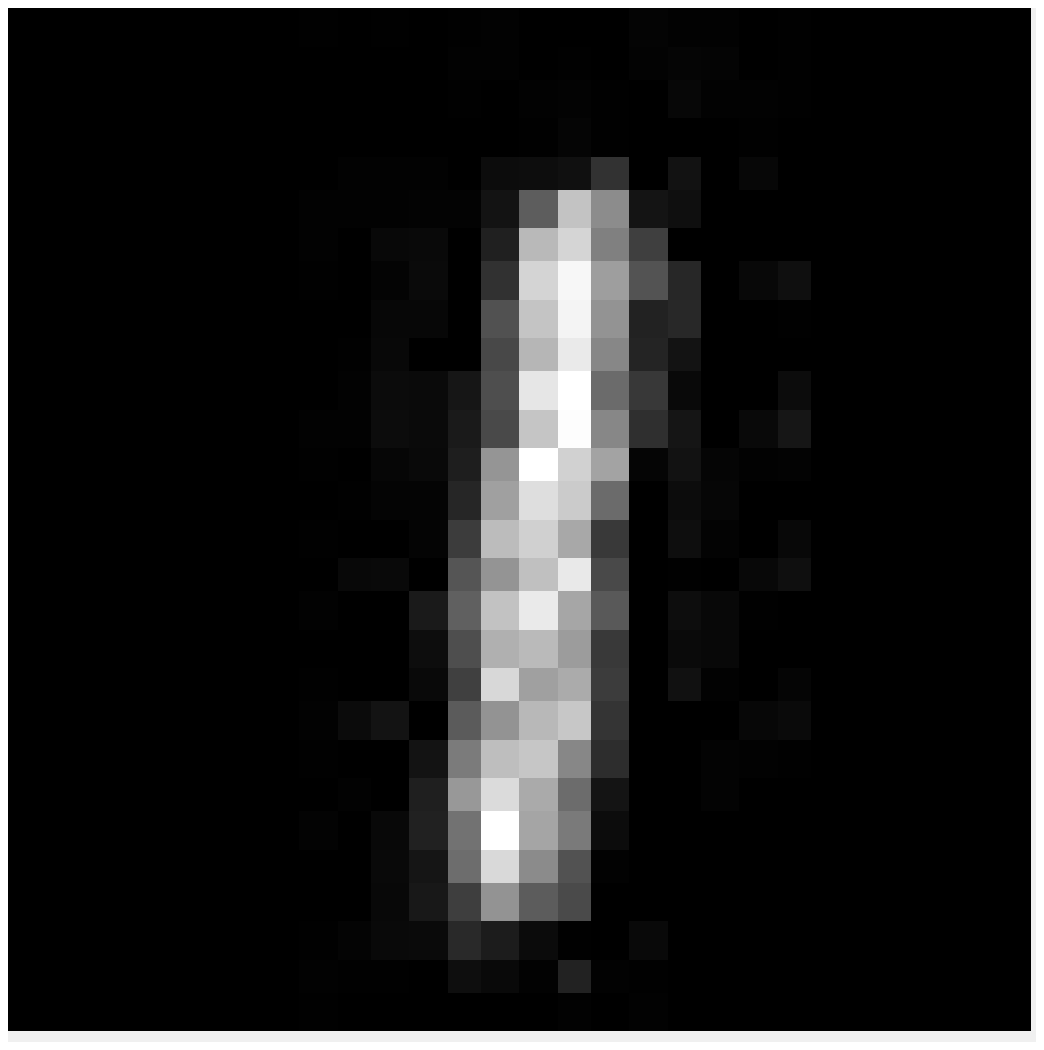} & 
\cincludegraphics[scale=0.05]{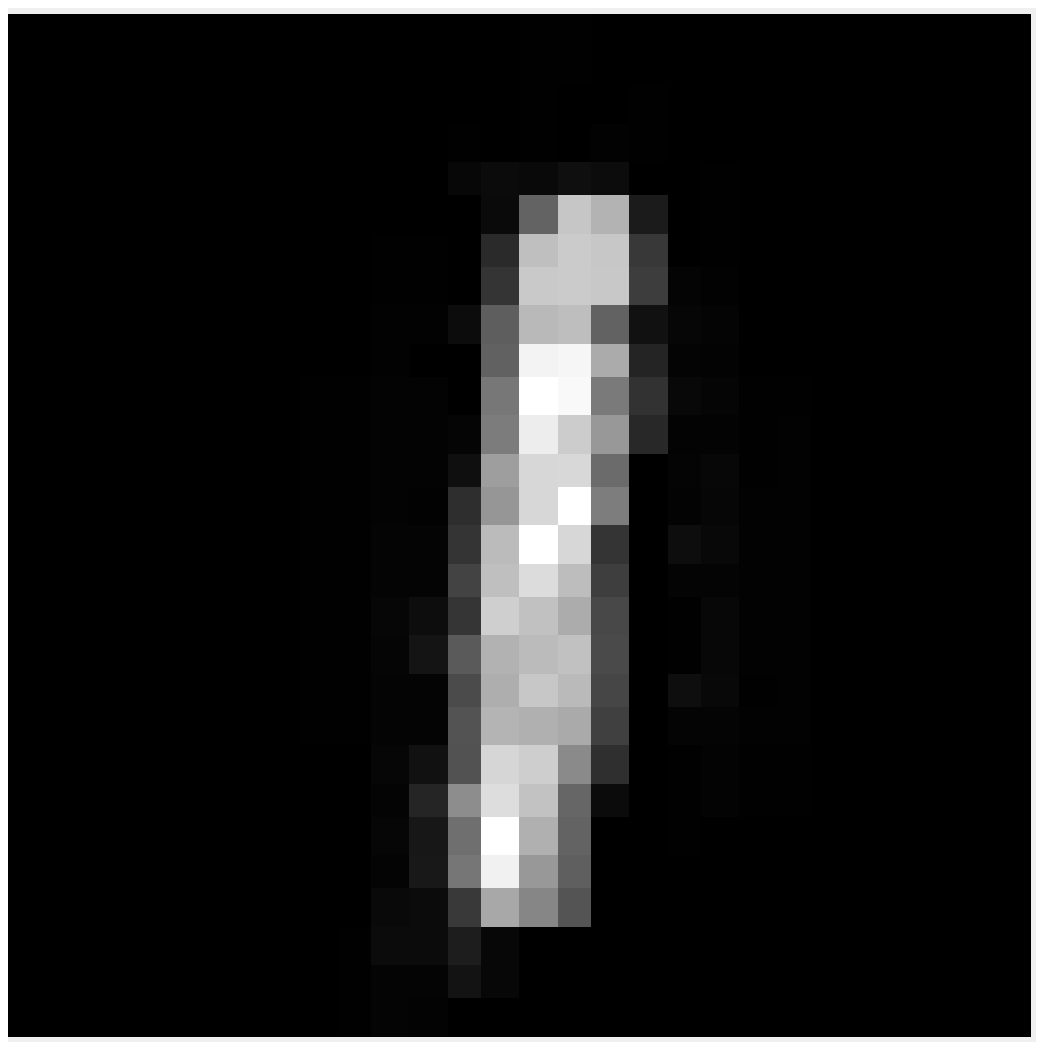} & 
\cincludegraphics[scale=0.05]{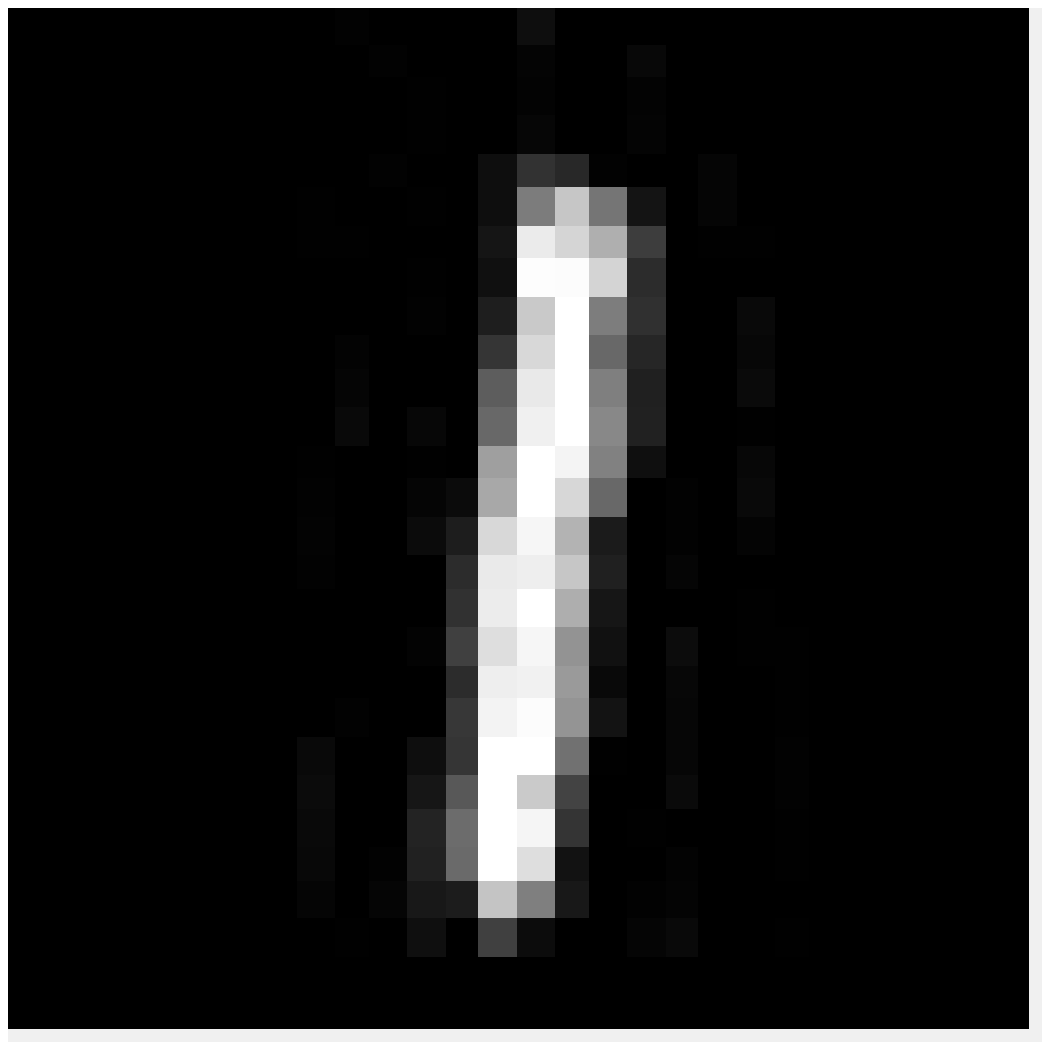} & 
\cincludegraphics[scale=0.05]{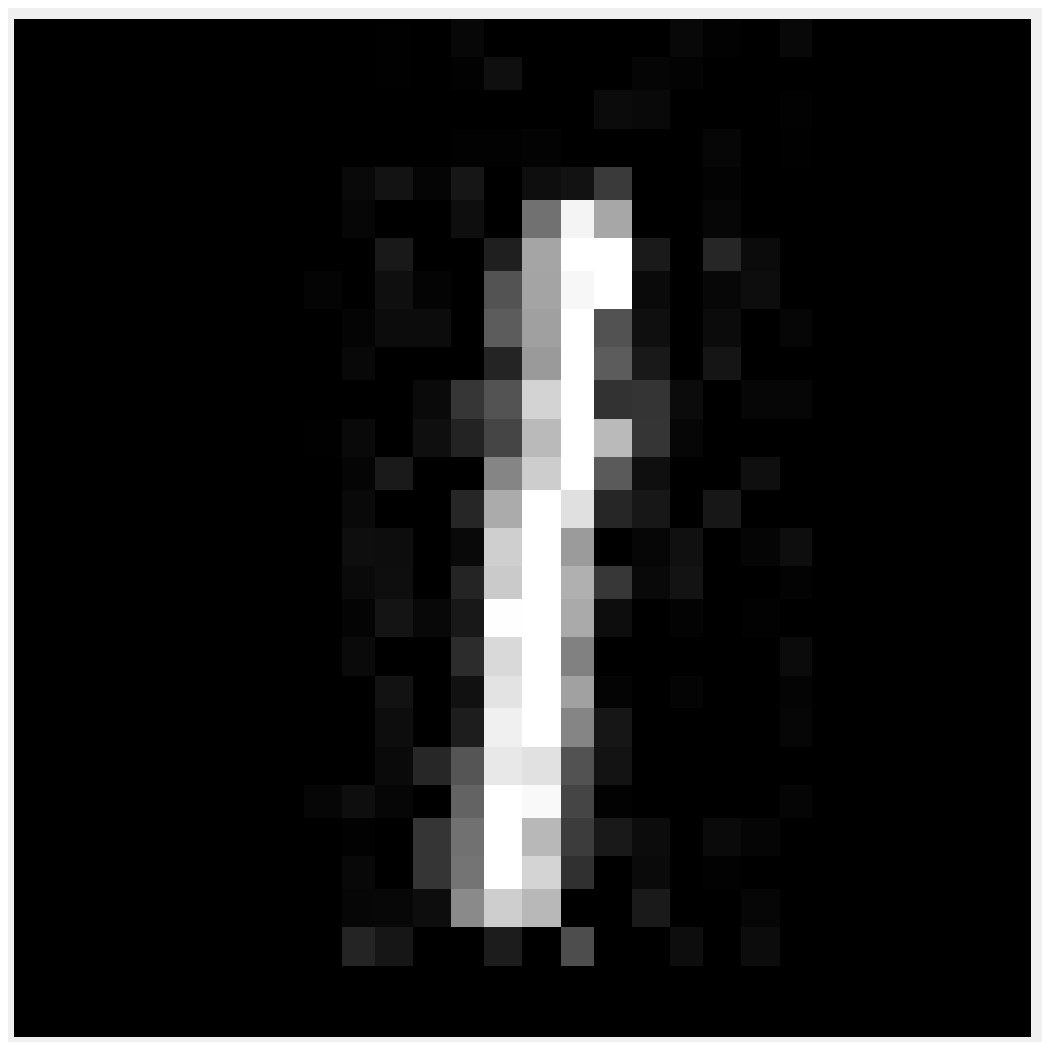} & 
\cincludegraphics[scale=0.05]{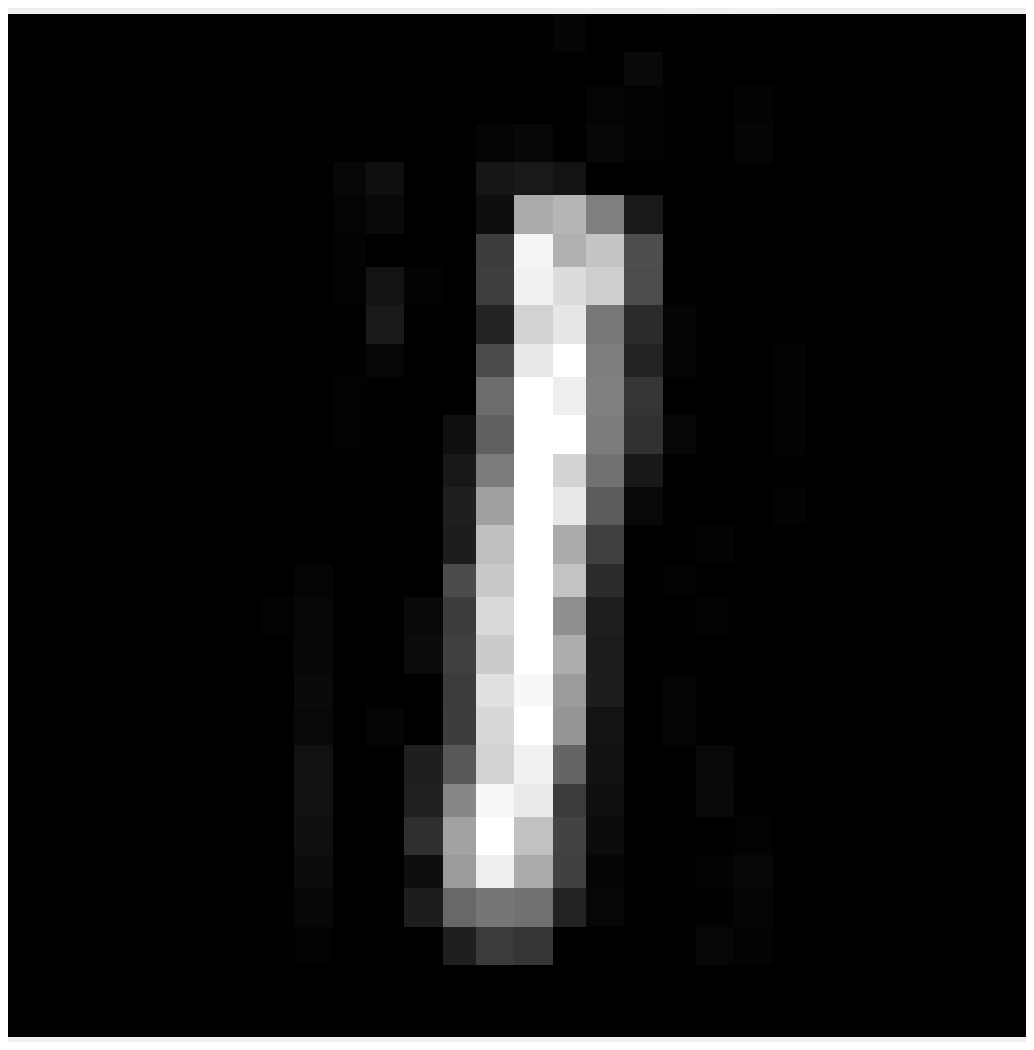} & 
\cincludegraphics[scale=0.05]{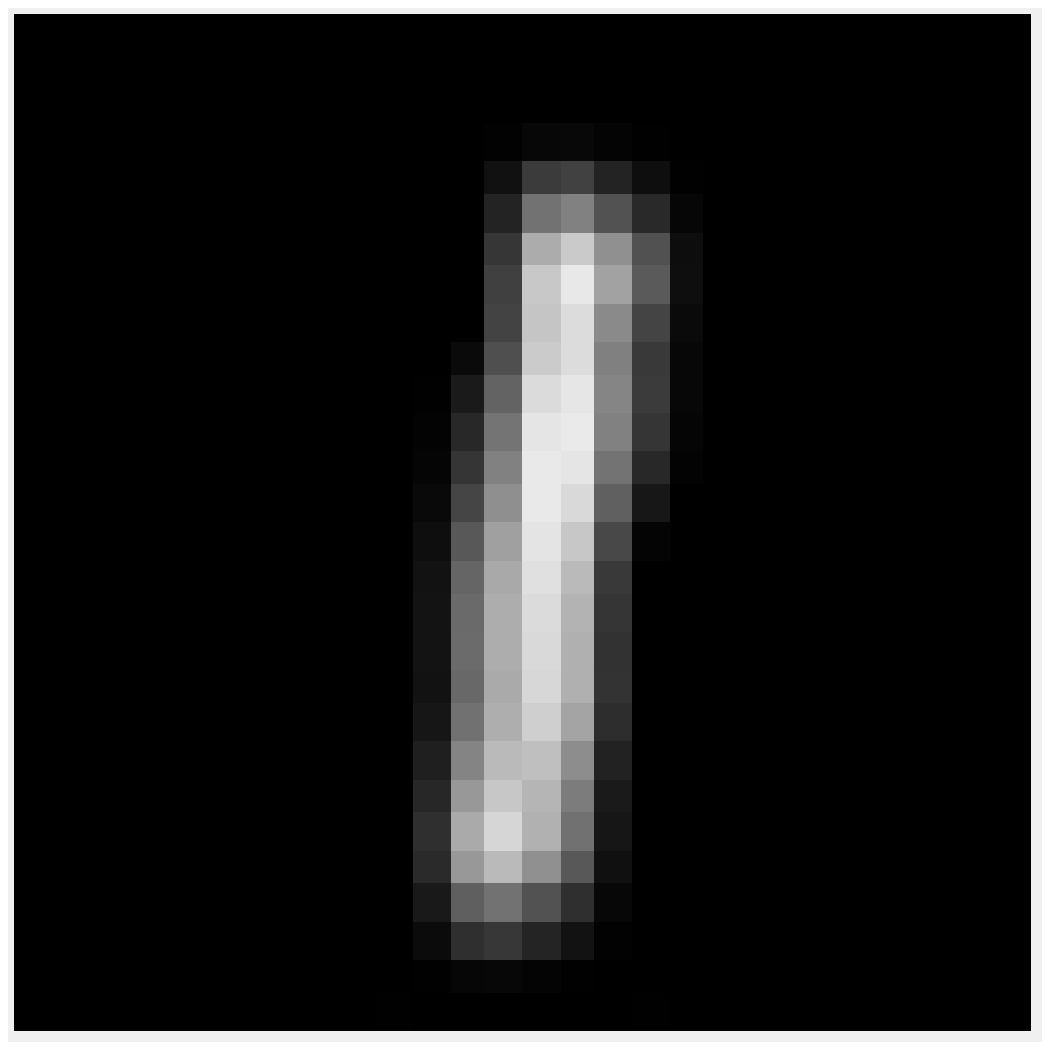}\\
 & 21.8881 & 24.0165 & 14.1423 & 22.1150 & 23.4254 
& 26.9159 & 22.4646 & 24.6290 & 21.8881\\

\hline
\end{tabular}
\par\end{centering}
\end{table*}

This subsection focuses on the performance comparison of dictionary learning algorithms when applied for single image super-resolution problem. The basic idea of this subsection is that with a relatively small number of training samples, our algorithm still can obtain good results even for real data. We first use MNIST as the data set, and follow two different algorithms, that is Yang et. al. \cite{yang2010image} and Zeyde et. al. \cite{zeyde2010single} as the approaches for super-resolution. The difficulty of this problem is that we can only extract few training samples from one digit. Then we apply the later one to natural images. Here we only use three images as training samples, which is far from enough for the most of the super-resolution problems. 

For both approaches by Yang et. al. \cite{yang2010image} and Zeyde et. al. \cite{zeyde2010single}, the basic idea is that given pairs of low and high resolution images as training data, a pair of dictionaries are learned so that sparse approximations of each pair of low/high resolution images share the same coefficients. For a test image of low resolution, one first finds its sparse representation under the low-resolution dictionary, and then apply the corresponding sparse coefficients to the high-resolution dictionary to generate a high resolution image. We usually extract patches with overlap from images to obtain more samples. The difference between Yang et. al. \cite{yang2010image} and Zeyde et. al. \cite{zeyde2010single} is mainly at the step of training low and high resolution dictionaries. By Yang et. al. \cite{yang2010image}, we stack a pair of low and high resolution training images vertically to obtain the sample data, and learn both low and high resolution dictionaries at the same time. While for Zeyde et. al.'s approach \cite{zeyde2010single}, we look at the low and high resolution training images separately. We first scale-up the low-resolution images by interpolation, returning to the size of the high-resolution training images, and learn the low-resolution dictionary and corresponding sparse coefficients. We then directly compute the high-resolution dictionary using the obtained sparse coefficients and high-resolution images by least-squares method.   

Our first simulation is based on MNIST dataset which contains images for handwritten digits from 0 to 9. Each image is of $28\times 28$ pixels. We generate low-resolution images of size $14 \times 14$ by bicubic interpolation, that is the output pixel value is a weighted average of pixels in the nearest 4-by-4 neighborhood. 

The training data used for dictionary learning is patch based. For the method by Yang et. al \cite{yang2010image}, in practice, it is suggested to extract features from upsampled version of the low-resolution to obtain better results \cite[\uppercase\expandafter{\romannumeral3}.C.]{yang2010image}. Thus, here we first scale-up the $14\times 14$ low-resolution images to the size of $28\times 28$ by bicubic interpolation, then patches of size $6\times 6$ are extracted from the low-resolution images with 2 pixels overlap in either direction for adjacent patches. We then find the corresponding patches of size $6\times 6$ at the same position from the high-resolution images. Each pair of low and high resolution patches is stacked to form a column in the training data, i.e.,  $\bm{Y}_{:,n}=\left[{\rm vec}(\bm{P}_{L})_n^{T},{\rm vec}(\bm{P}_{H})_n^{T}\right]^{T}$, where $\bm{P}_L$ and $\bm{P}_H$ are low/high resolution patches respectively. In the simulations, we use five different digit ones as training images and another digit one as testing image. Hence, the training sample matrix $\bm{Y}$ is of size $72 \times 245$. We then apply different algorithms for dictionary learning. As here we do not know the number of sparsity of sparse coefficients, we replace OMP with $\ell_1$ norm based sparse coding approach, that is LASSO for benchmark algorithms, and so do the following simulations. Denote the acquired dictionary by $\bm{D}=\left[\bm{D}_{L}^{T},\bm{D}_{H}^{T}\right]^{T}$, where $\bm{D}_{L}$ and $\bm{D}_{H}$ are the sub-dictionaries corresponding to low and high resolution patches respectively. Here we set $K=64$.

For Zeyde et. al.'s approach \cite{zeyde2010single}, we also upsample the low-resolution training image by factor of two. We then extract patches of size $6\times 6$ from the interpolated images with 2 pixels overlap, and the corresponding high-resolution patches of the same size at the same position. Here we only use the low-resolution patches as the training samples, and therefore the training sample matrix $\bm{Y}$ is of size $36\times 245$. We then apply different dictionary learning methods to obtain the low-resolution dictionary $\bm{D}_{L}$ of size $36\times 64$ and the sparse representations $\bm{X}_{0}$. The high-resolution dictionary $\bm{D}_{H}$ is hence computed by $\bm{D}_{H}=\bm{P}_{H}\bm{X}_{0}(\bm{X}_{0}\bm{X}_{0}^{T})^{-1}$.  

 Same as the low-resolution training images, we also use the upsampled version of the testing image. Similarly, $6\times 6$ patches with overlap of 2 pixels are 
extracted for both Yang et. al's algorithm \cite{yang2010image} and Zeyde et. al.'s approach \cite{zeyde2010single}. For each patch, a sparse representation coefficient vector $\bm{\alpha}$ is obtained so that $\bm{P}_L \approx \bm{D}_L \bm{\alpha}$ using sparse coding technique. The corresponding high resolution patches are generated via $\bm{D}_H \bm{\alpha}$ and the high resolution image is generated by aligning the patches and taking average of overlapped pixels across patches. 

The simulation results are presented in Table \ref{tab:1}. In numerical comparison, peak signal-to-noise ratio (PSNR) is used as the performance criterion, which is formulated as
\begin{equation}
    {\rm PSNR}=10{\rm log}_{10}\frac{N_{e}}{\Vert\hat{\bm{I}}-\bm{I}^0\Vert_{F}^{2}},
\end{equation}
where $\bm{I}^0$ and $\hat{\bm{I}}$ are the `ground-truth' high-resolution image and a high-resolution image generated using the learned dictionary respectively, and $N_{e}$ denotes the number of entries in $\bm{I}^{0}$. Simulation results demonstrate the significant improvement of ROAD in both the numerical error and the visual effect. By Yang et. al.'s method \cite{yang2010image}, only our algorithm can obtain higher resolution after super-resolution process. By Zeyde et. al.'s approach \cite{zeyde2010single}, even though the performance of other benchmark algorithms is improved, promotion of the resolution is not obvious than our method.

\begin{table*}
\caption{\label{tab:2}Comparison of single image super-resolution using different
dictionary learning methods, where both the figures of super-resolution
results and the PSNR between the estimated high-resolution images
and the ground truth image are shown in the table.}
\renewcommand{\arraystretch}{1.8}
\newcommand\cincludegraphics[2][]{\raisebox{-0.3\height}{\includegraphics[#1]{#2}}}
\begin{centering}
\begin{tabular}{C{3.1cm} C{3.1cm} C{3.1cm} C{3.1cm} C{3.1cm}}

Bicubic interpolation & ROAD & Blotless & K-SVD & MOD\\

\cincludegraphics[scale=0.2]{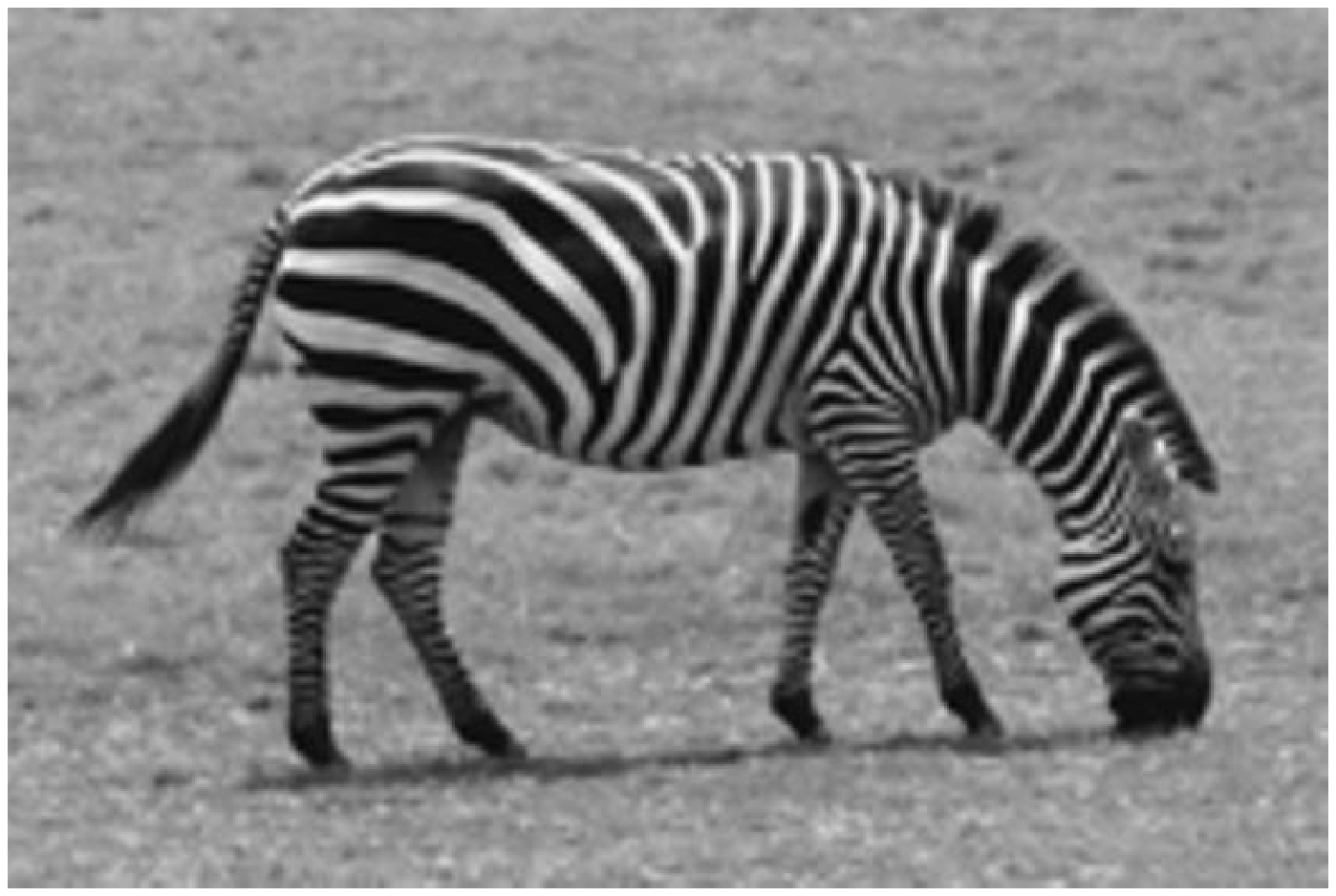}  & 
\cincludegraphics[scale=0.2]{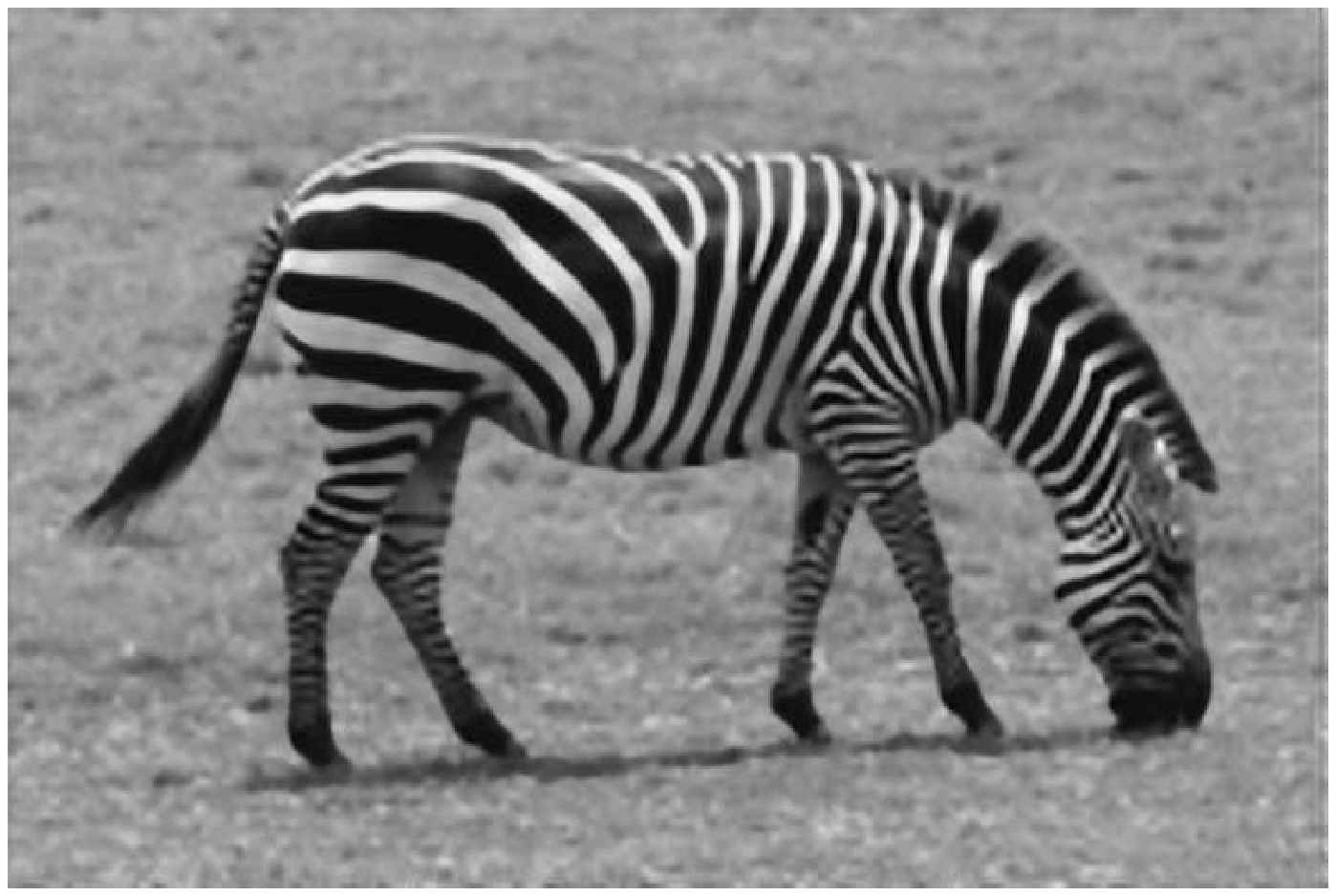}  &
\cincludegraphics[scale=0.2]{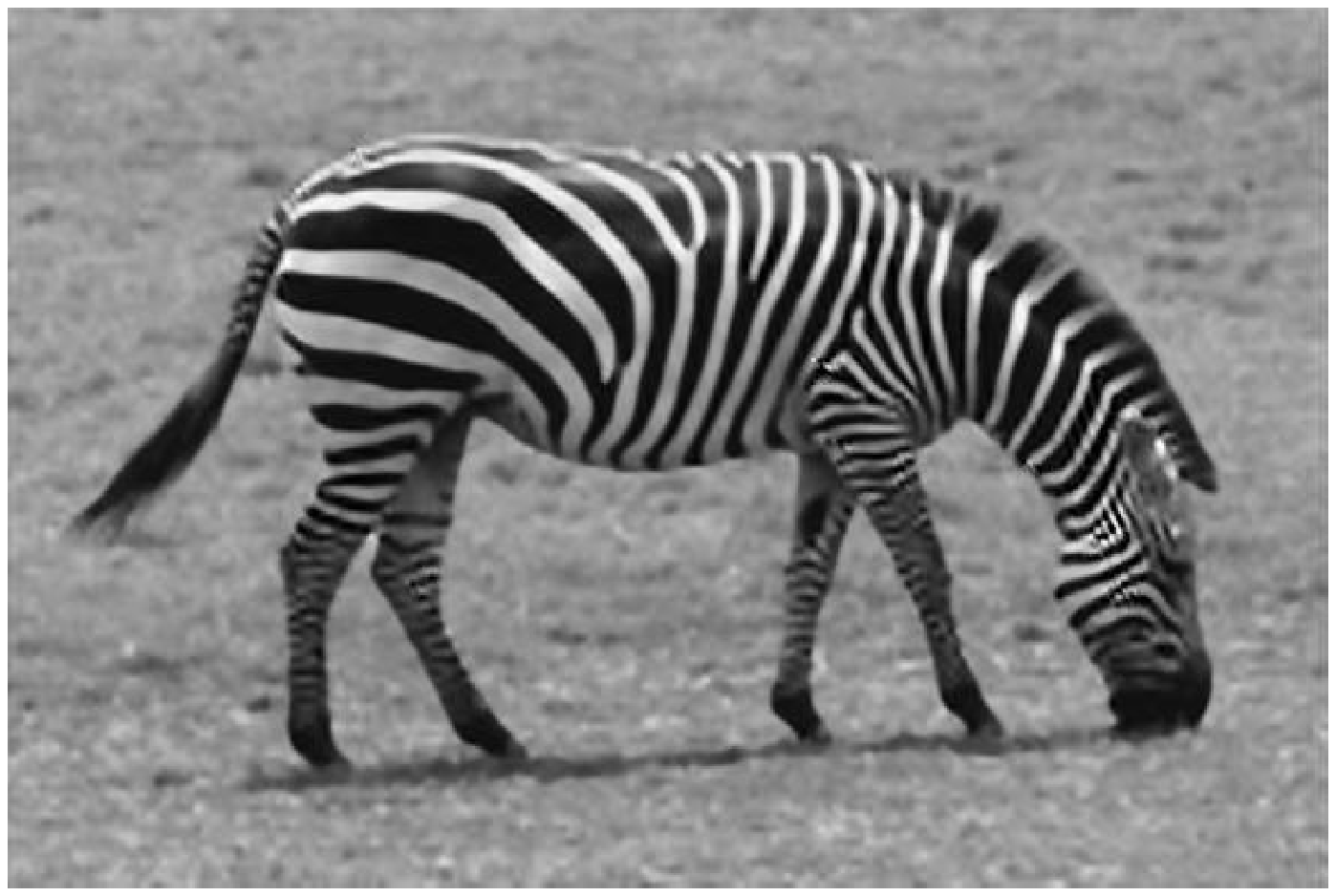} & 
\cincludegraphics[scale=0.2]{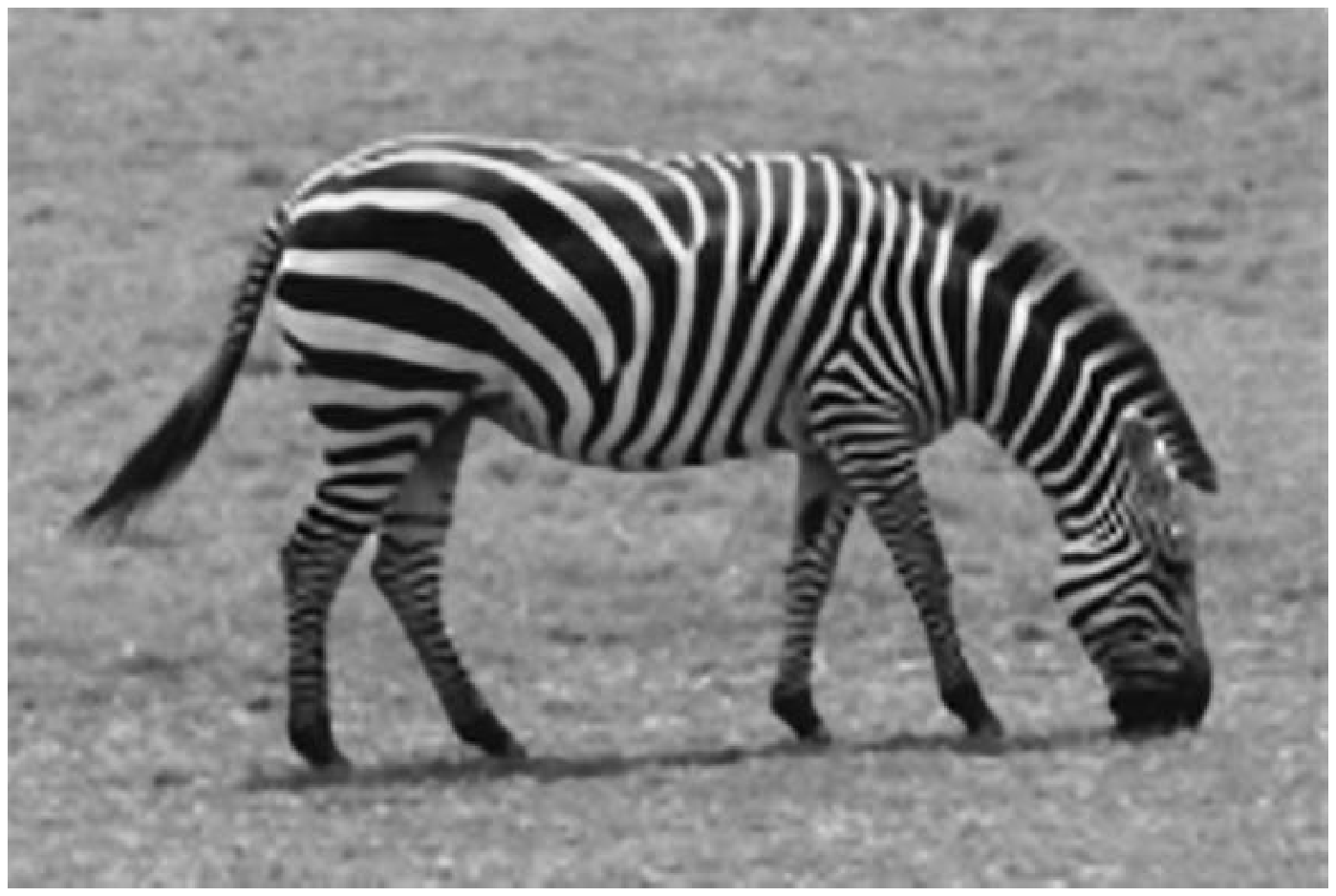} & 
\cincludegraphics[scale=0.2]{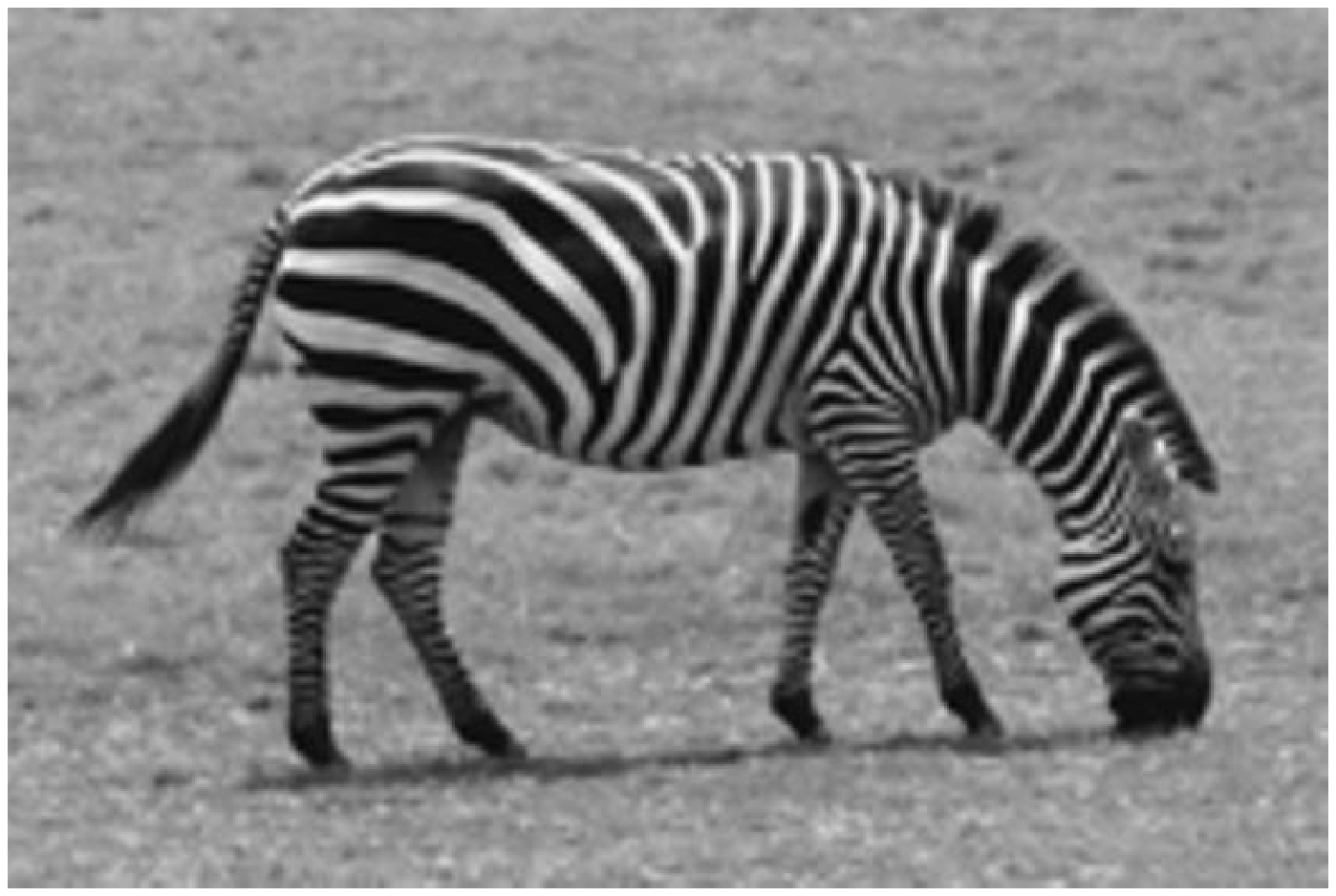}\\

26.7055 & 28.1758 & 27.7594 & 27.8736 & 26.7055 \\
\end{tabular}
\par\end{centering}
\end{table*}

\begin{figure}
\begin{centering}
\subfloat{\begin{centering}
\includegraphics[scale=0.15]{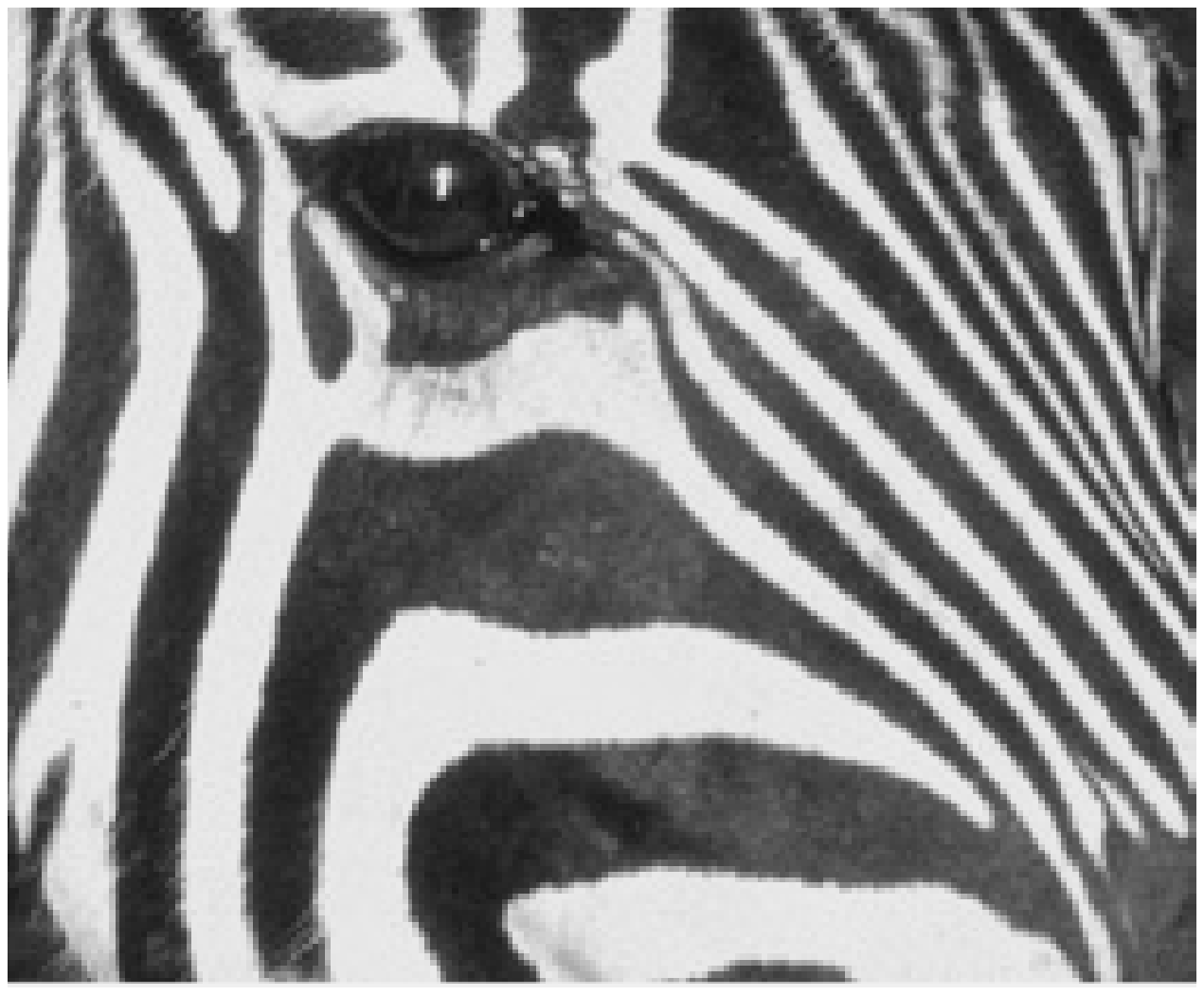}
\par\end{centering}
}\subfloat{\begin{centering}
\includegraphics[scale=0.15]{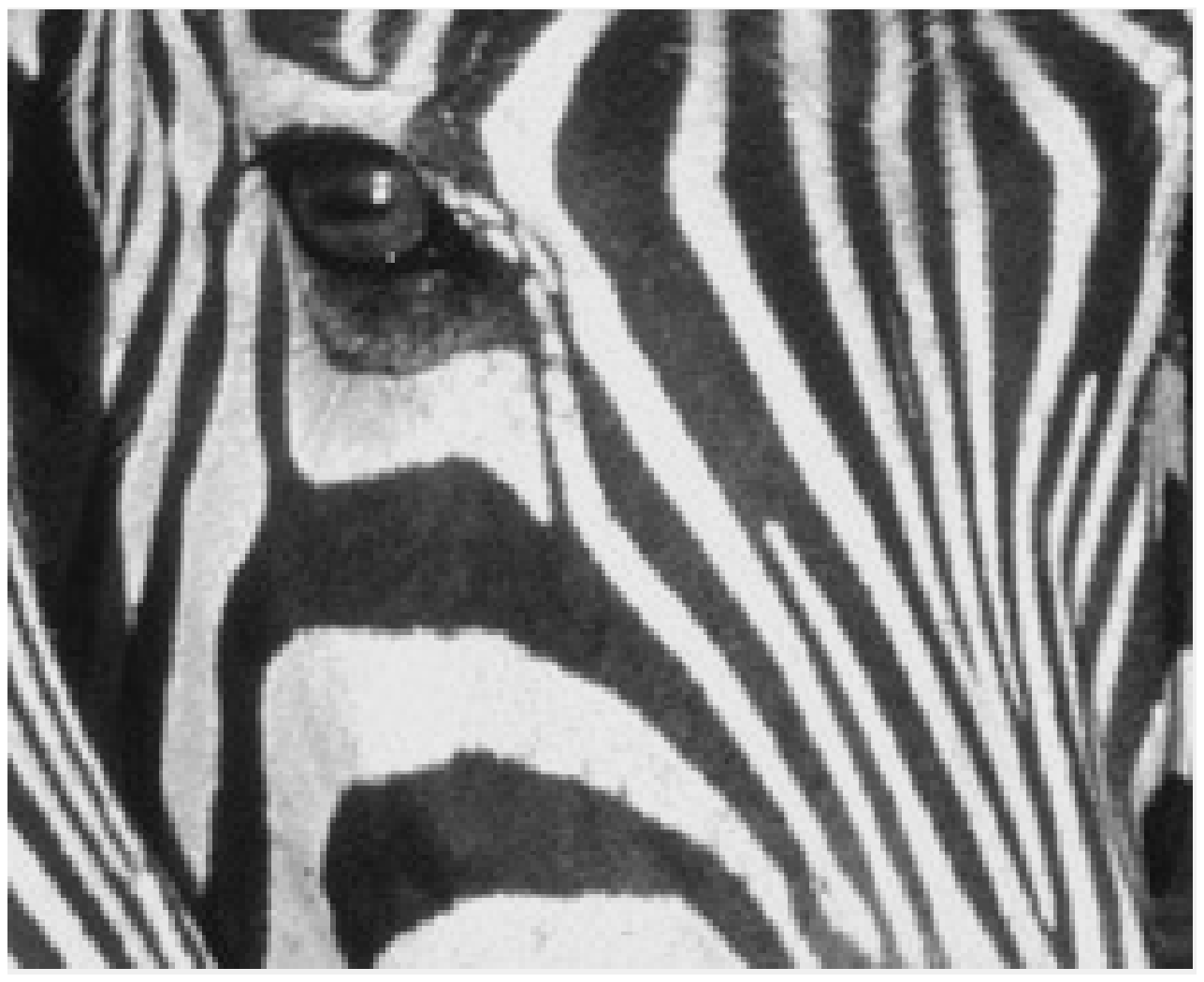}
\par\end{centering}
}
\subfloat{\begin{centering}
\includegraphics[scale=0.15]{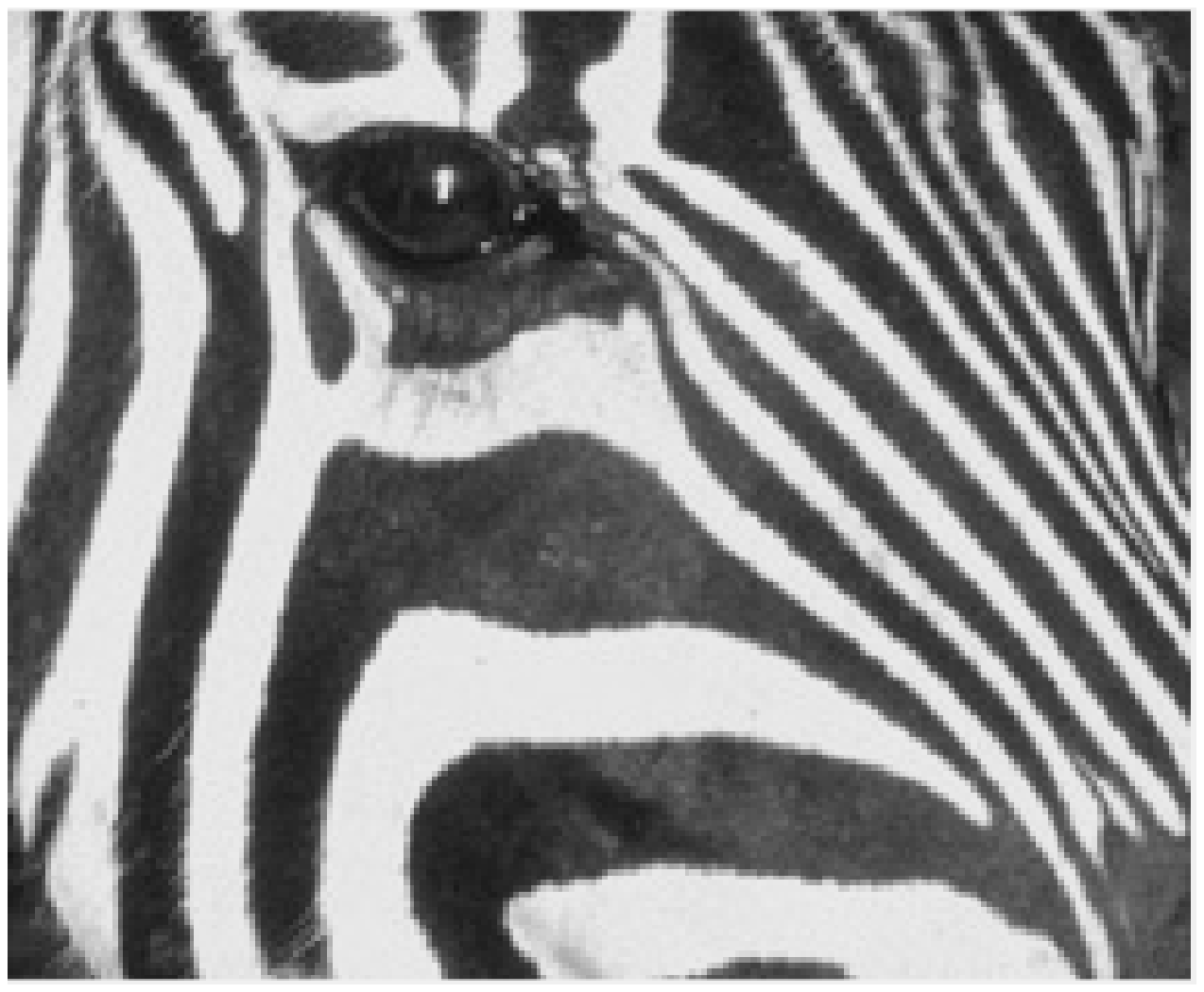}
\par\end{centering}
}
\par\end{centering}
\caption{\label{fig:4}Training images for natural image super-resolution.}
\end{figure}

\begin{figure}
    \centering
    \includegraphics[scale=0.23]{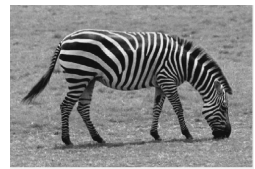}
    \caption{Testing image}
    \label{fig:5}
\end{figure}

According to the simulation results of MNIST, we only use Zeyde et. al.'s approach \cite{zeyde2010single} for the last simulation. The training samples are extracted from the images shown in Fig. \ref{fig:4} which are three images of zebra's head, and the testing image is the figure of a zebra shown in Fig. \ref{fig:5}. Here we set the scale-down ratio to $1/3$ for both training images and testing image. To acquire low resolution training samples, we first scale-down the original training images by $3$, and scale-up them to original size by bicubic interpolation. We then collect $9\times 9$ patches with 6 pixels overlap in either direction for adjacent patches. Instead of directly using raw patches as training samples, we use four 1-D filters to extract the edge information, which are
\begin{align}
    & \bm{f}_{1}=[1,\, 0,\, 0,\, -1],\; & \bm{f}_{2}=\bm{f}_{1}^{T},\nonumber\\
    & \bm{f}_{3}=[1,\, 0,\, 0,\, -2,\, 0,\, 0,\, 1]/2,\; & \bm{f}_{4}=\bm{f}_{3}^{T}. \nonumber
\end{align}
To further dig out the key features, we add PCA after the filtering. After the procedure of preprocessing, we obtain the training data $\bm{Y}$ of size $34\times 11775$. Here we set the dictionary size as $34\times 512$, and train low-resolution dictionary $\bm{D}_{L}$ and sparse coefficients $\bm{X}_{0}$ by different dictionary learning methods. For the high-resolution training images, we first subtract the interpolated images from them to remove low frequencies. Similar to the preprocessing procedure to the low-resolution images, we collect the high-resolution patches of the same size as the low-resolution ones, and we also use the filters to extract derivatives of high-resolution patches. While, instead of adding PCA procedure, we here directly use filtered high-resolution patches along with the the acquired sparse coefficients $\bm{X}_{0}$ to compute the high-resolution dictionary $\bm{D}_{H}$.

Given the interpolation of the training image, we collect the interpolated patches of size $9\times 9$ with 6 pixels overlap. After filtering the patches using the same filters, we also utilize PCA to extract key features. We then do sparse coding to the features by using the trained low-resolution dictionary $\bm{D}_{L}$, producing the sparse coefficient matrix $\bm{X}$. The high-resolution patches are eventually calculated by $\bm{P}_{H}=\bm{D}_{H}\bm{X}+\bm{P}_{M}$, where $\bm{P}_{M}$ is the matrix of interpolated patches. 

The results are presented in Table \ref{tab:2}, where the numbers under the images are the PSNR values comparing the estimated high-resolution images using different dictionary learning methods with the testing image. Table \ref{tab:2} reveals that ROAD outperforms other benchmark algorithms. Although we only use three images as training data, ROAD can still provide resolution promotion of the testing image in both PSNR and the visual effect.

\section{Conclusion}\label{sec:conclusion}
In this paper, we propose a novel dictionary learning algorithm
using rank-one atomic decomposition (ROAD), where the problem is cast as an optimization w.r.t. a single variable. Practically ROAD reduces
the number of tuning parameters required in other benchmark
algorithms. Two ADMM solvers including iADMM and eADMM are adopted to address the optimization problem. We prove the global convergence of iADMM and also show the better performance of eADMM in numerical tests. We compare ROAD with other benchmark algorithms by using both synthetic data and real data. The test results demonstrate that
ROAD outperforms other benchmark algorithms in all the scenarios.

\clearpage

\bibliography{ICASSP2019ref}

\begin{thebibliography}{10}
\providecommand{\url}[1]{#1}
\csname url@samestyle\endcsname
\providecommand{\newblock}{\relax}
\providecommand{\bibinfo}[2]{#2}
\providecommand{\BIBentrySTDinterwordspacing}{\spaceskip=0pt\relax}
\providecommand{\BIBentryALTinterwordstretchfactor}{4}
\providecommand{\BIBentryALTinterwordspacing}{\spaceskip=\fontdimen2\font plus
\BIBentryALTinterwordstretchfactor\fontdimen3\font minus
  \fontdimen4\font\relax}
\providecommand{\BIBforeignlanguage}[2]{{%
\expandafter\ifx\csname l@#1\endcsname\relax
\typeout{** WARNING: IEEEtran.bst: No hyphenation pattern has been}%
\typeout{** loaded for the language `#1'. Using the pattern for}%
\typeout{** the default language instead.}%
\else
\language=\csname l@#1\endcsname
\fi
#2}}
\providecommand{\BIBdecl}{\relax}
\BIBdecl

\bibitem{elad2006image}
M.~Elad and M.~Aharon, ``Image denoising via sparse and redundant
  representations over learned dictionaries,'' \emph{IEEE Transactions on Image
  processing}, vol.~15, no.~12, pp. 3736--3745, 2006.

\bibitem{dabov2007image}
K.~Dabov, A.~Foi, V.~Katkovnik, and K.~Egiazarian, ``Image denoising by sparse
  3-d transform-domain collaborative filtering,'' \emph{IEEE Transactions on
  image processing}, vol.~16, no.~8, pp. 2080--2095, 2007.

\bibitem{mairal2008sparse}
J.~Mairal, M.~Elad, and G.~Sapiro, ``Sparse representation for color image
  restoration,'' \emph{IEEE Transactions on image processing}, vol.~17, no.~1,
  pp. 53--69, 2008.

\bibitem{dong2013nonlocally}
W.~Dong, L.~Zhang, G.~Shi, and X.~Li, ``Nonlocally centralized sparse
  representation for image restoration,'' \emph{IEEE Transactions on Image
  Processing}, vol.~22, no.~4, pp. 1620--1630, 2013.

\bibitem{li2006underdetermined}
Y.~Li, S.-I. Amari, A.~Cichocki, D.~W. Ho, and S.~Xie, ``Underdetermined blind
  source separation based on sparse representation,'' \emph{IEEE Transactions
  on signal processing}, vol.~54, no.~2, pp. 423--437, 2006.

\bibitem{abolghasemi2012blind}
V.~Abolghasemi, S.~Ferdowsi, and S.~Sanei, ``Blind separation of image sources
  via adaptive dictionary learning,'' \emph{IEEE Transactions on Image
  Processing}, vol.~21, no.~6, pp. 2921--2930, 2012.

\bibitem{tosic2011dictionary}
I.~Tosic and P.~Frossard, ``Dictionary learning,'' \emph{IEEE Signal Processing
  Magazine}, vol.~28, no.~2, pp. 27--38, 2011.

\bibitem{huang2007sparse}
K.~Huang and S.~Aviyente, ``Sparse representation for signal classification,''
  in \emph{Advances in neural information processing systems}, 2007, pp.
  609--616.

\bibitem{wright2009robust}
J.~Wright, A.~Y. Yang, A.~Ganesh, S.~S. Sastry, and Y.~Ma, ``Robust face
  recognition via sparse representation,'' \emph{IEEE transactions on pattern
  analysis and machine intelligence}, vol.~31, no.~2, pp. 210--227, 2009.

\bibitem{wright2010sparse}
J.~Wright, Y.~Ma, J.~Mairal, G.~Sapiro, T.~S. Huang, and S.~Yan, ``Sparse
  representation for computer vision and pattern recognition,''
  \emph{Proceedings of the IEEE}, vol.~98, no.~6, pp. 1031--1044, 2010.

\bibitem{zhang2011sparse}
L.~Zhang, M.~Yang, and X.~Feng, ``Sparse representation or collaborative
  representation: Which helps face recognition,'' in \emph{Computer vision
  (ICCV), 2011 IEEE international conference on}.\hskip 1em plus 0.5em minus
  0.4em\relax IEEE, 2011, pp. 471--478.

\bibitem{yang2010image}
J.~Yang, J.~Wright, T.~S. Huang, and Y.~Ma, ``Image super-resolution via sparse
  representation,'' \emph{IEEE transactions on image processing}, vol.~19,
  no.~11, pp. 2861--2873, 2010.

\bibitem{dong2011image}
W.~Dong, L.~Zhang, G.~Shi, and X.~Wu, ``Image deblurring and super-resolution
  by adaptive sparse domain selection and adaptive regularization,'' \emph{IEEE
  Transactions on Image Processing}, vol.~20, no.~7, pp. 1838--1857, 2011.

\bibitem{ahmed1974discrete}
N.~Ahmed, T.~Natarajan, and K.~R. Rao, ``Discrete cosine transform,''
  \emph{IEEE transactions on Computers}, vol. 100, no.~1, pp. 90--93, 1974.

\bibitem{allen1977unified}
J.~B. Allen and L.~R. Rabiner, ``A unified approach to short-time fourier
  analysis and synthesis,'' \emph{Proceedings of the IEEE}, vol.~65, no.~11,
  pp. 1558--1564, 1977.

\bibitem{stephane1999wavelet}
M.~Stephane, ``A wavelet tour of signal processing,'' \emph{The Sparse Way},
  1999.

\bibitem{candes2000curvelets}
E.~J. Candes and D.~L. Donoho, ``Curvelets: A surprisingly effective
  nonadaptive representation for objects with edges,'' Stanford Univ Ca Dept of
  Statistics, Tech. Rep., 2000.

\bibitem{olshausen1996emergence}
B.~A. Olshausen and D.~J. Field, ``Emergence of simple-cell receptive field
  properties by learning a sparse code for natural images,'' \emph{Nature},
  vol. 381, no. 6583, p. 607, 1996.

\bibitem{engan1999method}
K.~Engan, S.~O. Aase, and J.~H. Husoy, ``Method of optimal directions for frame
  design,'' in \emph{Acoustics, Speech, and Signal Processing, 1999.
  Proceedings., 1999 IEEE International Conference on}, vol.~5.\hskip 1em plus
  0.5em minus 0.4em\relax IEEE, 1999, pp. 2443--2446.

\bibitem{aharon2006k}
M.~Aharon, M.~Elad, A.~Bruckstein \emph{et~al.}, ``K-svd: An algorithm for
  designing overcomplete dictionaries for sparse representation,'' \emph{IEEE
  Transactions on signal processing}, vol.~54, no.~11, p. 4311, 2006.

\bibitem{engan2007family}
K.~Engan, K.~Skretting, and J.~H. Hus{\o}y, ``Family of iterative ls-based
  dictionary learning algorithms, ils-dla, for sparse signal representation,''
  \emph{Digital Signal Processing}, vol.~17, no.~1, pp. 32--49, 2007.

\bibitem{skretting2010recursive}
K.~Skretting and K.~Engan, ``Recursive least squares dictionary learning
  algorithm,'' \emph{IEEE Transactions on Signal Processing}, vol.~58, no.~4,
  pp. 2121--2130, 2010.

\bibitem{dai2012simultaneous}
W.~Dai, T.~Xu, and W.~Wang, ``Simultaneous codeword optimization (simco) for
  dictionary update and learning,'' \emph{IEEE Transactions on Signal
  Processing}, vol.~60, no.~12, pp. 6340--6353, 2012.

\bibitem{yu2019bilinear}
Q.~Yu, W.~Dai, Z.~Cvetkovic, and J.~Zhu, ``Bilinear dictionary update via
  linear least squares,'' in \emph{ICASSP 2019-2019 IEEE International
  Conference on Acoustics, Speech and Signal Processing (ICASSP)}.\hskip 1em
  plus 0.5em minus 0.4em\relax IEEE, 2019, pp. 7923--7927.

\bibitem{mallat1993matching}
S.~Mallat and Z.~Zhang, ``Matching pursuit with time-frequency dictionaries,''
  Courant Institute of Mathematical Sciences New York United States, Tech.
  Rep., 1993.

\bibitem{pati1993orthogonal}
\emph{Orthogonal matching pursuit: Recursive function approximation with
  applications to wavelet decomposition}.\hskip 1em plus 0.5em minus
  0.4em\relax IEEE, 1993.

\bibitem{tropp2007signal}
J.~A. Tropp and A.~C. Gilbert, ``Signal recovery from random measurements via
  orthogonal matching pursuit,'' \emph{IEEE Transactions on information
  theory}, vol.~53, no.~12, pp. 4655--4666, 2007.

\bibitem{dai2009subspace}
W.~Dai and O.~Milenkovic, ``Subspace pursuit for compressive sensing signal
  reconstruction,'' \emph{IEEE transactions on Information Theory}, vol.~55,
  no.~5, pp. 2230--2249, 2009.

\bibitem{needell2009cosamp}
D.~Needell and J.~A. Tropp, ``Cosamp: Iterative signal recovery from incomplete
  and inaccurate samples,'' \emph{Applied and computational harmonic analysis},
  vol.~26, no.~3, pp. 301--321, 2009.

\bibitem{chen2001atomic}
S.~S. Chen, D.~L. Donoho, and M.~A. Saunders, ``Atomic decomposition by basis
  pursuit,'' \emph{SIAM review}, vol.~43, no.~1, pp. 129--159, 2001.

\bibitem{tibshirani1996regression}
R.~Tibshirani, ``Regression shrinkage and selection via the lasso,''
  \emph{Journal of the Royal Statistical Society. Series B (Methodological)},
  pp. 267--288, 1996.

\bibitem{daubechies2004iterative}
I.~Daubechies, M.~Defrise, and C.~De~Mol, ``An iterative thresholding algorithm
  for linear inverse problems with a sparsity constraint,''
  \emph{Communications on Pure and Applied Mathematics: A Journal Issued by the
  Courant Institute of Mathematical Sciences}, vol.~57, no.~11, pp. 1413--1457,
  2004.

\bibitem{hale2007fixed}
E.~T. Hale, W.~Yin, and Y.~Zhang, ``A fixed-point continuation method for
  l1-regularized minimization with applications to compressed sensing,''
  \emph{CAAM TR07-07, Rice University}, vol.~43, p.~44, 2007.

\bibitem{beck2009fast}
A.~Beck and M.~Teboulle, ``A fast iterative shrinkage-thresholding algorithm
  for linear inverse problems,'' \emph{SIAM journal on imaging sciences},
  vol.~2, no.~1, pp. 183--202, 2009.

\bibitem{chartrand2007exact}
R.~Chartrand, ``Exact reconstruction of sparse signals via nonconvex
  minimization,'' \emph{IEEE Signal Processing Letters}, vol.~14, no.~10, pp.
  707--710, 2007.

\bibitem{chartrand2008restricted}
R.~Chartrand and V.~Staneva, ``Restricted isometry properties and nonconvex
  compressive sensing,'' \emph{Inverse Problems}, vol.~24, no.~3, p. 035020,
  2008.

\bibitem{xu2012l_}
Z.~Xu, X.~Chang, F.~Xu, and H.~Zhang, ``$ l\_ $\{$1/2$\}$ $ regularization: A
  thresholding representation theory and a fast solver,'' \emph{IEEE
  Transactions on neural networks and learning systems}, vol.~23, no.~7, pp.
  1013--1027, 2012.

\bibitem{seghouane2015sequential}
A.-K. Seghouane and M.~Hanif, ``A sequential dictionary learning algorithm with
  enforced sparsity,'' in \emph{2015 IEEE International Conference on
  Acoustics, Speech and Signal Processing (ICASSP)}.\hskip 1em plus 0.5em minus
  0.4em\relax IEEE, 2015, pp. 3876--3880.

\bibitem{seghouane2018consistent}
A.-K. Seghouane and A.~Iqbal, ``Consistent adaptive sequential dictionary
  learning,'' \emph{Signal Processing}, vol. 153, pp. 300--310, 2018.

\bibitem{wang2019global}
Y.~Wang, W.~Yin, and J.~Zeng, ``Global convergence of admm in nonconvex
  nonsmooth optimization,'' \emph{Journal of Scientific Computing}, vol.~78,
  no.~1, pp. 29--63, 2019.

\bibitem{cheng2020dictionary}
C.~Cheng and W.~Dai, ``Dictionary learning using rank-one projection,'' in
  \emph{2020 28th European Signal Processing Conference (EUSIPCO)}.\hskip 1em
  plus 0.5em minus 0.4em\relax IEEE, pp. 2030--2034.

\bibitem{ling2018self}
S.~Ling and T.~Strohmer, ``Self-calibration and bilinear inverse problems via
  linear least squares,'' \emph{SIAM Journal on Imaging Sciences}, vol.~11,
  no.~1, pp. 252--292, 2018.

\bibitem{boyd2011distributed}
S.~Boyd, N.~Parikh, E.~Chu, B.~Peleato, J.~Eckstein \emph{et~al.},
  \emph{Distributed optimization and statistical learning via the alternating
  direction method of multipliers}.\hskip 1em plus 0.5em minus 0.4em\relax Now
  Publishers, Inc., 2011, vol.~3, no.~1.

\bibitem{nocedal2006conjugate}
J.~Nocedal and S.~J. Wright, ``Conjugate gradient methods,'' \emph{Numerical
  optimization}, pp. 101--134, 2006.

\bibitem{andersson2016operator}
F.~Andersson, M.~Carlsson, and K.-M. Perfekt, ``Operator-lipschitz estimates
  for the singular value functional calculus,'' \emph{Proceedings of the
  American Mathematical Society}, vol. 144, no.~5, pp. 1867--1875, 2016.

\bibitem{poliquin2000local}
R.~Poliquin, R.~Rockafellar, and L.~Thibault, ``Local differentiability of
  distance functions,'' \emph{Transactions of the American mathematical
  Society}, vol. 352, no.~11, pp. 5231--5249, 2000.

\bibitem{poliquin1996prox}
R.~Poliquin and R.~Rockafellar, ``Prox-regular functions in variational
  analysis,'' \emph{Transactions of the American Mathematical Society}, vol.
  348, no.~5, pp. 1805--1838, 1996.

\bibitem{zeyde2010single}
R.~Zeyde, M.~Elad, and M.~Protter, ``On single image scale-up using
  sparse-representations,'' in \emph{International conference on curves and
  surfaces}.\hskip 1em plus 0.5em minus 0.4em\relax Springer, 2010, pp.
  711--730.

\end{thebibliography}

\end{document}